\documentclass[11pt,draftclsnofoot, onecolumn]{IEEEtran}

\usepackage{graphicx,cite,epsfig,amssymb,amsmath,subfigure,url,stfloats,latexsym}
\usepackage{array}
\usepackage{arydshln}
\usepackage{amsfonts}
\usepackage{pgfplots}
\usepackage{algorithm}
\usepackage[noend]{algpseudocode}
\usepackage{epstopdf}
\usepackage{amsfonts,amsthm}
\usepackage{multirow}
\usepackage{mathrsfs}
\usepackage{subfigure}
\usepackage{subcaption} 

\usepackage{algorithm}
\usepackage{algpseudocode}
\usepackage{color}

\newtheorem{theorem}{Theorem}
\newtheorem{lemma}{Lemma}
\newtheorem{corollary}{Corollary}

\newtheorem{definition}{Definition}

\graphicspath{{./}{icons/}}
\tikzset{My Line Style/.style={samples=400}}
\graphicspath{{figures/}}

\usepackage[flushleft]{threeparttable}

\begin{document}

\title{\mbox{}\vspace{0.1cm}\\
\textsc{\huge Finite Field Multiple Access III: \\ from $2$-ary to $p$-ary} \vspace{0.2cm}}

\vspace{0.2cm}
\author{\normalsize
Qi-yue~Yu, {\it IEEE Senior Member}
\thanks{Q.-Y.~Yu (email: yuqiyue@hit.edu.cn) is affiliated with Harbin Institute of Technology, Heilongjiang, China.}
}

\maketitle
\vspace{-0.5in}
\begin{abstract}
This paper extends finite-field multiple-access (FFMA) techniques from binary to general $p$-ary source transmission. We introduce element-assemblage (EA) codes over GF($p^m$), which generalize element-pair (EP) codes, and define two specific types for ternary transmission: orthogonal EA codes and double codeword EA (D-CWEA) codes. We propose a unique sum-pattern mapping (USPM) constraint for the design of uniquely-decodable CWEA (UD-CWEA) codes, which include additive inverse D-CWEA (AI-D-CWEA) and basis decomposition D-CWEA (BD-D-CWEA) codes. Additionally, we introduce non-orthogonal CWEA (NO-CWEA) codes and their corresponding USPM constraint in the complex field.
Furthermore, $p$-ary CWEA codes are constructed using a basis decomposition method, leveraging ternary decomposition for faster convergence and simplified encoder/decoder design. We present a performance analysis of the proposed FFMA system from two complementary perspectives: channel capacity and error performance. We demonstrate that equal power allocation achieves the theoretical channel capacity, and then investigate the finite blocklength (FBL) characteristics of FFMA systems.
Moreover, we develop a rate-driven capacity alignment (CA) theorem based on the capacity-to-rate ratio (CRR) metric for error performance analysis. Finally, we compare $p$-ary transmission systems with classical binary transmission systems, revealing that low-order $p$-ary systems (e.g., $p = 3$) outperform binary systems at small loading factors, while higher-order systems (e.g., $p = 257$) excel at larger loading factors. These findings highlight the potential of $p$-ary systems, although practical implementations may benefit from decomposing $p$-ary systems into ternary systems to manage complexity.
\end{abstract}


\newpage
\section{Introduction}
Next-generation communication systems must simultaneously support massive connectivity and short-packet transmissions while ensuring reliable error performance \cite{6G}. Current state-of-the-art approaches typically combine advanced channel coding schemes with conventional multiple access (MA) techniques, such as polar codes with spreading sequences \cite{Polar_1, Polar_2, Polar_3, Polar_4} or traditional codes with interleave-division multiple access (IDMA) \cite{IDMA_1, IDMA_2}. However, these methods face a fundamental challenge: the inherent incompatibility between single-user channel coding requirements and multiuser transmission constraints. The conventional processing order, which applies channel coding before multiplexing (MUX), reduces the effective codeword length and exacerbates the multiuser finite blocklength (FBL) problem \cite{FBL_1, FBL_2, FBL_3, FBL_MU}.

To address this issue, we propose a finite-field multiple-access (FFMA) technique \cite{FFMA_ITW, FFMA, FFMA2}, which fundamentally reconfigures the signal processing chain by implementing multiplexing before channel coding. This architectural innovation significantly increases the codeword length, thereby improving error performance. The FFMA framework achieves user separation in finite-field domains using carefully designed \textit{element-pairs (EPs)}, which serve as virtual resources. Each EP is mathematically constructed over a Galois field GF($p^m$), where the \textit{prime factor (PF)} $p$ and the \textit{extension factor (EF)} $m$ together define the system's core characteristics. The Cartesian product of $J$ distinct EPs can form an EP code. We have developed several classes of EP codes, including symbol-wise and codeword-wise EP codes.

Building upon the various types of EP codes, the FFMA framework supports multiple transmission modes: \textit{time-division multiple access in finite field (FF-TDMA)}, \textit{code-division multiple access in finite field (FF-CDMA)}, \textit{channel-codeword multiple access in finite field (FF-CCMA)}, and \textit{non-orthogonal multiple access in finite field (FF-NOMA)} \cite{FFMA2}. Each mode provides distinct advantages for different deployment scenarios and quality-of-service (QoS) requirements, while retaining the core benefits of the FFMA architecture.

In \cite{FFMA2}, the \textit{finite-field to complex-field transform function} is a 3ASK (amplitude shift keying) signal.
Based on the EP codes constructed over GF($3^m$) and 3ASK signal, we can obtain \textit{error-correction orthogonal spreading codes} and \textit{error-correction non-orthogonal spreading codes}. 
Nevertheless, the information sequence of each user is from a binary source, i.e., a finite-field GF($2$).

As known, binary logic has been widely used in modern world, since most of the computing systems are implemented based on the famous von Neumann architecture \cite{History_Binary_1990}.
In addition to binary logic, ternary logic (or three-valued logic) is also promising, 
as it can encompass binary logic while retaining all of its advantages \cite{Ternary_Computing_1990, Ternary_Computing}. Based on the ternary logic, the computing can be implemented even faster \cite{Ternary_Computing}, and the storage is more compact \cite{TernaryCode_2017}.
Therefore, many works have been studied on ternary codes \cite{Ternary_CC_2020, Ternary_cyclic_2013,TPSK_2016, TPSK_2022, Ternary_Hamm_2013, Ternary_ParityCode_2019,Ternary_Sum_Code_2020, QAM9_2017, Ternary_MPA_2022}, such as 
ternary cyclic codes \cite{Ternary_cyclic_2013},
ternary convolutional codes for ternary phase shift keying (TPSK) signal\cite{TPSK_2016}, 
Turbo codes for TPSK signal \cite{TPSK_2022}, 
ternary Hamming codes \cite{Ternary_Hamm_2013}, 
ternary parity codes \cite{Ternary_ParityCode_2019}, 
ternary sum codes \cite{Ternary_Sum_Code_2020}, and etc.
In these works, TPSK and 9QAM (quadrature amplitude modulation) signals are widely used for supporting ternary transmission in a Gaussian point-to-point (P2P) channel. 
To support MA transmission in a Gaussian multiple-access channel (GMAC), uniquely decodable ternary codes (ternary UDC) are presented \cite{UD_CDMA1_2012, UD_CDMA2_2012, UD_CDMA3_2014, UD_CDMA4_2016, UD_CDMA5_2018, UD_CDMA6_2019}. These ternary UDCs can offer increased spectral efficiency (SE) for supporting multiuser transmission, classifiable as \textit{a form of non-orthogonal multiple-access (NOMA)} technique \cite{ZDing2017_survey, YChen_2018}.

Hence, extending the binary source to a more general $p$-ary source, such as the ternary source, is highly appealing as it enables the exploration of advanced features in FFMA systems. To support $p$-ary source transmission, we first introduce \textit{element-assemblage (EA)} codes over GF($p^m$), which generalize the existing EP codes. Specifically, for ternary source transmission, we define two types of EA codes: \textit{orthogonal EA} codes and \textit{double codewords EA (D-CWEA)} codes. Notably, the orthogonal EA code is a special case of the D-CWEA code. We then present the encoder structure for EA codes and introduce the concepts of \textit{parallel generator matrix} and \textit{parallel user block} to facilitate further analysis of EA codes.

Next, we propose a general \textit{unique sum-pattern mapping (USPM) structural property constraint} for designing uniquely-decodable CWEA (UD-CWEA) codes. In this work, we construct two specific types of UD-CWEA codes: \textit{additive inverse D-CWEA (AI-D-CWEA)} codes and \textit{basis decomposition D-CWEA (BD-D-CWEA)} codes. For a $J$-user AI-D-CWEA code, its codeword is uniquely determined by its full-one generator matrix, enabling the use of popular decoding algorithms such as the $q$-ary sum-product algorithm (QSPA). For BD-D-CWEA codes, we employ the \textit{parallel generator matrix} for efficient codeword decoding. Furthermore, from the perspective of various transmission modes in FFMA systems, we extend EP-coding to EA-coding and investigate its applications, particularly focusing on non-orthogonal CWEA (NO-CWEA) codes and their corresponding \textit{USPM constraint in the complex field}. Additionally, we construct $p$-ary CWEA codes using the basis decomposition method. We first apply ternary decomposition to an integer $p$, leveraging its fast convergence properties. Moreover, for ternary digits, the EA encoder and decoder can be simplified into a single generator, enhancing computational efficiency.

Following the introduction of element-assemblage (EA) coding for FFMA systems, we analyze the system performance based on the frameworks proposed in \cite{FFMA, FFMA2} and this work. First, we define key parameters for FFMA systems, such as the \textit{loading factor} and \textit{multirate sequences}. As \textit{channel capacity} and \textit{error performance} are two fundamental metrics in communication systems, we thoroughly examine both. We begin by analyzing the channel capacity of the proposed FFMA system along with its corresponding power allocation scheme, demonstrating that \textit{equal power allocation (EPA)} is optimal for achieving maximum capacity. Building on this analysis of channel capacity, we investigate the \textit{finite blocklength (FBL)} characteristics of FFMA systems.

Next, we define the \textit{capacity-to-rate ratio (CRR)} as the error performance metric and present the \textit{rate-driven capacity alignment (CA) theorem}. Using the CA theorem, we derive the CA power allocation scheme for FFMA systems, showing that EPA does not always align with the CA theorem. Finally, through systematic comparison enabled by the CA theorem, we quantify the performance advantages of $p$-ary transmission over conventional binary systems across different loading factors. This analysis provides valuable insights for system design trade-offs between spectral efficiency and energy requirements, specifically $E_b/N_0$.

The remainder of this paper is organized as follows. 
Section II begins with the definition of EA codes over GF($p^m$) and introduces two distinct types of EA codes. 
The encoding process for EA codes is detailed in Section III. 
Section IV focuses on the construction and decoding of $3$-ary EA codes. 
EA codes are further explored in Section V, where various transmission modes within FFMA systems are examined. 
Section VI demonstrates the construction of $p$-ary CWEA codes using the basis decomposition method. 
Key definitions for FFMA systems, including the loading factor and multirate sequences, are presented in Section VII. 
In Section VIII, the channel capacity of the proposed FFMA systems is analyzed, followed by the FBL analysis in Section IX. 
Section X introduces the rate-driven capacity alignment theorem. 
Building on the CA theorem, $p$-ary transmission is explored in Section XI. 
Finally, Section XII concludes the paper.

In this paper, the symbols \( \mathbb{B} = \{0, 1\} \), \( \mathbb{T} = \{0, 1, 2\} \), \( \mathbb{P} =\{0, 1, ..., p-1 \} \), and \( \mathbb{C} \) represent the binary field, ternary field, finite field \( \text{GF}(p) \), and the complex field, respectively.
The notation $(a)_q$ stands for modulo-$q$, and/or an element in GF($q$).
The symbols $\lceil x \rceil$ and $\lfloor x \rfloor$ denote the smallest integer that is greater than or equal to $x$ and the largest integer that is less than or equal to $x$, respectively.
In addition, we use $\Psi$ and $\Phi$ to express EP-coding and EA-coding, respectively.

\section{EA codes over Finite Fields}

For \( p \)-ary source transmission, we define \textit{\( p \)-ary element-assemblage (EA)} codes over an extension field \( \text{GF}(p^m) \), where \( p \) is a prime number, \( m \) is a positive integer with \( m \geq 2 \), and \( q = p^m \). As defined in \cite{FFMA}, \( p \) represents the prime factor (PF), and \( m \) represents the extension factor (EF).

In this paper, we focus primarily on the case of \( 3 \)-ary source transmission, i.e., when \( p = 3 \), and introduce two specific types of \( 3 \)-ary EA codes: \textit{orthogonal EA} codes and \textit{double codewords EA (D-CWEA)} codes, both of which are constructed over \( \text{GF}(3^m) \).

\vspace{-0.1in}
\subsection{Definition of $p$-ary EA Codes}

Let $\alpha$ be a primitive element of GF($p^m$). The powers of $\alpha$, i.e., $\alpha^{-\infty} = 0, \alpha^0 = 1, \alpha, \alpha^2, \ldots, \alpha^{(p^m - 2)}$, represent all the $p^m$ elements of GF($p^m$). 
Each element $\alpha^j$, where $j = -\infty, 0, \ldots, p^m - 2$, in GF($p^m$) can be expressed as a linear combination of the powers of $\alpha^0 = 1, \alpha, \alpha^2, \ldots, \alpha^{(m - 1)}$, with coefficients from GF($p$), as follows:
\vspace{-0.1in}
\begin{equation} \label{e2.2}
    \alpha^j = a_{j,0} + a_{j,1} \alpha + a_{j,2} \alpha^2 + \ldots + a_{j,m-1} \alpha^{(m-1)}.
\end{equation}
This indicates that the element $\alpha^j$ can be uniquely represented by an $m$-tuple $(a_{j,0}, a_{j,1}, a_{j,2}, \ldots, a_{j,m-1})$ over GF($p$).

For a $p$-ary source, there are $p$ digits, namely $(0)_p, (1)_p, \ldots, (p-1)_p$. Given an extension field GF($p^m$) of the prime field GF($p$), let a $p$-ary element-assemblage (EA) be denoted as 
\(
C_j = \left( \alpha^{l_{j,0}}, \alpha^{l_{j,1}}, \ldots, \alpha^{l_{j,p-1}} \right),
\)
where $l_{j,\varsigma} = -\infty, 0, 1, ..., p^m - 2$ for $0 \le \varsigma < p$ and the elements $\alpha^{l_{j,0}}, \alpha^{l_{j,1}}, \ldots, \alpha^{l_{j,p-1}}$ are distinct. Here, the subscript $j$ denotes the $j$-th EA, and the subscripts $0, 1, \ldots, p-1$ correspond to the input digits $(0)_p, (1)_p, \ldots, (p-1)_p$, respectively.

For $1 \leq j \leq M$, suppose we have $M$ $p$-ary EAs, given by
$C_1 = ({\alpha}^{l_{1,0}}, {\alpha}^{l_{1,1}}, \ldots, {\alpha}^{l_{1,p-1}}),$ 
$C_2 = ({\alpha}^{l_{2,0}}, {\alpha}^{l_{2,1}}, \ldots, {\alpha}^{l_{2,p-1}}),...,$ 
$C_M = ({\alpha}^{l_{M,0}}, {\alpha}^{l_{M,1}}, \ldots, {\alpha}^{l_{M,p-1}})$.
The Cartesian product of these $M$ $p$-ary EAs is defined as
\begin{equation}
  {\Phi} \triangleq C_1 \times C_2 \times \ldots \times C_M,
\end{equation}
which can form a $p$-ary EA code ${\Phi} = \{C_1, C_2, \ldots, C_j, \ldots, C_M\}$ over GF($p^m$), containing $p^M$ distinct EA codewords. 
Each element ${\alpha}^{l_{j,\varsigma}}$ of $C_j$ can be expressed as an $m$-tuple, where $0 \leq \varsigma < p$. Hence, we can form an $M \times m$ matrix ${\bf G}_{\rm M}^{\bf \varsigma}$ by arranging the $M$ elements $\alpha^{l_{1,\varsigma}}, \alpha^{l_{2,\varsigma}}, \ldots, \alpha^{l_{M,\varsigma}}$ as the rows of ${\bf G}_{\rm M}^{\bf \varsigma}$, specifically placing them from the first row to the $M$-th row, as follows:
\begin{equation*}
  \begin{aligned}
  {\bf G}_{\rm M}^{\bf \varsigma} = \left[
      \begin{matrix}
          \alpha^{l_{1,\varsigma}},
          \alpha^{l_{2,\varsigma}},
          \cdots,
          \alpha^{l_{M,\varsigma}}
      \end{matrix}
    \right]^{\rm T},
  \end{aligned}
\end{equation*}
where the subscript ``M'' stands for ``multiplexing'', and the superscript ``$\varsigma$'' indicates that all $M$ users are transmitting the digit $(\varsigma)_p$. This $M \times m$ matrix ${\bf G}_{\rm M}^{\bf \varsigma}$ is referred to as the \textit{full-$\varsigma$ generator matrix}.

In this paper, we mainly focus on the case where $p = 3$, with the corresponding extension field GF($3^m$).
Define the basic EA over GF($3^m$) by $C_{\rm T} = (0, 1, 2)$, where the subscript ``T'' stands for ``ternary''.
For the given finite field GF($3^m$), let the $3$-ary EA code $\Phi$ be defined as 
\(\Phi = \{C_1, C_2, \ldots, C_j, \ldots, C_M\}\),
where each $C_j = (\alpha^{l_{j,0}}, \alpha^{l_{j,1}}, \alpha^{l_{j,2}})$ for $1 \le j \le M$. Consequently, we obtain three $M \times m$ generator matrices: the \textit{full-zero generator matrix} denoted by ${\bf G}_{\rm M}^{\bf 0}$, the \textit{full-one generator matrix} denoted by ${\bf G}_{\rm M}^{\bf 1}$, and the \textit{full-two generator matrix} denoted by ${\bf G}_{\rm M}^{\bf 2}$. These matrices are given by
\(
  {\bf G}_{\rm M}^{\bf 0} = \left[
          \alpha^{l_{1,0}},
          \alpha^{l_{2,0}},
          \cdots,
          \alpha^{l_{j,0}},
          \cdots,
          \alpha^{l_{J,0}}
    \right]^{\rm T}, \\
    {\bf G}_{\rm M}^{\bf 1} = \left[
          \alpha^{l_{1,1}},
          \alpha^{l_{2,1}},
          \cdots,
          \alpha^{l_{j,1}},
          \cdots,
          \alpha^{l_{J,1}}
    \right]^{\rm T},
    {\bf G}_{\rm M}^{\bf 2} = \left[
          \alpha^{l_{1,2}},
          \alpha^{l_{2,2}},
          \cdots,
          \alpha^{l_{j,2}},
          \cdots,
          \alpha^{l_{J,2}}
    \right]^{\rm T},
\)
where the superscripts ``${\bf 0}$'', ``${\bf 1}$'', and ``${\bf 2}$'' indicate that all $M$ users transmit the digits $(0)_3$, $(1)_3$, and $(2)_3$, respectively.

In fact, ${\bf G}_{\rm M}^{\bf 0}$, ${\bf G}_{\rm M}^{\bf 1}$, and ${\bf G}_{\rm M}^{\bf 2}$ can be designed as distinct generator matrices, similar to the approach used in the paper \cite{FFMA2}. Thus, $\alpha^{l_{j,0}}$, $\alpha^{l_{j,1}}$, and $\alpha^{l_{j,2}}$ represent the codewords of ${\bf G}_{\rm M}^{\bf 0}$, ${\bf G}_{\rm M}^{\bf 1}$, and ${\bf G}_{\rm M}^{\bf 2}$, respectively, where $1 \le j \le M$. By carefully designing the generator matrices ${\bf G}_{\rm M}^{\bf 0}$, ${\bf G}_{\rm M}^{\bf 1}$, and ${\bf G}_{\rm M}^{\bf 2}$, we can obtain various $3$-ary EA codes. In the following, for brevity, we will refer to a ``$3$-ary EA code'' simply as an ``EA code'' when $p = 3$.

\vspace{-0.1in}
\subsection{Orthogonal EA codes}
For a given finite-field GF($3^m$), let the orthogonal EA code $\Phi_{\rm o, T}$ denote by
\begin{equation}
  \Phi_{\mathrm{o}, \mathrm{T}} = \bigl\{\alpha^{j-1} \cdot C_{\mathrm{T}} \bigm| 1 \leq j \leq m\bigr\}
  = \bigl\{C_j \bigm| 1 \leq j \leq m\bigr\},
\end{equation}
where $C_j = \alpha^{j-1} \cdot C_{\rm T} = \alpha^{j-1} \cdot (0, 1, 2)$ for $1 \le j \le m$.
The subscripts ``o'' stands for ``orthogonal''.
The Cartesian product
\begin{equation*}
{\Phi}_{\rm o, T} \triangleq (\alpha^{0} \cdot C_{\rm T}) \times (\alpha^{1} \cdot C_{\rm T}) \times \ldots \times (\alpha^{m-1} \cdot C_{\rm T}),
\end{equation*}
of the $m$ orthogonal EAs forms an orthogonal EA code ${\Phi}_{\rm o, T}$ over GF($3^m$) with $3^m$ EA codewords. 
For the orthogonal EA code $\Phi_{\rm o, T}$ over GF($3^m$), it can support a maximum of $m$ users, each transmitting $1$ bit.
The full-zero generator matrix ${\bf G}_{\rm M}^{\bf 0}$ of the orthogonal EA code ${\Phi}_{\rm o, T}$ is an $m \times m$ zero matrix.
The full-one and full-two generator matrices ${\bf G}_{\rm M}^{\bf 1}$ and ${\bf G}_{\rm M}^{\bf 2}$ of the orthogonal EA code ${\Phi}_{\rm o, T}$ are given as

\vspace{-0.1in}
\begin{small}
  \begin{equation}
{\bf G}_{\rm M}^{\bf 1} = \left[ 
  \begin{matrix}
  1 & 0 & \ldots & 0\\
  0 & 1 & \ldots & 0\\
  \vdots & \vdots & \ddots & \vdots \\
  0 & 0 & \ldots & 1\\
  \end{matrix}
\right], \quad
{\bf G}_{\rm M}^{\bf 2} = \left[ 
  \begin{matrix}
  2 & 0 & \ldots & 0\\
  0 & 2 & \ldots & 0\\
  \vdots & \vdots & \ddots & \vdots \\
  0 & 0 & \ldots & 2\\
  \end{matrix}
\right],
\end{equation}
\end{small}

indicating that ${\bf G}_{\rm M}^{\bf 1} = {\bf I}_m$ is an $m \times m$ identity matrix, 
and ${\bf G}_{\rm M}^{\bf 2} = 2 \cdot {\bf I}_m$. 
The loading factor $\eta$ of the generator matrix is equal to $1$, i.e., $\eta = 1$.

\vspace{-0.1in}
\subsection{Double CWEA codes}
For a given finite field GF($3^m$), if the three elements ${\alpha}^{l_{j,0}}$, ${\alpha}^{l_{j,1}}$, and ${\alpha}^{l_{j,2}}$ of the EA $C_j^{} = ({\alpha}^{l_{j,0}}, {\alpha}^{l_{j,1}}, {\alpha}^{l_{j,2}})$ satisfy
${\alpha}^{l_{j,0}} = {\bf 0}$, ${\alpha}^{l_{j,1}} \neq {\bf 0}$, and ${\alpha}^{l_{j,2}} \neq {\bf 0}$, 
i.e., $C_j^{} = ({\bf 0}, {\alpha}^{l_{j,1}}, {\alpha}^{l_{j,2}})$, where ${\bf 0}$ is an $m$-tuple and ${\alpha}^{l_{j,1}} \neq {\alpha}^{l_{j,2}}$, 
then the EA $C_j^{}$ is called a \textit{double codeword EA (D-CWEA)}.
The Cartesian product
\begin{equation*}
{\Phi}_{\rm cw} \triangleq C_1^{} \times C_2^{} \times \ldots \times C_M^{},
\end{equation*}
of the $M$ double CWEAs forms a D-CWEA code ${\Phi}_{\rm cw}$ over GF($3^m$) with $3^M$ EA codewords. The subscript ``cw'' denotes ``codeword''.
The full-zero generator matrix of the D-CWEA code ${\Phi}_{\rm cw}$ is an $M \times m$ zero matrix, i.e., ${\bf G}_{\rm M}^{\bf 0} = {\bf 0}$.
Therefore, it is sufficient to design the full-one generator matrix ${\bf G}_{\rm M}^{\bf 1}$ and the full-two generator matrix ${\bf G}_{\rm M}^{\bf 2}$ for the D-CWEA code.
Clearly, the orthogonal EA code can be considered as a special case of the D-CWEA code.

\textbf{Example 1:}
For an extension field GF($3^4$) of the prime field GF($3$), 
we can construct a $4$-user D-CWEA code 
${\Phi}_{\rm cw} = \{C_1^{}, C_2^{}, C_3^{}, C_4^{}\}$, given as

\vspace{-0.15in}
\begin{small}
  \begin{equation*} 
  \begin{aligned}
 C_1^{}  = (0000, 1 1 1 1, 2 2 2 2), \quad
 C_2^{} &= (0000, 2 1 2 1, 1 2 1 2),\\
 C_3^{}  = (0000, 2 2 1 1, 1 1 2 2), \quad
 C_4^{} &= (0000, 1 2 2 1, 2 1 1 2),\\
  \end{aligned}
\end{equation*}
\end{small}

whose Cartesian product can form a D-CWEA code ${\Phi}_{\rm cw}$ with $3^4=81$ codewords.
The full-one and full-two generator matrices ${\bf G}_{\rm M}^{\bf 1}$ and ${\bf G}_{\rm M}^{\bf 2}$ of the D-CWEA code $\Phi_{\rm cw}$ can be given as

\vspace{-0.1in}
\begin{small}
  \begin{equation} 
{\bf G}_{\rm M}^{\bf 1} = 
\left[
  \begin{matrix}
  1 & 1 & 1 & 1 \\
  2 & 1 & 2 & 1 \\
  2 & 2 & 1 & 1 \\
  1 & 2 & 2 & 1 \\
  \end{matrix}
  \right], \quad  
  {\bf G}_{\rm M}^{\bf 2} = 
\left[
  \begin{matrix}
  2 & 2 & 2 & 2 \\
  1 & 2 & 1 & 2 \\
  1 & 1 & 2 & 2 \\
  2 & 1 & 1 & 2 \\
  \end{matrix}
  \right],
\end{equation}
\end{small}
which are two $4 \times 4$ ternary matrices, similar to Example 2 in Ref. \cite{FFMA2}.
$\blacktriangle \blacktriangle$

\textbf{Example 2:}
Consider an extension field GF($3^8$) of the prime field GF($3$), and suppose there are four users. In this case, we can construct a 4-user D-CWEA code ${\Phi}_{\rm cw} = \{C_1^{}, C_2^{}, C_3^{}, C_4^{}\}$, as follows:
\begin{small}
  \begin{equation*} \label{e.Ex2_1}
  \begin{aligned}
 C_1^{} = ({\bf 0}, {\alpha}^{l_{1,1}}, {\alpha}^{l_{1,2}}) = (0000 0000, 1111 1111, 2222 2222) \\
 C_2^{} = ({\bf 0}, {\alpha}^{l_{2,1}}, {\alpha}^{l_{2,2}}) = (0000 0000, 0000 1111, 0000 2222) \\
 C_3^{} = ({\bf 0}, {\alpha}^{l_{3,1}}, {\alpha}^{l_{3,2}}) = (0000 0000, 0011 0011, 0022 0022) \\
 C_4^{} = ({\bf 0}, {\alpha}^{l_{4,1}}, {\alpha}^{l_{4,2}}) = (0000 0000, 0101 0101, 0202 0202) \\
  \end{aligned},
\end{equation*}
\end{small}

whose Cartesian product can form a D-CWEA code ${\Phi}_{\rm cw}$ with $3^{4}= 81$ codewords.
Thus, its full-one generator matrix ${\bf G}_{\rm M}^{\bf 1}$ can be given as,
\begin{small}
 \begin{equation*} 
  {\mathbf G}_{\rm M}^{\bf 1} = \left[
  \begin{array}{c}
    {\alpha^{l_{1,1}}}\\
    {\alpha^{l_{2,1}}}\\
    {\alpha^{l_{3,1}}}\\
    {\alpha^{l_{4,1}}}\\
  \end{array} 
  \right] =
  \left[
  \begin{array}{*{16}{cccccccccccccccc}}
    1 & 1  & 1 & 1 & 1 & 1 & 1 & 1  \\
    0 & 0  & 0 & 0 & 1 & 1 & 1 & 1  \\
    0 & 0  & 1 & 1 & 0 & 0 & 1 & 1  \\
    0 & 1  & 0 & 1 & 0 & 1 & 0 & 1  \\
  \end{array}
  \right],
\end{equation*} 
\end{small}
and its full-two generator matrix ${\bf G}_{\rm M}^{\bf 2}$ is
\begin{small}
  
\begin{equation*} 
  {\mathbf G}_{\rm M}^{\bf 2} = \left[
  \begin{array}{c}
    {\alpha^{l_{1,2}}}\\
    {\alpha^{l_{2,2}}}\\
    {\alpha^{l_{3,2}}}\\
    {\alpha^{l_{4,2}}}\\
  \end{array} 
  \right] =
  \left[
  \begin{array}{*{16}{cccccccccccccccc}}
    2 & 2  & 2 & 2 & 2 & 2 & 2 & 2  \\
    0 & 0  & 0 & 0 & 2 & 2 & 2 & 2  \\
    0 & 0  & 2 & 2 & 0 & 0 & 2 & 2  \\
    0 & 2  & 0 & 2 & 0 & 2 & 0 & 2  \\
  \end{array}
  \right],
\end{equation*}
\end{small}
which are two $4 \times 8$ matrices.
In addition, it is found that ${\mathbf G}_{\rm M}^{\bf 2} = 2 \cdot {\mathbf G}_{\rm M}^{\bf 1}$.
$\blacktriangle \blacktriangle$

\vspace{-0.1in}
\section{Encoding of EA Codes}

In this section, we first introduce the EA encoder in serial mode. The EA encoder in parallel mode is not discussed further, as it is merely a combination of the serial-mode EA encoder and the EP encoder in parallel mode. We then define the \textit{parallel generator matrix} of a D-CWEA code, which is used for decoding the D-CWEA code at the receiver. It is important to note that the effect of the channel code ${\mathcal C}_{gc}$ has been studied in \cite{FFMA, FFMA2}. For EA coding, the globally encoded process is identical to that of EP coding and will not be repeated here.

\vspace{-0.1in}
\subsection{EA encoder in Serial Mode}
For a ternary source, let the digit sequence of the $j$-th user be denoted by ${\mathbf d}_j = (d_{j,0}, d_{j,1}, \ldots, d_{j,k}, \ldots, d_{j,K-1})$, where $K$ is the number of digits per user, and $d_{j,k} \in {\mathbb T}$, with $1 \leq j \leq M$ and $0 \leq k < K$. Suppose the EA encoder operates in serial mode \cite{FFMA2}, and each user is assigned a unique EA, such as the EA $C_j = (\alpha^{l_{j,0}}, \alpha^{l_{j,1}}, \alpha^{l_{j,2}})$, which corresponds to the $j$-th user. Based on the EA $C_j$, the digit sequence ${\mathbf d}_j$ is encoded using a \textit{finite-field GF($p$) to finite-field GF($q$) transform function}, denoted as ${\mathrm F}_{p2q}$, producing the element sequence ${\mathbf u}_j = (u_{j,0}, u_{j,1}, \ldots, u_{j,k}, \ldots, u_{j,K-1})$.
For the $k$-th component $u_{j,k}$ of ${\bf u}_j$, we have
\begin{equation} \label{F_b2q}
  u_{j,k} = {\mathrm F}_{p2q}(d_{j,k}) \triangleq d_{j,k} \odot C_j = 
  \left\{
    \begin{aligned}
      \alpha^{l_{j,0}}, \quad d_{j,k} = (0)_3  \\
      \alpha^{l_{j,1}}, \quad d_{j,k} = (1)_3  \\
      \alpha^{l_{j,2}}, \quad d_{j,k} = (2)_3  \\
    \end{aligned},
  \right.
\end{equation}
where $d_{j,k} \odot C_j$ is defined as a \textit{switching function} \cite{FFMA, FFMA2}.
If the input digit is $d_{j,k} = (0)_3$, the corresponding transformed element is \( u_{j,k} = \alpha^{l_{j,0}} \); if the input digit is $d_{j,k} = (1)_3$, the transformed element is $u_{j,k} = \alpha^{l_{j,1}}$; otherwise, the transformed element is $u_{j,k} = \alpha^{l_{j,2}}$.

In this paper, the \textit{finite-field multiplex module (FF-MUX)} \( {\mathcal A}_{\rm M} \) is defined as \( {\mathcal A}_{\rm M} = [1, 1, \ldots, 1]^{\rm T} \).  
Thus, finite-field addition is performed on the \( k \)-th components of the \( M \) users at the receiving end, i.e., \( (u_{1,k}, u_{2,k}, \ldots, u_{j,k}, \ldots, u_{M,k}) \), which also forms an EA codeword of \( \Phi \).  
From this, we obtain a \textit{finite-field sum-pattern (FFSP)} block, expressed as
\vspace{-0.05in}
\begin{small}
\begin{equation} \label{e.w_k}
  \begin{aligned}
  w_k = \bigoplus_{j=1}^{M} u_{j,k} 
  = u_{1,k} \oplus u_{2,k} \oplus \ldots \oplus u_{M,k} 
  = {\bf d}[k] \cdot {\mathbf G}_{\mathrm M}^{{\bf d}[k]}, 
  \end{aligned}
\end{equation}
\end{small}
where ${\bf d}[k] = (d_{1,k}, d_{2,k}, \ldots, d_{M,k})$ represents a $1 \times M$ user block of $M$ users on the $k$-th component for $0 \le k < K$.
${\mathbf G}_{\mathrm M}^{{\bf d}[k]}$ is denoted as the generator matrix of ${\bf d}[k]$, given as
\begin{equation}
  \begin{aligned}
 {\mathbf G}_{\mathrm M}^{{\bf d}[k]} 
  = [d_{1,k} \odot C_1, d_{2,k} \odot C_2, \ldots, d_{M,k} \odot C_M ]^{\rm T} 
  = [\alpha^{l_{d_{1,k}}}, \alpha^{l_{d_{2,k}}},\cdots,\alpha^{l_{d_{M,k}}}]^{\rm T}
  \end{aligned},
\end{equation}
which is an $M \times m$ matrix, and $\alpha^{l_{d_{j,k}}} = d_{j,k} \odot C_j$.
All the combinations of \( {\mathbf G}_{\mathrm M}^{{\bf d}[k]} \) can form a generator matrix set \( {\mathcal G}_{\mathrm M} \), i.e.,  
\(
{\mathcal G}_{\rm M} = \{{\mathbf G}_{\mathrm M}^{{\bf 0}}, {\mathbf G}_{\mathrm M}^{{\bf 1}}, \ldots, {\mathbf G}_{\mathrm M}^{{\bf d}[k]}, \ldots, {\mathbf G}_{\mathrm M}^{{\bf 2}}\}
\)
which consists of \( 3^M \) matrices of size \( M \times m \), including the full-zero, full-one, full-two, and other generator matrices.  

Utilizing the generator matrix set ${\mathcal G}_{\mathrm M}$, the input $K$ user blocks ${\bf d}[0], {\bf d}[1], \ldots, {\bf d}[K-1]$ can be encoded into $K$ FFSP blocks $w_0, w_1, \ldots, w_{K-1}$. 
Then, the $K$ FFSP blocks, i.e., $w_0, w_1, \ldots, w_{K-1}$, together form a $1 \times mK$ FFSP sequence ${\bf w} = (w_0, w_1, \ldots, w_{K-1})$.

\vspace{-0.1in}
\subsection{Parallel Generator Matrix of a D-CWEA Code}

Although the FFSP block can be obtained through Eq. (\ref{e.w_k}), it is not straightforward to operate on and decode the D-CWEA code. Therefore, it is necessary to provide a more direct expression of the FFSP block to assist in decoding the EA codeword. In order to achieve this, we further explore the definition of the D-CWEA code.

According to the definition of a D-CWEA code $\Phi_{\rm cw}$, it is characterized by its full-one and full-two generator matrices, denoted as ${\bf G}_{\rm M}^{\bf 1}$ and ${\bf G}_{\rm M}^{\bf 2}$.
To consider both the full-one and full-two generator matrices ${\bf G}_{\rm M}^{\bf 1}$ and ${\bf G}_{\rm M}^{\bf 2}$ of the D-CWEA code $\Phi_{\rm cw}$ together, we define the \textit{parallel generator matrix} ${\bf G}_{\rm M, pll}$ of the D-CWEA code $\Phi_{\rm cw}$ as follows:
\begin{equation}
{\bf G}_{\rm M, pll} = 
\left[
  \begin{matrix}
  {\bf G}_{\rm M}^{\bf 2} \\
  \hdashline
  {\bf G}_{\rm M}^{\bf 1} \\
  \end{matrix}
\right],
\end{equation}
where ${\bf G}_{\rm M}^{\bf 1}$ and ${\bf G}_{\rm M}^{\bf 2}$ occupy the lower and upper sections of the parallel generator matrix ${\bf G}_{\rm M, pll}$, respectively. The size of the parallel generator matrix ${\bf G}_{\rm M, pll}$ is $2M \times m$, where the subscript ``pll'' stands for ``parallel''.

Next, we investigate the relationship between the input user block \({\bf d}[k]\) and the parallel generator matrix \({\bf G}_{\rm M, pll}\). The input user block \({\bf d}[k] = (d_{1,k}, d_{2,k}, \ldots, d_{M,k}) \in {\mathbb T}^{1 \times M}\) is a ternary vector of length \(M\), while the parallel generator matrix \({\bf G}_{\rm M, pll}\) is a \(2M \times m\) matrix. Consequently, the input user block \({\bf d}[k]\) should first be transformed into a \textit{parallel user block} \({\bf a}[k]\), defined as:
\[
{\bf a}[k] = {\rm F}_{\rm T2B}({\bf d}[k]) 
= \left( d_{1,k}^{(2)}, d_{2,k}^{\rm (2)}, \ldots, d_{j,k}^{\rm (2)}, \ldots, d_{M,k}^{\rm (2)}, 
d_{1,k}^{\rm (1)}, d_{2,k}^{\rm (1)}, \ldots, d_{j,k}^{\rm (1)}, \ldots, d_{M,k}^{\rm (1)} \right)_2 \in {\mathbb B}^{1 \times 2M},
\]
which is a binary vector of length \(2M\), and \({\rm F}_{\rm T2B}\) denotes the \textit{ternary-to-binary transform function}, also referred to as the \textit{ternary-to-binary (T2B)} encoder.

Here, each digit \(d_{j,k}\) of the input user block \({\bf d}[k]\) is uniquely mapped to a binary vector through the transform function \({\rm F}_{\rm T2B}\), given as:
\begin{equation} \label{e.T2B}
  {\rm F}_{\rm T2B}(d_{j,k}) = 
  \left( d_{j,k}^{\rm (2)}, d_{j,k}^{\rm (1)} \right)_2 =
\begin{cases}
(0, 0)_2, & \text{if } d_{j,k} = (0)_3, \\
(0, 1)_2, & \text{if } d_{j,k} = (1)_3, \\
(1, 0)_2, & \text{if } d_{j,k} = (2)_3.
\end{cases}
\end{equation}
Note that since \(d_{j,k} \neq (3)_3\), there will be no case where \((d_{j,k}^{\rm (2)}, d_{j,k}^{\rm (1)})_2 = (1,1)_2\), meaning that \((d_{j,k}^{\rm (2)}, d_{j,k}^{\rm (1)})_2 \neq (1,1)_2\). Hence, the efficiency of the transform function \({\rm F}_{\rm T2B}\) is \(3/4\).

From equation (\ref{e.T2B}), the binary vector \(\left( d_{j,k}^{\rm (2)}, d_{j,k}^{\rm (1)} \right)_2\) can also be interpreted as the \textit{binary decomposition} of the integer number \(d_{j,k}\). Here, the superscripts ``\(\rm (1)\)'' and ``\(\rm (2)\)'' indicate the first and second bits in the binary representation of the number \(d_{j,k}\).

When there are \(M\) users, the system generates \(3^M\) parallel user blocks \({\bf a}[k]\), each of length \(1 \times 2M\). These blocks are then processed by the \(2M \times m\) parallel generation matrix \({\bf G}_{\rm M, pll}\), producing the FFSP block \(w_k\) as:
\begin{equation} \label{e.w_pll}
  w_k = {\bf a}[k] \cdot {\mathbf G}_{\rm M, pll},
\end{equation}
where \(w_k\) represents the codeword of the \textit{parallel multiuser code \({\mathcal C}_{mc, \rm pll}\)}, whose generator matrix is \({\bf G}_{\rm M, pll}\).
Based on the parallel generator matrix \({\bf G}_{\rm M, pll}\), we can directly decode the FFSP block.

\section{Construction of $3$-ary EA codes}

In this section, we examine the general \textit{unique sum-pattern mapping (USPM) structural property constraint}, which is essential for constructing \textit{uniquely decodable CWEA (UD-CWEA)} codes. We then proceed to construct two types of D-CWEA codes: \textit{additive inverse double CWEA (AI-D-CWEA) codes} and \textit{basis-decomposing double CWEA (BD-D-CWEA) codes}.

\vspace{-0.1in}
\subsection{USPM Structural Property Constraint}

To construct \textit{uniquely decodable EA (UD-EA)} codes, we first introduce the general USPM structural property constraint.

\vspace{-0.1in}
\begin{theorem} \label{theorem0}
  (\textbf{USPM Constraint})  
  Let \( \text{GF}(p^m) \) denote a finite field, where \( p \) is a prime number and \( m \) is an integer with \( m \geq 2 \). An EA code \( \Phi = \{ C_1, C_2, \ldots, C_M \} \) supports an \( M \)-user FFMA system, where \( M \leq m \).  
  For the FFSP block \( w_k \) to be uniquely decodable, there must exist a one-to-one mapping between the input user block \( \mathbf{d}[k] \) and the FFSP block \( w_k \), i.e., \( \mathbf{d}[k] \leftrightarrow w_k \) for \( 0 \leq k < K \). This ensures that the code satisfies the unique sum-pattern mapping (USPM) structural property. The generator matrix of the user block \( \mathbf{d}[k] \), denoted as \( \mathbf{G}_{\rm M}^{\mathbf{d}[k]} \), is defined as  
  \(
  \mathbf{G}_{\rm M}^{\mathbf{d}[k]} = \left[ \alpha^{l_{d_{1,k}}}, \alpha^{l_{d_{2,k}}}, \cdots, \alpha^{l_{d_{M,k}}} \right]^{\rm T}.
  \)
  This matrix must have full row rank, i.e.,  
  \begin{equation} 
      {\rm Rank}\left( \mathbf{G}_{\rm M}^{\mathbf{d}[k]} \right) = M,
  \end{equation}  
  which implies that the vectors \( \alpha^{l_{d_{1,k}}}, \alpha^{l_{d_{2,k}}}, \cdots, \alpha^{l_{d_{M,k}}} \) are linearly independent.  
  We refer to \( \Phi \) as a uniquely decodable EA (UD-EA) code.
\end{theorem}

\begin{proof}
According to (\ref{e.w_k}), the FFSP block \( w_k = {\bf d}[k] \cdot {\mathbf G}_{\mathrm M}^{{\bf d}[k]} \) can be expressed as:
  \[
  w_k = d_{1,k} \cdot \alpha^{l_{d_{1,k}}} + d_{2,k} \cdot \alpha^{l_{d_{2,k}}} + \cdots + d_{M,k} \cdot \alpha^{l_{d_{M,k}}}.
  \]
Thus, for \( \alpha^{l_{d_{1,k}}}, \alpha^{l_{d_{2,k}}}, \dots, \alpha^{l_{d_{M,k}}} \) to be linearly independent, it ensures a one-to-one mapping between the input user block \( {\bf d}[k] \) and the corresponding FFSP block \( w_k \), thereby satisfying the USMP structural property.
In other words, the generator matrix \( {\mathbf G}_{\mathrm M}^{{\bf d}[k]} \) must be of full row rank.
\end{proof}

Theorem \ref{theorem0} presents the general USMP structural property constraint. However, the generator matrix set \( {\mathcal G}_{\rm M} \) contains \( 3^M \) matrices, which introduces complexity in calculating the USMP constraint. Therefore, this paper focuses solely on the construction of UD-D-CWEA codes.

\vspace{-0.1in}
\subsection{Construction of AI-D-CWEA Codes}

In \cite{FFMA2}, we have introduced the \textit{USPM constraint of AI-CWEP codes over GF($3^m$)}.
It shows that if the full-one generator matrix of the AI-CWEP code $\Psi_{\rm ai}$ is a full row rank matrix, then the rows of any matrix in the generator matrix set ${\mathcal G}_{\rm M}$ are also linearly independent. In this paper, the subscript ``ai'' stands for ``additive inverse''.
Thus, we can deduce the following Lemma.

\begin{lemma} \label{lemma0}
  A uniquely decodable AI-CWEP code over GF\((3^m)\) is also a uniquely decodable D-CWEA code. 
\end{lemma}

\begin{proof}
Recall the definition of a D-CWEA, i.e., \( C_j^{\rm cw} = ({\bf 0}, \alpha^{l_{j,1}}, \alpha^{l_{j,2}}) \). The output element of an EA encoder is defined as \( \alpha^{l_{d_{j,k}}} \), which can either be from a row of the full-one generator matrix \( {\bf G}_{\rm M}^{\bf 1} \) or from a row of the full-two generator matrix \( {\bf G}_{\rm M}^{\bf 2} \). Specifically, we have \( \alpha^{l_{d_{j,k}}} \in \{ {\bf 0}, \alpha^{l_{1,1}}, \ldots, \alpha^{l_{M,1}}, \alpha^{l_{1,2}}, \ldots, \alpha^{l_{M,2}} \} \).

First, we construct a uniquely decodable AI-CWEP (UD-AI-CWEP) code \( \Psi_{\rm ai} \) over GF\((3^m)\), where the full-one generator matrix and its \textit{additive inverse matrix} are set to be \( {\bf G}_{\rm M, ep}^{\bf 1} \) and \( {\bf P} - {\bf G}_{\rm M, ep}^{\bf 1} \), respectively. In this case, \( {\bf P} \) is an \( M \times m \) matrix with all elements equal to \( p \).

Next, let the full-one and full-two generator matrices of the D-CWEA code \( \Phi_{\rm cw} \), denoted as \( {\bf G}_{\rm M}^{\bf 1} \) and \( {\bf G}_{\rm M}^{\bf 2} \), be given by
\(
  {\bf G}_{\rm M}^{\bf 1} = {\bf G}_{\rm M, ep}^{\bf 1} \)
and 
\(
  {\bf G}_{\rm M}^{\bf 2} = {\bf P} - {\bf G}_{\rm M, ep}^{\bf 1}.
\)
Since the UD-AI-CWEP code \( \Psi_{\rm ai} \) satisfies the USPM constraint, it follows that the generator matrix set \( {\mathcal G}_{\rm M} \) of the D-CWEA code, i.e., \( {\mathcal G}_{\rm M} = \{ {\bf 0}, {\bf G}_{\rm M}^{\bf 1}, \dots, {\bf G}_{\rm M}^{\bf 2} \} \), also satisfies the USPM constraint. Thus, a UD-AI-CWEP code over GF\((3^m)\) is also a UD-D-CWEA code.
\end{proof}

According to Lemma \ref{lemma0}, we can construct a D-CWEA code using a generator matrix \( \mathbf{G}_{\rm M}^{\mathbf{1}} \) and its additive inverse matrix \( \mathbf{P} - \mathbf{G}_{\rm M}^{\mathbf{1}} \). We refer to this construction, where the full-one generator matrix is \( \mathbf{G}_{\rm M}^{\mathbf{1}} \) and the full-two generator matrix is its additive inverse \(\mathbf{G}_{\rm M}^{\mathbf{2}}= \mathbf{P} - \mathbf{G}_{\rm M}^{\mathbf{1}} \), as an \textit{additive inverse D-CWEA (AI-D-CWEA) code}, denoted by \( \Phi_{\rm ai} \).

For the AI-D-CWEA code, the generator matrix \( \mathbf{G}_{\rm M}^{{\bf d}[k]} \) of \( \Phi_{\rm ai} \) can be significantly simplified. Based on the properties of the prime field GF(3) as outlined in \cite{FFMA2}, we have:
${\bf G}_{\rm M}^{\bf 2} = \mathbf{P} - \mathbf{G}_{\rm M}^{\mathbf{1}} = 2{\bf G}_{\rm M}^{\bf 1}$.

Hence, the \( k \)-th component \( u_{j,k} \) of \( \mathbf{u}_j \), given by equation (\ref{F_b2q}), can be rewritten as:
\begin{equation} 
  u_{j,k} = {\mathrm F}_{p2q}(d_{j,k}) = d_{j,k} \cdot \alpha^{l_{j,1}}
  = \left\{
    \begin{aligned}
      0, \quad d_{j,k} = (0)_3,  \\
      \alpha^{l_{j,1}}, \quad d_{j,k} = (1)_3,  \\
      2\alpha^{l_{j,1}}, \quad d_{j,k} = (2)_3.  \\
    \end{aligned}
  \right.
\end{equation}
Here, the switching operation \( d_{j,k} \odot C_j \) has been replaced by a simple scalar multiplication, i.e., \( d_{j,k} \cdot \alpha^{l_{j,1}} \).

Thus, the FFSP given by equation (\ref{e.w_k}) simplifies to:
\begin{equation} \label{e.w_k_ai}
  \begin{aligned}
  w_k = {\bf d}[k] \cdot {\mathbf G}_{\mathrm M}^{{\bf d}[k]}
      = {\bf d}[k] \cdot {\mathbf G}_{\mathrm M}^{\bf 1}, 
  \end{aligned}
\end{equation}
which is equivalent to the user block \( \mathbf{d}[k] = (d_{1,k}, d_{2,k}, \ldots, d_{M,k}) \) passing through an \( M \times m \) full-one generator matrix \( \mathbf{G}_{\mathrm{M}}^{\bf 1} \).
Since \( \mathbf{G}_{\mathrm{M}}^{\bf 1} \) is constructed over GF($3$), any ternary decoding algorithm, such as the \( 3 \)-ary sum-product algorithm (QSPA), can be used to recover the user block \( \mathbf{d}[k] \).

The generator matrix \( \mathbf{G}_{\rm M}^{\mathbf{1}} \) may either be a ternary orthogonal matrix with a loading factor of $1$ or a generator matrix of a linear block code with a loading factor less than $1$. We will now discuss these two cases in detail.

\subsubsection{Based on a Ternary Orthogonal Matrix}

As introduced in \cite{FFMA2}, the $\kappa$-fold ternary orthogonal matrix ${\bf T}_{\rm o}(2^{\kappa}, 2^{\kappa})$ (or simply ${\bf T}_{\rm o}$) is a $2^{\kappa} \times 2^{\kappa}$ matrix. Based on ${\bf T}_{\rm o}$ and its additive inverse matrix ${\bf T}_{\rm o, ai}$, we can construct an AI-D-CWEA code \( \Phi_{\rm ai} \) over GF($3^m$), where \( m = 2^{\kappa} \), and the loading factor of the generator matrix is 1.
In fact, upon reviewing Example 1, it is observed that the proposed D-CWEA code \( \Phi_{\rm cw} \) is equivalent to the UD-AI-CWEP code \( \Psi_{\rm ai} \) as described in \cite{FFMA2}. Consequently, the D-CWEA code \( \Phi_{\rm cw} \) presented in Example 1 is a UD-D-CWEA code.

\textbf{Example 3:}
Now, we construct a $2$-user UD-D-CWEA code, denoted as \( \Phi_{\rm cw, eg, 1} \), with a loading factor of 1, where the subscript ``eg'' stands for ``example''. Based on the ternary orthogonal matrix ${\rm T}_{\rm o}(2,2)$ over GF($3^2$), we define the $2$-user UD-D-CWEA code \( \Phi_{\rm cw, eg, 1} \), whose full-one and full-two generator matrices are given as follows:
\begin{equation} 
{\bf G}_{\rm M, eg, 1}^{\bf 1} = 
\left[
  \begin{matrix}
  1 & 1 \\
  2 & 1 \\
  \end{matrix}
  \right], \quad  
  {\bf G}_{\rm M, eg, 1}^{\bf 2} = 
\left[
  \begin{matrix}
  2 & 2 \\
  1 & 2 \\
  \end{matrix}
  \right],
\end{equation}
which can form $3^2 = 9$ EA codewords.

For a $2$-user D-CWEA code, the input user blocks are given as follows:
\begin{equation} \label{e.d_k_2user}
  \begin{array}{cc}
  {\bf d}[k] \in \{00, 01, 02, 10, 11, 12, 20, 21, 22\}. \\
  \end{array}
\end{equation} 
According to Eq. (\ref{e.w_k_ai}), i.e., \( w_k = {\bf d}[k] \cdot {\mathbf G}_{\mathrm M}^{{\bf 1}} \), the corresponding FFSP blocks can be calculated as \( (00)_3 \), \( (21)_3 \), \( (12)_3 \), \( (11)_3 \), \( (02)_3 \), \( (20)_3 \), \( (22)_3 \), \( (10)_3 \), and \( (01)_3 \).

Next, we recalculate the FFSP blocks using the parallel generator matrix. The parallel generator matrix ${\bf G}_{\rm M, pll, eg, 1}$ for \( \Phi_{\rm cw, eg, 1} \) is given by:
\begin{equation} 
{\bf G}_{\rm M, pll, eg, 1} = 
\left[
  \begin{matrix}
  {\bf G}_{\rm M, eg, 1}^{\bf 2} \\
  \hdashline
  {\bf G}_{\rm M, eg, 1}^{\bf 1} \\
  \end{matrix}
  \right]
=
\left[
  \begin{matrix}
  2 & 2 \\
  1 & 2 \\
  \hdashline
  1 & 1 \\
  2 & 1 \\
  \end{matrix}
  \right],
\end{equation}
which is a \( 4 \times 2 \) matrix with row rank \( 2 \).

The input user blocks are given by Eq. (\ref{e.d_k_2user}), and their corresponding parallel user blocks are as follows:
\begin{equation} \label{e.a_k_2user}
  \begin{array}{cc}
  {\bf a}[k] \in \Lambda =\{0000, 0001, 0100, 0010, 0011, 0110, 1000, 1001, 1100 \}.
  \end{array}
\end{equation}

When the input parallel user blocks are taken from \( \Lambda \), the output codewords of \( {\bf G}_{\rm M, pll, eg, 1} \) are \( 00 \), \( 21 \), \( 12 \), \( 11 \), \( 02 \), \( 20 \), \( 22 \), \( 10 \), and \( 01 \), which match the FFSP blocks calculated by Eq. (\ref{e.w_k_ai}).
$\blacktriangle \blacktriangle$

\subsubsection{Based on a Generator Matrix of a Linear Block Code}

From Lemma 1, we can directly construct a UD-D-CWEA code based on the generator matrix of a binary linear block code \( {\mathcal C}_{mc} \).
Let \( \mathcal{C}_{{mc}} \) be a binary linear block code defined over GF($2$). The generator matrix of \( \mathcal{C}_{{mc}} \), denoted by \( \mathbf{G}_{{mc}} \), is an \( M \times m \) binary matrix. Next, we extend the dimensionality of \( \mathbf{G}_{{mc}} \) and generalize its elements from GF($2$) to GF($3$).
Let the \textit{additive inverse matrix} \( \mathbf{G}_{{mc,\rm ai}} \) of \( \mathbf{G}_{{mc}} \) be defined as
\begin{equation}
  \mathbf{G}_{{mc, \rm ai}} = \mathbf{P} - \mathbf{G}_{mc} \overset{(a)}{=} 2 \cdot \mathbf{G}_{mc},
\end{equation}
where the relation (a) follows from Property 1 in \cite{FFMA2}. Each element of \( \mathbf{G}_{mc,\rm ai} \) is either \( (0)_3 \) or \( (2)_3 \).

Subsequently, we obtain a UD-D-CWEA code \( \Phi_{\rm ai} \), where the full-one and full-two matrices are given by \( {\bf G}_{\rm M}^{\bf 1} = {\bf G}_{mc} \) and \( {\bf G}_{\rm M}^{\bf 2} = {\bf G}_{mc, \rm ai} \), respectively.
It is important to note that the additive inverse matrix \( {\bf G}_{mc, \rm ai} \) of the linear code \( {\mathcal C}_{mc} \) has the same minimum distance as the generating matrix \( {\bf G}_{mc} \).
Summarizing the above process, we derive the following corollary:

\begin{corollary}
  Let \( \Phi_{\rm ai} \) be a D-CWEA code over GF($3^m$), and let \( {\mathcal C}_{mc} \) denote a binary linear block code. If the D-CWEA code \( \Phi_{\rm ai} \) has the full-one generator matrix \( {\bf G}_{\rm M}^{\bf 1} \) and the full-two generator matrix \( {\bf G}_{\rm M}^{\bf 2} \), where \( {\bf G}_{\rm M}^{\bf 1} = {\bf G}_{mc} \) and \( {\bf G}_{\rm M}^{\bf 2} = {\bf G}_{mc, \rm ai} \), with \( {\bf G}_{mc} \) being the \( M \times m \) generator matrix of the linear code \( {\mathcal C}_{mc} \) and \( {\bf G}_{mc,\rm ai} \) being the \( M \times m \) additive inverse matrix of \( {\bf G}_{mc} \), then \( \Phi_{\rm ai} \) is an \( M \)-user UD-D-CWEA code over GF($3^m$).
\end{corollary}

Recall Example 2, where the rank of the matrix \( {\bf G}_{\rm M}^{\bf 1} \) for the D-CWEA code \( \Phi_{\rm cw} \) is $4$, indicating that it has full row rank. Since \( {\bf G}_{\rm M}^{\bf 2} = 2 \cdot {\bf G}_{\rm M}^{\bf 1} \), it follows that \( {\bf G}_{\rm M}^{\bf 2} \) is the additive inverse matrix of \( {\bf G}_{\rm M}^{\bf 1} \). Therefore, the D-CWEA code \( \Phi_{\rm cw} \) from Example 2 is a UD-D-CWEA code with a loading factor of 0.5. Furthermore, the minimum distance of \( {\bf G}_{\rm M}^{\bf 2} \) remains $4$.

\textbf{Example 4:}
Given a $(7, 4)$ linear cyclic block code over GF($2$), whose generator polynomial is $g(X) = 1+X^2+X^3$. We can construct a $2$-user UD-D-CWEA code ${\Phi}_{\rm cw, eg, 2}$ based on the linear block code.
The full-one generator matrix ${\bf G}_{\rm M, eg, 2}^{\bf 1}$ and the full-two generator matrix ${\bf G}_{\rm M, eg, 2}^{\bf 2}$ of the UD-D-CWEA code ${\Phi}_{\rm cw, eg, 2}$ can be given as
\begin{equation*}
{\bf G}_{\rm M, eg, 2}^{\bf 1} =\left[
  \begin{matrix}
    g(X)\\
    x \cdot g(X)\\
   \end{matrix}
\right]
= \left[
  \begin{matrix}
    1 & 0 & 1 & 1 & 0 & 0 & 0\\
    0 & 1 & 0 & 1 & 1 & 0 & 0\\
  \end{matrix}
\right], 
{\bf G}_{\rm M, eg, 2}^{\bf 2} =  2{\bf G}_{\rm M}^{\bf 1}
= \left[
  \begin{matrix}
    2 & 0 & 2 & 2 & 0 & 0 & 0\\
    0 & 2 & 0 & 2 & 2 & 0 & 0\\
  \end{matrix}
\right],
\end{equation*}
which are two $2 \times 7$ matrices, and form a $2$-user UD-D-CWEA code $\Psi_{\rm cw, eg, 2}$ over GF($3^7$) with $3^2 = 9$ EA codewords.
According to Eq.~(\ref{e.w_k_ai}), where the input user blocks are defined by Eq.~(\ref{e.d_k_2user}), we calculate the corresponding FFSP blocks as follows:
\(
0000000, 0101100, 0202200, 1011000, 1112100, 1210200, \\ 2022000, 2120100, 2221200.
\)
The loading factor of $\Psi_{\rm cw, eg, 2}$ is equal to $\frac{2}{7}$.
$\blacktriangle \blacktriangle$

\vspace{-0.1in}
\subsection{Construction of BD-D-CWEA Codes}

As presented in \cite{FFMA}, the finite field GF($p^m$) is a vector space ${\mathbb V}_p(m)$ over GF($p$) with dimension $m$. Each vector in ${\mathbb V}_p(m)$ is an $m$-tuple over GF($p$). 
For a $J$-dimensional subspace $\mathbb S$ of the vector space ${\mathbb V}_p(m)$ over GF($p$), there exists a basis consisting of $J$ linearly independent vectors that span the subspace $\mathbb S$, where $J \le m$. Let the basis set of the subspace $\mathbb S$ be ${\mathcal B} = \{ \alpha^{l_0}, \alpha^{l_1}, \ldots, \alpha^{l_{J-1}} \}$, satisfying the condition
\begin{equation*}
c_0 \cdot \alpha^{l_0} + c_1 \cdot \alpha^{l_1} + \ldots + c_{J-1} \cdot \alpha^{l_{J-1}} \neq 0,
\end{equation*}
where $\alpha^{l_j}$ is an $m$-tuple in ${\mathbb V}_p(m)$, and $c_j \in$ GF($p$) for $0 \le j < J$.
Clearly, based on the basis set ${\mathcal B}$, we can construct a channel code ${\mathcal C}_{lc}$, with a loading factor equal to $J/m$.

Next, we decompose the basis set $\mathcal{B}$ into $N_d$ disjoint subsets $\mathcal{B}_1, \mathcal{B}_2, \ldots, \mathcal{B}_{N_d}$, where each subset contains $M$ elements with $M = J/N_d$. The $N_d$ subsets are explicitly expressed as:
${\mathcal B}_1 = \{\alpha^{l_{1,1}}, \ldots, \alpha^{l_{j,1}}, \ldots, \alpha^{l_{M,1}}\},$
${\mathcal B}_2 = \{\alpha^{l_{1,2}}, \ldots, \alpha^{l_{j,2}}, \ldots, \alpha^{l_{M,2}}\},$
$\ldots,$
${\mathcal B}_{N_d} = \{\alpha^{l_{1,N_d}}, \ldots, \alpha^{l_{j,N_d}}, \ldots, \alpha^{l_{M,N_d}}\}$,
where ${\mathcal B} = {\mathcal B}_1 \cup {\mathcal B}_2 \cup \ldots \cup {\mathcal B}_{N_d}$.
The $N_d$ subsets ${\mathcal B}_1, {\mathcal B}_2, \ldots, {\mathcal B}_{N_d}$ satisfy the following conditions:
\begin{itemize}
  \item
  The elements in ${\mathcal B}_1, {\mathcal B}_2, \ldots, {\mathcal B}_{N_d}$ are all equal to $M$, i.e.,
  $|{\mathcal B}_1| = |{\mathcal B}_2| = \ldots = |{\mathcal B}_{N_d}| = M$.
  \item
  For $s \neq t$ and $1 \le s, t \le M$, suppose $\alpha^{l_{j,s}}$ is an element in the subset ${\mathcal B}_s$, i.e., $\alpha^{l_{j,s}} \in {\mathcal B}_s$, and $\alpha^{l_{j,t}}$ is an element in the subset ${\mathcal B}_{t}$, i.e., $\alpha^{l_{j,t}} \in {\mathcal B}_t$. Then, we set $\alpha^{l_{j,s}} \neq \alpha^{l_{j,t}}$,
  and ${\mathcal B}_s \cap {\mathcal B}_t = \emptyset$.
\end{itemize}
Based on the subsets ${\mathcal B}_1, {\mathcal B}_2, \ldots, {\mathcal B}_{N_d}$, we can construct $N_d$ independent channel codes ${\mathcal C}_{lc,1}, {\mathcal C}_{lc,2}, \ldots, {\mathcal C}_{lc,N_d}$, whose loading factors are equal to $M/m$.

To construct a $(J/2)$-user D-CWEA code $\Phi_{\rm cw}$, we set $N_d = 2$. Let the full-one generator matrix ${\bf G}_{\rm M}^{\bf 1}$ be the generator matrix of ${\mathcal C}_{lc,1}$, and the full-two generator matrix ${\bf G}_{\rm M}^{\bf 2}$ be the generator matrix of ${\mathcal C}_{lc,2}$. 
Clearly, the constructed D-CWEA code $\Phi_{\rm cw}$ satisfies the USPM constraint, and thus, it is a UD-D-CWEA code.
The aforementioned D-CWEA construction method is referred to as \textit{basis-decomposition (BD)}, which involves decomposing a large basis set into several smaller subsets. The resulting D-CWEA code is also called a \textit{BD-D-CWEA code}, denoted by $\Phi_{\rm bd}$. Based on the original basis set ${\mathcal B}$ and its decomposed subsets ${\mathcal B}_1, {\mathcal B}_2, \ldots, {\mathcal B}_{N_d}$, we can obtain a channel code with a relatively higher loading factor, along with several channel codes that have relatively lower loading factors.

For example, the channel code ${\mathcal C}_{lc}$ with a loading factor of $J/m$ can be decomposed into two codes, ${\mathcal C}_{lc,1}$ and ${\mathcal C}_{lc,2}$, each with a loading factor of $J/(2m)$. 
The basis-decomposing method also serves as the foundation for addressing the finite block length (FBL) issue, as discussed in \cite{FFMA}.

In summarize, we can use the channel codes over GF($2^m$) and/or GF($3^m$) for constructing UD-D-CWEA codes.
Thus, the output FFSP block is either an element in GF($2^m$) or an element in GF($3^m$).


\textbf{Example 5:}
Given a $(7, 4)$ linear cyclic code over GF($2$), whose generator polynomial is $g(X) = 1+X^2+X^3$.
The basis set is given as ${\mathcal B} = \{1 0 1 1 0 0 0, 0 1 0 1 1 0 0, 0 0 1 0 1 1 0, 0 0 0 1 0 1 1\}$, 
which can be divided into two subsets ${\mathcal B}_1$ and ${\mathcal B}_2$,
i.e.,
${\mathcal B}_1 = \{1 0 1 1 0 0 0, 0 1 0 1 1 0 0\}$ and 
${\mathcal B}_2 = \{0 0 1 0 1 1 0, 0 0 0 1 0 1 1\}$.

We construct a $2$-user D-CWEA code ${\Phi}_{\rm cw, eg, 3}$ based on the subsets ${\mathcal B}_1$ and ${\mathcal B}_2$.
The full-one generator matrix ${\bf G}_{\rm M}^{\bf 1}$ and the full-two generator matrix ${\bf G}_{\rm M}^{\bf 2}$ of the $2$-user D-CWEA code ${\Phi}_{\rm cw, eg, 3}$ can be given as
\begin{equation*}
{\bf G}_{\rm M, eg, 3}^{\bf 2} =\left[
  \begin{matrix}
    1 & 0 & 1 & 1 & 0 & 0 & 0\\
    0 & 1 & 0 & 1 & 1 & 0 & 0\\
  \end{matrix}
\right], 
{\bf G}_{\rm M, eg, 3}^{\bf 1} =  
 \left[
  \begin{matrix}
    0 & 0 & 1 & 0 & 1 & 1 & 0\\
    0 & 0 & 0 & 1 & 0 & 1 & 1\\
  \end{matrix}
\right],
\end{equation*}
which are two $2 \times 7$ matrices.
Based on the full-one and full-two generator matrices, we can form a $2$-user D-CWEA code $\Psi_{\rm cw, eg, 3}$ with $3^2 = 9$ EA codewords.

Next, the parallel generator matrix ${\bf G}_{\rm M, pll, eg, 3}$ of $\Phi_{\rm cw, eg, 3}$ is given as
\begin{equation}
{\bf G}_{\rm M, pll, eg, 3} 
= \left[
  \begin{matrix}
  {\bf G}_{\rm M, eg, 3}^{\bf 2} \\
  \hdashline
  {\bf G}_{\rm M, eg, 3}^{\bf 1} \\
  \end{matrix}
  \right]
= \left[
  \begin{matrix}
    1 & 0 & 1 & 1 & 0 & 0 & 0\\
    0 & 1 & 0 & 1 & 1 & 0 & 0\\
    \hdashline
    0 & 0 & 1 & 0 & 1 & 1 & 0\\
    0 & 0 & 0 & 1 & 0 & 1 & 1\\
  \end{matrix}
\right],
\end{equation}
which is a $4 \times 7$ matrix. The rank of ${\bf G}_{\rm M, pll, eg, 3}$ is equal to $4$, which is a full row rank matrix.

To obtain the FFSP blocks, there are two cases determined by the calculation of finite fields, namely GF($2$) and GF($3$).
\begin{itemize}
  \item
  If the FFSP block is calculated over GF($2$), the corresponding FFSP blocks are given as
  $0000000$, $0101100$, $0001011$, $1011000$, $1110100$, $1010011$, $0010110$, $0111010$, 
  $0011101$, whose elements are from GF($2^7$). 
  \item
  If the FFSP block is calculated over GF($3$), the corresponding FFSP blocks are given as
  $0000000$, $0101100$, $0001011$, $1011000$, $1112100$, $1012011$, $0010110$, $0111210$,
  $0011121$, whose elements are from GF($3^7$). 
\end{itemize}
The loading factor of the $2$-user D-CWEA code $\Psi_{\rm cw, eg, 3}$ is equal to $2/7$, which is half of the $(7, 4)$ linear cyclic code.
$\blacktriangle \blacktriangle$

From Example 5, it is shown that the FFSP blocks are determined by the generator matrices. If both generator matrices (${ \bf G}_{\rm M}^{\bf 1}$ and ${ \bf G}_{\rm M}^{\bf 2}$) are derived from the binary field, then the resulting FFSP blocks can belong to either the binary field or the ternary field. However, if one or both of the generator matrices (${ \bf G}_{\rm M}^{\bf 1}$ and ${ \bf G}_{\rm M}^{\bf 2}$) are derived from the ternary field, then the resulting FFSP blocks must belong to the ternary field.

\textbf{Example 6:}
Now, we use the additive inverse matrix of the full-one generator matrix ${\bf G}_{\rm M, eg, 3}^{\bf 1}$, namely $2{\bf G}_{\rm M, eg, 3}^{\bf 1}$, to reconstruct the original full-one generator matrix ${\bf G}_{\rm M, eg, 3}^{\bf 1}$. Using this, we can construct another $2$-user UD-D-CWEA code, denoted as $\Phi_{\rm cw, eg, 4}$, given by:
\begin{equation*}
{\bf G}_{\rm M, eg, 4}^{\bf 2} =
\left[
  \begin{matrix}
    1 & 0 & 1 & 1 & 0 & 0 & 0\\
    0 & 1 & 0 & 1 & 1 & 0 & 0\\
  \end{matrix}
\right],
{\bf G}_{\rm M, eg, 4}^{\bf 1} = 
\left[
  \begin{matrix}
     0 & 0 & 2 & 0 & 2 & 2 & 0\\
     0 & 0 & 0 & 2 & 0 & 2 & 2\\
  \end{matrix}
\right],
\end{equation*}
which can form another $2$-user double UD-CWEA code $\Psi_{\rm cw, eg, 4}$ over GF($3^7$) with $3^2 = 9$ EA codewords.

The parallel generator matrix ${\bf G}_{\rm M, pll, eg, 4}$ of $\Phi_{\rm cw, eg, 4}$ is given as
\begin{equation}
{\bf G}_{\rm M, pll, eg, 4} 
= \left[
  \begin{matrix}
  {\bf G}_{\rm M, eg, 4}^{\bf 2} \\
  \hdashline
  {\bf G}_{\rm M, eg, 4}^{\bf 1} \\
  \end{matrix}
  \right]
= \left[
  \begin{matrix}
    1 & 0 & 1 & 1 & 0 & 0 & 0\\
    0 & 1 & 0 & 1 & 1 & 0 & 0\\
    \hdashline
    0 & 0 & 2 & 0 & 2 & 2 & 0\\
    0 & 0 & 0 & 2 & 0 & 2 & 2\\
  \end{matrix}
\right],
\end{equation}
which is a $4 \times 7$ matrix. The rank of ${\bf G}_{\rm M, pll, eg, 4}$ is equal to $4$, which is a full row rank matrix.
Based on the nine EA codewords, their corresponding FFSP blocks are $0000000$, $0101100$, $0002022$, $1011000$, $1112100$, $1010022$, $0020220$, $0121020$, $0022212$, whose elements are from GF($3^7$). 
$\blacktriangle \blacktriangle$

\vspace{-0.2in}
\subsection{Decoding of D-CWEA codes}

In general, the decoding scheme is determined by the generator matrix set ${\mathcal G}_{\rm M}$ of the proposed CWEA codes. However, the generator matrix set ${\mathcal G}_{\rm M}$ consists of $3^M$ matrices of size $M \times m$, which presents significant challenges for decoding of the CWEA codes. To mitigate the decoding complexity, this paper focuses solely on the decoding of the AI-D-CWEA and BD-D-CWEA codes.

\begin{itemize}
  \item \textbf{Decoding of AI-D-CWEA codes:} The full-one generator matrix ${\bf G}_{\rm M}^{\bf 1}$ can be directly used for decoding the FFSP blocks of AI-D-CWEA codes. Since the FFSP block is a codeword of the full-one generator matrix ${\bf G}_{\rm M}^{\bf 1}$, we have $w_k = {\bf d}[k] \cdot {\bf G}_{\rm M}^{\bf 1}$. If the full-one generator matrix ${\bf G}_{\rm M}^{\bf 1}$ is constructed based on an LDPC code, we can apply the QSPA algorithm to recover the user block ${\bf d}[k]$, which is the codeword of the multiuser code ${\mathcal C}_{mc}$.

  \item \textbf{Decoding of BD-D-CWEA codes:} For decoding BD-D-CWEA codes, the parallel generator matrix ${\bf G}_{\rm M, pll}^{\bf 1}$ can be used. Since the FFSP block is a codeword of the parallel generator matrix ${\bf G}_{\rm M, pll}$, we have $w_k = {\bf a}[k] \cdot {\bf G}_{\rm M, pll}$. Similarly, if the parallel generator matrix ${\bf G}_{\rm M, pll}$ is determined by an LDPC code, the QSPA decoding algorithm can be used to recover the parallel user block ${\bf a}[k]$, which is the codeword of the parallel multiuser code ${\mathcal C}_{mc, \rm pll}$. Afterward, the user block ${\bf d}[k]$ can be obtained using the inverse transform function of ${\rm F}_{\rm T2B}$ as given in Eq.~(\ref{e.T2B}). 
  Note that if $(d_{j,k}^{(2)}, d_{j,k}^{(1)})_2 = (1,1)_2$, it indicates an error has occurred. In this case, we can detect $d_{j,k}$ as $(1)_3$ or $(2)_3$ randomly.
\end{itemize}

\vspace{-0.1in}
\section{From Binary EP-coding to Ternary EA-coding}
For binary source transmission, we introduce four modes of FFMA, which are FF-TDMA, FF-CDMA, FF-CCMA, and FF-NOMA \cite{FFMA2}. These four modes are based on different types of EP codes.
For ternary source transmission, this section introduces EA codes for supporting different modes of FFMA systems.

\vspace{-0.1in}
\subsection{Orthogonal EA Codes for FF-TDMA}

For binary source transmission, orthogonal EP codes \(\Phi_{\rm o, B}\) are employed to support FF-TDMA mode. The EP of the orthogonal EP code \(\Psi_{\rm o, B}\) is defined as 
\(
C_j = \alpha^i \cdot C_{\rm B}, \quad \text{where} \quad C_{\rm B} = (0, 1).
\)
For ternary source transmission, the EA of the orthogonal EA code \(\Phi_{\rm o, T}\) is defined as 
\(
C_j = \alpha^i \cdot C_{\rm T}, \quad \text{where} \quad C_{\rm T} = (0, 1, 2).
\)
Both \(\Psi_{\rm o, B}\) and \(\Phi_{\rm o, T}\) have a loading factor of \(1\). Consequently, the orthogonal EA code \(\Phi_{\rm o, T}\) can also be utilized to support FF-TDMA mode, which will be further investigated in a subsequent paper. 

\vspace{-0.1in}
\subsection{D-CWEA Codes for FF-CCMA}

We have introduced the construction, encoding, and decoding of AI-D-CWEA and BD-D-CWEA codes. These constructed D-CWEA codes can be effectively utilized to support FF-CCMA mode.

An intriguing phenomenon arises in the FF-CCMA mode. For a ternary source, the input digit $d_{j,k}$ belongs to GF($3$). However, when the generator matrices of $\Phi_{\rm cw}$ are constructed over GF($2^m$), the output element $u_{j,k}$ becomes an $m$-tuple over GF($2^m$). This implies that the EA-encoding function ${\rm F}_{p2q}$ transforms a ternary digit into a binary $m$-tuple.

For a given extension factor $m$ and a fixed number of users $J$, increasing the prime factor dimensionality (e.g., from $p = 2$ to $p = 3$) allows for the accommodation of more information. Conversely, decreasing the PF dimensionality (e.g., from $p = 3$ to $p = 2$) reduces the number of users that can be supported, while significantly lowering the design and decoding complexity.

\vspace{-0.1in}
\subsection{Open Issues in Designing CWEA Codes for FF-CDMA}

For binary transmission, the FF-CDMA mode relies on the AI-CWEP code \( \Psi_{\rm ai} \), which is constructed over \( \text{GF}(3^m) \). Although the loading factor of the \( \Psi_{\rm ai} \) code is 1, the number of EP codewords equals \( 2^m \), which is smaller than \( 3^m \). Thus, there exists a coding gain or spreading gain in the AI-CWEP code \( \Psi_{\rm ai} \).

However, for ternary source transmission, the number of EA codewords in the AI-D-CWEA code \( \Phi_{\rm ai} \) equals \( 3^m \). Therefore, there is no further available space for additional codewords. Consequently, for ternary transmission, if the EA code is constructed over \( \text{GF}(3^m) \), there is no FF-CDMA mode. 

One potential solution is to increase the PF dimensionality. For instance, an EA code could be constructed over \( \text{GF}(5^m) \), providing a larger field and enabling the extension of the FF-CDMA mode.

\vspace{-0.1in}
\subsection{NO-CWEA codes for FF-NOMA}

If the \( M \times m \) generator matrix \( {\bf G}_{\rm M}^{{\bf d}[k]} \) of an CWEA code satisfies \( M > m \) and \( m \geq 2 \), then the EA code is referred to as a \textit{non-orthogonal CWEA (NO-CWEA)} code. 
For a \( p \)-ary transmission system, let the \( M \times m \) full-\( \varsigma \) generator matrix of a non-orthogonal CWEA (NO-CWEA) code \( \Phi_{\rm no} \) be denoted as
\begin{equation}
{\bf G}_{\rm M}^{\varsigma} = \left[
    \begin{array}{ccccc}
      g_{1,0}^{(\varsigma)} & g_{1,1}^{(\varsigma)} & \ldots & g_{1,m-1}^{(\varsigma)} \\
      g_{2,0}^{(\varsigma)} & g_{2,1}^{(\varsigma)} & \ldots & g_{2,m-1}^{(\varsigma)} \\
      \vdots                & \vdots                & \ddots & \vdots \\
      g_{M,0}^{(\varsigma)} & g_{M,1}^{(\varsigma)} & \ldots & g_{M,m-1}^{(\varsigma)} \\
    \end{array}
  \right],
\end{equation}
where \( g_{j,i}^{(\varsigma)} \in \text{GF}(\breve{p}) \) for $0 \le \varsigma <p$, \( 1 \le j \le M \) and \( 0 \le i < m \). Here, \( \breve{p} \) can either be equal to \( p \) or differ from \( p \). The subscript ``no'' indicates that the code is non-orthogonal.

Then, an \( M \)-user NO-CWEA code \( \Phi_{\rm no} = \{C_1^{}, C_2^{}, \ldots, C_j^{}, \ldots, C_M^{}\} \) over \( \text{GF}(\breve{q}) \) is obtained, where $\breve{q} = \breve{p}^m$, \(C_j^{} = (\alpha^{l_{j,0}}, \alpha^{l_{j,1}}, ..., \alpha^{l_{j,\varsigma}}, ..., \alpha^{l_{j,p-1}})\) for $1 \le j \le M$.
In this case, \( \alpha^{l_{j,\varsigma}} \) is the \( j \)-th rows of \( {\bf G}_{\rm M}^{\varsigma} \), i.e,
\(\alpha^{l_{j,\varsigma}} = (g_{j,0}^{(\varsigma)}, g_{j,1}^{(\varsigma)}, \ldots, g_{j,m-1}^{(\varsigma)})\).
The loading factor \( \eta \) of the CWEA code \( \Phi_{\rm no} \) is given by \( \eta = M/m \). When \( M > m \), it is evident that the loading factor satisfies \( \eta > 1 \).


It is crucial to design the \textit{finite-field to complex-field transform function} ${\rm F}_{\rm F2C}$, which maps a finite-field symbol to a complex-field signal. By applying this transform function to the \textit{finite-field full-$\varsigma$ generator matrix}, the resulting \textit{complex-field full-$\varsigma$ generator matrix} is denoted as ${\bf S}_{\rm M}^{\varsigma}$, given by
\begin{equation} 
  {\bf S}_{\rm M}^{\varsigma} = {\rm F}_{\rm F2C}({\bf G}_{\rm M}^{\varsigma})
  = \left[
    \begin{array}{cc}
      {\bf s}_{1}^{(\varsigma)}\\
      {\bf s}_{2}^{(\varsigma)}\\
      \vdots\\
      {\bf s}_{M}^{(\varsigma)}\\
    \end{array}
    \right]
  = \left[
    \begin{array}{ccccc}
      s_{1,0}^{(\varsigma)} & s_{1,1}^{(\varsigma)} & \ldots & s_{1,m-1}^{(\varsigma)}\\
      s_{2,0}^{(\varsigma)} & s_{2,1}^{(\varsigma)} & \ldots & s_{2,m-1}^{(\varsigma)}\\
      \vdots        & \vdots        & \ddots & \vdots  \\
      s_{M,0}^{(\varsigma)} & s_{M,1}^{(\varsigma)} & \ldots & s_{M,m-1}^{(\varsigma)}\\
    \end{array}
  \right],
\end{equation}
which is an $M \times m$ matrix, and ${\bf s}_{j}^{(\varsigma)} = {\rm F}_{\rm F2C}(\alpha^{l_{j,\varsigma}})$ for $1 \le j \le M$, $0 \le i < m$ and $0 \le \varsigma < p$.

As presented early, the user block of $M$ users of the $k$-th component is ${\bf d}[k] = (d_{1,k}, d_{2,k}, \ldots, d_{j,k}, ..., d_{M,k})$. 
Consider the relationship between $\alpha^{l_{j,\varsigma}}$ and ${\bf s}_{j}^{(\varsigma)}$,
it is able to know 
\begin{equation} 
  {\bf s}_{j}^{(d_{j,k})} \triangleq {\rm F}_{\rm F2C} \left({\mathrm F}_{p2q}(d_{j,k})\right)  = {\rm F}_{\rm F2C}(\alpha^{l_{j,\varsigma}}), 
\end{equation}
indicating there is a mapping between the input digit $d_{j,k}$ and the output signal ${\bf s}_{j}^{(d_{j,k})}$. Hence, we can calculate the complex-field sum-pattern (CFSP) block ${\bf r}[k]$ of the user block ${\bf d}[k]$ as
\begin{equation}
  {\bf r}[k] = \sum_{j=1}^{M} {\bf s}_{j}^{(d_{j,k})}.
\end{equation}
If there exists a one-to-one mapping between the user block ${\bf d}[k]$ and CFSP block ${\bf r}[k]$, the user block can be recovered without ambiguity. This result is summarized by the following theorem.

\begin{theorem} \label{theorem.USPM_CF}
(\textbf{USPM Constraint in Complex Field})
For \( p \)-ary source transmission, let the finite-field and complex-field generator matrices of an NO-CWEA code \( \Phi_{\rm no} \) be denoted as \( {\bf G}_{\rm M}^{\varsigma} \) and \( {\bf S}_{\rm M}^{\varsigma} \), respectively, where \( 0 \le \varsigma < p \). Here, \( {\bf S}_{\rm M}^{\varsigma} = {\rm F}_{\rm F2C}({\bf G}_{\rm M}^{\varsigma}) \), and \( {\rm F}_{\rm F2C} \) represents the finite field to the complex field transform function. To ensure the unique decoding of the user block \( {\bf d}[k] \), there must exist a one-to-one mapping between the input user block \( {\bf d}[k] \) and the complex-field sum-pattern (CFSP) block \( {\bf r}[k] \), i.e., \( {\bf d}[k] \leftrightarrow {\bf r}[k] \) for \( 0 \le k < K \). This indicates that the system possesses the unique sum-pattern mapping (USPM) structural property in the complex field.
\end{theorem}

Based on Theorem \ref{theorem.USPM_CF}, we derive the following key features:
\begin{itemize}
  \item 
  A \( p \)-ary NOMA system can be directly constructed by designing \( p \) complex-field generator matrices \( {\bf S}_{\rm M}^{\bf 0}, {\bf S}_{\rm M}^{\bf 1}, \ldots, {\bf S}_{\rm M}^{\bf p-1} \), satisfying the USPM constraint in the complex field. Each generator matrix is an \( M \times m \) matrix with \( M > m \).
  \item 
  The \( p \) finite-field generator matrices \( {\bf G}_{\rm M}^{\bf 0}, {\bf G}_{\rm M}^{\bf 1}, \ldots, {\bf G}_{\rm M}^{\bf p-1} \) can be constructed over GF(\( \breve{p} \)), where \( \breve{p} \) may either equal \( p \) or differ from \( p \).
  \item 
  The finite-field to complex-field transform function \( {\rm F}_{\rm F2C} \) can be flexibly designed. A finite-field symbol can be mapped to a complex-field signal, which may be either a real number or an imaginary number.
\end{itemize}

\textbf{Example 7:}  
We construct a ternary NO-CWEA code $\Phi_{\rm no}$ over GF($5^2$), which supports 3 users with a codeword length of 2. The full-zero generator matrix is defined as the zero matrix ${\bf G}_{\rm M}^{\bf 0} = {\bf 0}$, while the full-one generator matrix ${\bf G}_{\rm M}^{\bf 1}$ and the full-two generator matrix ${\bf G}_{\rm M}^{\bf 2}$ of the NO-CWEA code $\Phi_{\rm no}$ are 
\begin{equation} \label{e.T_no}
{\bf G}_{\rm M}^{\bf 1} = 
\left[
  \begin{array}{cc}
    1 & 1\\
    4 & 1\\
    0 & 1\\
  \end{array}
\right],
\quad
{\bf G}_{\rm M}^{\bf 2} = 
\left[
  \begin{array}{cc}
    4 & 4\\
    1 & 4\\
    2 & 4\\
  \end{array}
\right].
\end{equation}
From these generator matrices, we obtain the \( 3 \)-user NO-CWEA code
\[
\Phi_{\rm no} = \left\{ C_1^{\rm no} = (00, 11, 44), \, C_2^{\rm no} = (00, 41, 14), \, C_3^{\rm no} = (00, 01, 24) \right\}
\]
over GF(\( 5^2 \)).
Next, we define the finite-field to complex-field transform function \( {\rm F}_{\rm F2C} \) as follows:
\[
+1 = {\rm F}_{\rm F2C}(1), \quad -1 = {\rm F}_{\rm F2C}(4), \quad 0 = {\rm F}_{\rm F2C}(0), \quad \text{and} \quad +1i = {\rm F}_{\rm F2C}(2),
\]
where \( 1i \) denotes the imaginary unit. Using this transform, the complex-field generator matrices \( {\bf S}_{\rm M}^{\bf 1} \) and \( {\bf S}_{\rm M}^{\bf 2} \) are derived as
\vspace{-0.1in}
\begin{equation} \label{e.NO_CWEA_CF_GF5}
{\bf S}_{\rm M}^{\bf 1} = {\rm F}_{\rm F2C}({\bf G}_{\rm M}^{\bf 1}) = 
\left[
  \begin{array}{cc}
    +1 & +1\\
    -1 & +1\\
    0 & +1\\
  \end{array}
\right],
\quad
{\bf S}_{\rm M}^{\bf 2} = {\rm F}_{\rm F2C}({\bf G}_{\rm M}^{\bf 2}) = 
\left[
  \begin{array}{cc}
    -1 & -1\\
    +1 & -1\\
    \textcolor{blue}{1i} & -1\\
  \end{array}
\right].
\end{equation}
It can be proven that, based on (\ref{e.NO_CWEA_CF_GF5}), there exists a one-to-one mapping between the CFSP signal block \( {\bf r}[k] \) and the user block \( {\bf d}[k] \).  
$\blacktriangle \blacktriangle$

\section{Construction of $p$-ary CWEA Codes via Basis Decomposition}

In the previous sections, our primary focus has been on $3$-ary CWEA codes, where the input digits are defined over GF($3$). In this section, we extend our discussion to the construction of $p$-ary CWEA codes using the basis decomposition method, where $p$ is a prime number greater than $3$.

As discussed in Section II, an $M$-user $p$-ary EA code is represented as $\Phi_{\rm cw} = \{C_1, C_2, \ldots, C_j, \ldots, C_M\}$, where the $j$-th EA is given by 
\(
C_j = (\alpha^{l_{j,0}}, \alpha^{l_{j,1}}, \ldots, \alpha^{l_{j,\varsigma}}, \ldots, \alpha^{l_{j,p-1}})
\)
for $0 \le \varsigma < p$ and $1 \le j \le M$. If each element $\alpha^{l_{j,\varsigma}}$ of the EA $C_j$ is a codeword generated by a generator matrix ${\bf G}_{\rm M}^{\varsigma}$, the code is referred to as a $p$-ary CWEA code.
To construct a $p$-ary CWEA code, $p$ generator matrices must be designed, one for each digit $\varsigma$, where $\varsigma$ maps uniquely to its corresponding generator matrix ${\bf G}_{\rm M}^{\varsigma}$. These generator matrices, denoted as ${\bf G}_{\rm M}^{\bf 0}, {\bf G}_{\rm M}^{\bf 1}, \ldots, {\bf G}_{\rm M}^{\bf p-1}$, must satisfy the USPM constraint. However, designing such matrices is a highly challenging task, particularly for large values of $p$. To address this issue, we propose the use of basis decomposition for constructing $p$-ary CWEA codes, as demonstrated in the construction of a $3$-ary CWEA code. The BD method leverages low-dimensional structures to achieve high-dimensional transmission. Specifically, if a large integer $p$ can be decomposed into a sum of smaller integers, the dimensionality of the problem is reduced, thereby simplifying the design of the corresponding generator matrices ${\bf G}_{\rm M}^{\bf 0}, {\bf G}_{\rm M}^{\bf 1}, \ldots, {\bf G}_{\rm M}^{\bf p-1}$.

It is well-known that any integer \( p > 3 \) can be expressed as a linear combination of \( 2 \)s and \( 3 \)s. Consequently, binary (base \( 2 \)) and ternary (base \( 3 \)) systems form the foundational bases for integer representation. While binary decomposition has been extensively studied in Section IV, our focus here is on \textit{ternary decomposition}, which offers a faster convergence rate.

To construct a \( p \)-ary EA code using the BD method, we first determine the number of subsets \( N_d \). For a given integer \( p \), the number of subsets \( N_d \) is set as \( N_d = \lceil \log_3(p) \rceil \). The basis of dimension \( J \) is then divided into \( N_d \) subsets, each with a dimension of \( M = J / N_d \). With \( N_d \) determined, both the digits and generator matrices can be processed in a low-dimensional framework.

For the digits in such a \( p \)-ary system, let \( d_{j,k} \) denote the \( p \)-ary digit of the \( j \)-th user. The digit \( d_{j,k} \) is transformed into its ternary form using the \( p \)-ary to ternary transform function \( {\rm F}_{p2{\rm T}} \), denoted as
\[
\left(d_{j,k}^{(N_d)}, \ldots, d_{j,k}^{(1)} \right)_3 = {\rm F}_{p2{\rm T}}(d_{j,k}),
\]
where \( d_{j,k} = \sum_{n_d=1}^{N_d} d_{j,k}^{(n_d)} \cdot 3^{n_d-1} \), with \( d_{j,k}^{(n_d)} \in \mathbb{T} \). The transform function \( {\rm F}_{p2{\rm T}} \) performs an integer decomposition in ternary form. The superscripts \((1), (2), \ldots, (N_d)\) represent the integer-decomposing positions from \( 1 \) to \( N_d \).

For the generator matrices, the required \( p \) generator matrices \( {\bf G}_{\rm M}^{\bf 0}, {\bf G}_{\rm M}^{\bf 1}, \ldots, {\bf G}_{\rm M}^{\bf p-1} \) can be reduced to \( N_d \) generator matrices \( {\bf G}_{\rm M}^{(0)}, {\bf G}_{\rm M}^{(1)}, \ldots, {\bf G}_{\rm M}^{(N_d)} \). For example,
\[
\underbrace{{\bf G}_{\rm M}^{\bf 0}, {\bf G}_{\rm M}^{\bf 1}, {\bf G}_{\rm M}^{\bf 2}}_{ \textcolor{blue}{{\bf G}_{\rm M}^{(0)}}}, 
\underbrace{{\bf G}_{\rm M}^{\bf 3}, {\bf G}_{\rm M}^{\bf 4}, {\bf G}_{\rm M}^{\bf 5}}_{ \textcolor{blue}{{\bf G}_{\rm M}^{(1)}}}, 
\ldots, 
\underbrace{..., {\bf G}_{\rm M}^{\bf p-1}}_{ \textcolor{blue}{{\bf G}_{\rm M}^{(N_d)}}}.
\]
This reduction is enabled by the property of the \textit{additive inverse calculation} in the ternary system: for any \( a \in \mathbb{T} \), the relationship
\(
a_{\rm ai} = 3 - a = 2 \cdot a
\)
holds. Consequently, leveraging the properties of the generator matrix in the AI-D-CWEA code, every three generator matrices can be replaced by a single unified generator matrix. For instance, \( {\bf G}_{\rm M}^{\bf 0}, {\bf G}_{\rm M}^{\bf 1}, {\bf G}_{\rm M}^{\bf 2} \) are replaced by \( {\bf G}_{\rm M}^{(0)} \).

Next, according to the BD method, a \( J \)-dimensional basis set \( \mathcal{B} \) in \( \mathbb{V}_3(m) \) is first divided into \( N_d \) subsets, denoted as \( \mathcal{B}_1, \mathcal{B}_2, \ldots, \mathcal{B}_{N_d} \), where \( \mathcal{B} = \mathcal{B}_1 \cup \mathcal{B}_2 \cup \ldots \cup \mathcal{B}_{N_d} \). Based on these \( N_d \) subsets, we can construct \( N_d \) generator matrices, denoted as \( {\bf G}_{\rm M}^{(1)}, \ldots, {\bf G}_{\rm M}^{(N_d)} \). In this case, the parallel generator matrix \( {\bf G}_{\rm M, pll} \) is defined as
\[
{\bf G}_{\rm M, pll} = 
\left[
  {\bf G}_{\rm M}^{(N_d)}, \ldots, {\bf G}_{\rm M}^{(1)}
\right]^{\rm T},
\]
which is a \( J \times m \) matrix over GF(\( 3 \)). The dimensionality of the parallel generator matrix \( {\bf G}_{\rm M, pll} \) is reduced from \( p \) to \( 3 \), resulting in a low-dimensional structure.

Suppose there are \( M \) users. We define the parallel user block \( {\bf a}[k] \) as

\vspace{-0.15in}
\begin{small}
  \[
{\bf a}[k] = {\rm F}_{p2{\rm B}}({\bf d}[k]) = 
\left(
\underbrace{d_{1,k}^{(N_d)}, d_{2,k}^{(N_d)}, \ldots, d_{j,k}^{(N_d)}, \ldots, d_{M,k}^{(N_d)}}_{M}, 
\ldots,
\underbrace{d_{1,k}^{(1)}, d_{2,k}^{(1)}, \ldots, d_{j,k}^{(1)}, \ldots, d_{M,k}^{(1)}}_{M}  
\right), 
\]
\end{small}
which is a ternary vector of length \( M \cdot N_d \).

The FFSP block is obtained by encoding the parallel user block \( {\bf a}[k] \) using the parallel generator matrix \( {\bf G}_{\rm M, pll} \), yielding
\begin{equation} \label{e.W_k_pll}
  w_k = {\bf a}[k] \cdot {\bf G}_{\rm M, pll},
\end{equation}
where \( w_k \) is an \( m \)-tuple over GF(\( 3 \)), representing the codeword generated by the parallel generator matrix \( {\bf G}_{\rm M, pll} \). By leveraging the parallel generator matrix \( {\bf G}_{\rm M, pll} \) over GF(\( 3 \)), we can efficiently encode and decode the \( p \)-ary BD-CWEA codes in ternary form.

\textbf{Example 8:}
For a $5$-ary source transmission system, we construct a $2$-user $5$-ary CWEA code $\Phi_{\rm cw}$.
Since $p = 5$, it is derived that $N_d = \lceil \log_3(p) \rceil = 2$.
It means each digit $d_{j,k}$ of the $j$-th user can be expressed by two ternary digits, 
i.e., $\left(d_{j,k}^{(2)}, d_{j,k}^{(1)}\right)_3$.
Hence, the transform function ${\rm F}_{p2{\rm T}}$ is given as

\vspace{-0.1in}
\begin{small}
  \begin{equation} 
  {\rm F}_{p2{\rm T}}(d_{j,k}) = \left(d_{j,k}^{(2)}, d_{j,k}^{(1)}\right)_3 = 
  \left\{
    \begin{array}{cc}
      (0, 0)_3, & d_{j,k} = (0)_5\\
      (0, 1)_3, & d_{j,k} = (1)_5\\
      (0, 2)_3, & d_{j,k} = (2)_5\\
      (1, 0)_3, & d_{j,k} = (3)_5\\
      (1, 1)_3, & d_{j,k} = (4)_5\\
    \end{array},
  \right.
\end{equation}
\end{small}

Regarding as $J = 2$, we can list all the possible combinations of the input user block ${\bf d}[k]$ and its corresponding parallel user block ${\bf a}[k]$, given as
\begin{equation*} 
  \begin{array}{cc}
    {\bf d}[k] \in \{0\textcolor{red}{0}, 0\textcolor{red}{1}, 0\textcolor{red}{2}, 
                     0\textcolor{red}{3}, 0\textcolor{red}{4}\} \Rightarrow 
    {\bf a}[k] \in \{0\textcolor{red}{0}0\textcolor{red}{0}, 0\textcolor{red}{0}0\textcolor{red}{1}, 
                     0\textcolor{red}{0}0\textcolor{red}{2}, 0\textcolor{red}{1}0\textcolor{red}{0}, 
                     0\textcolor{red}{1}0\textcolor{red}{1}\},\\
    {\bf d}[k] \in \{1\textcolor{red}{0}, 1\textcolor{red}{1}, 1\textcolor{red}{2}, 
                     1\textcolor{red}{3}, 1\textcolor{red}{4}\} \Rightarrow 
    {\bf a}[k] \in \{0\textcolor{red}{0}1\textcolor{red}{0}, 0\textcolor{red}{0}1\textcolor{red}{1}, 
                     0\textcolor{red}{0}1\textcolor{red}{2}, 0\textcolor{red}{1}1\textcolor{red}{0}, 
                     0\textcolor{red}{1}1\textcolor{red}{1}\},\\
    {\bf d}[k] \in \{2\textcolor{red}{0}, 2\textcolor{red}{1}, 2\textcolor{red}{2}, 
                     2\textcolor{red}{3}, 2\textcolor{red}{4}\} \Rightarrow 
    {\bf a}[k] \in \{0\textcolor{red}{0}2\textcolor{red}{0}, 0\textcolor{red}{0}2\textcolor{red}{1}, 
                     0\textcolor{red}{0}2\textcolor{red}{2}, 0\textcolor{red}{1}2\textcolor{red}{0}, 
                     0\textcolor{red}{1}2\textcolor{red}{1}\},\\
    {\bf d}[k] \in \{3\textcolor{red}{0}, 3\textcolor{red}{1}, 3\textcolor{red}{2}, 
                     3\textcolor{red}{3}, 3\textcolor{red}{4}\} \Rightarrow 
    {\bf a}[k] \in \{1\textcolor{red}{0}0\textcolor{red}{0}, 1\textcolor{red}{0}0\textcolor{red}{1}, 
                     1\textcolor{red}{0}0\textcolor{red}{2}, 1\textcolor{red}{1}0\textcolor{red}{0}, 
                     1\textcolor{red}{1}0\textcolor{red}{1}\},\\
    {\bf d}[k] \in \{4\textcolor{red}{0}, 4\textcolor{red}{1}, 4\textcolor{red}{2}, 
                     4\textcolor{red}{3}, 4\textcolor{red}{4}\} \Rightarrow 
    {\bf a}[k] \in \{1\textcolor{red}{0}1\textcolor{red}{0}, 1\textcolor{red}{0}1\textcolor{red}{1}, 
                     1\textcolor{red}{0}1\textcolor{red}{2}, 1\textcolor{red}{1}1\textcolor{red}{0}, 
                     1\textcolor{red}{1}1\textcolor{red}{1}\}.\\

  \end{array}
\end{equation*}

We still use the two subsets ${\mathcal B}_1$ and ${\mathcal B}_2$ given by Example 5. Then, the parallel generator matrix ${\bf G}_{\rm M, pll}$ of the $2$-user $5$-ary CWEA code $\Phi_{cw}$ is given as
\begin{equation}
{\bf G}_{\rm M, pll} 
= \left[
  \begin{matrix}
  {\bf G}_{\rm M}^{(2)} \\
  \hdashline
  {\bf G}_{\rm M}^{(1)} \\
  \end{matrix}
  \right]
= \left[
  \begin{matrix}
    1 & 0 & 1 & 1 & 0 & 0 & 0\\
    0 & 1 & 0 & 1 & 1 & 0 & 0\\
    \hdashline
    0 & 0 & 1 & 0 & 1 & 1 & 0\\
    0 & 0 & 0 & 1 & 0 & 1 & 1\\
  \end{matrix}
\right].
\end{equation}
Since $p = 5$ and $J = 2$, there are totally $5^2 = 25$ FFSP blocks (or codewords of the parallel generator matrix ${\bf G}_{\rm M, pll}$).
By using (\ref{e.W_k_pll}), the FFSP blocks are given as following:
$0000000$, $0001011$, 
$0002022$, $0101100$,
$0102111$, $0010110$,
$0011121$, $0012102$,
$0111210$, $0112221$,
$0020220$, $0021201$,
$0022212$, $0121020$,
$0122001$, $1011000$,
$1012011$, $1010022$,
$1112100$, $1110111$,
$1021110$, $1022121$,
$1020102$, $1122210$,
$1120221$,
which are $25$ $7$-tuples over GF($3$).
$\blacktriangle \blacktriangle$

\section{Loading Factor and Multirate Sequence}

In this section, we begin by outlining some basic assumptions and preliminaries. Next, we redefine the concepts of the loading factor and coding rate, which apply to both multiuser codes (or UD-CWEA codes) and general channel codes, in order to unify these two types of codes. We then explore the rate characteristics of the transmitted encoded sequence. 

\vspace{-0.1in}
\subsection{Preliminaries}

Let the digit sequence of the $j$-th user be denoted by ${\bf d}_j \in \mathbb{P}^{1 \times K}$, where $K$ represents the length of the sequence. This means that the $j$-th user occupies $K$ degrees of freedom (DoFs) for transmitting information symbols.

We assume that the EA code is constructed based on an $(m, M)$ linear block code ${\mathcal C}_{mc}$ over GF($p$), with its generator matrix in systematic form denoted as ${\bf G}_{mc}$, an $M \times m$ matrix. Suppose the EA encoder operates in parallel mode to encode the digit sequence of each user. We assume that $K \le M$, so we append a zero vector of length $M - K$, denoted as ${\bf 0}$, to the digit sequence ${\bf d}_j$. This results in the $1 \times M$ information vector ${\bf u}_{j,D} = ({\bf d}_j, {\bf 0}) \in \mathbb{P}^{1 \times M}$.
The encoded codeword is then given by
\(
{\bf c}_j = {\bf u}_{j,D} \cdot {\bf G}_{mc} = ({\bf u}_{j,D}, {\bf c}_{j,\rm red}),
\)
which is an $m$-tuple over GF($p^m$). Here, ${\bf c}_{j,\rm red}$ represents the parity vector of ${\mathcal C}_{mc}$, and its length is $Q = m - M$, meaning that the $j$-th user occupies $Q$ DoFs for transmitting parity symbols. In this case, the encoded codeword ${\bf c}_j$ is the output element of the EA encoder in parallel mode. If the output element ${\bf c}_j$ is directly transmitted to the AWGNC, the FFMA system operates in either FF-TDMA mode or FF-CCMA without channel coding mode.

The corresponding parameters are summarized as follows:

\begin{small}
\begin{equation} \label{e.assumptions}
  \begin{array}{ll}
    \text{Digit sequence:} & {\bf d}_j \in {\mathbb P}^{1 \times K}, \\
    \hdashline
    \text{Information vector of ${\mathcal C}_{mc}$:} & {\bf u}_{j,D} = ({\bf d}_j, {\bf 0}) \in {\mathbb P}^{1 \times M}, \\
    \text{Parity vector of ${\mathcal C}_{mc}$:} & {\bf c}_{j,\rm red} \in {\mathbb P}^{1 \times Q}, \\
    \text{Codeword of ${\mathcal C}_{mc}$:} & {\bf c}_j = {\bf u}_{j,D} \cdot {\bf G}_{mc} = ({\bf u}_{j,D}, {\bf c}_{j,\rm red}) \in {\mathbb P}^{1 \times m},  \\
    \text{Codeword of ${\mathcal C}_{mc}$:} & {w} = \sum_{j=1}^{J} {\bf c}_j \in {\mathbb P}^{1 \times m},  \\
    \text{Polarization adjusted vector:} & {\mu}_{\rm reg}^{\rm td} = (\mu_{1}, \mu_{2}),  \\
    \text{Polarization adjusted scaling:} & {\mu}_{\rm pas} = \frac{\mu_{1}}{\mu_{2}},  \\
    \hdashline
    \text{Information vector of ${\mathcal C}_{gc}$:} & {\bf c}_{j,D} = ({\bf c}_j, {\bf 0}) \in {\mathbb P}^{1 \times K_{gc}}, \\
    \text{Parity vector of ${\mathcal C}_{gc}$:} & {\bf v}_{j,\rm red} \in {\mathbb P}^{1 \times R}, \\
    \text{Codeword of ${\mathcal C}_{gc}$:} & {\bf v}_j = {\bf c}_{j,D} \cdot {\bf G}_{mc} = ({\bf c}_{j,D}, {\bf v}_{j,\rm red}) \in {\mathbb P}^{1 \times N},  \\
    \text{Polarization adjusted vector:} & {\mu}_{\rm reg}^{\rm cc} = (\mu_{1}, \mu_{2}, \mu_{c}).  \\
  \end{array}
\end{equation}
\end{small}

If the EA code is concatenated with a channel code ${\mathcal C}_{gc}$, the FFMA system operates in the FF-CCMA with channel coding mode. Let ${\mathcal C}_{gc}$ be an $(N, K_{gc})$ linear block code over GF($p$). Assume that $m \le K_{gc}$, and similarly to the EA encoding process, we append a zero vector of length $K_{gc} - m$, denoted as ${\bf 0}$, to the vector ${\bf c}_j$, thereby obtaining the information vector ${\bf c}_{j,D} = ({\bf c}_j, {\bf 0}) \in \mathbb{P}^{1 \times K_{gc}}$.
The vector ${\bf c}_{j,D}$ is then encoded using the systematic generator matrix of the channel code ${\bf G}_{gc}$. The resulting encoded codeword is given by:
\(
{\bf v}_j = {\bf c}_{j,D} \cdot {\bf G}_{gc} = ({\bf c}_{j,D}, {\bf v}_{j,\rm red}),
\)
which is a $1 \times N$ vector over GF($p$). Here, ${\bf v}_{j,\rm red}$ represents the parity vector of the channel code ${\mathcal C}_{gc}$, with length $R = N - K_{gc}$, indicating that the $j$-th user occupies an additional $R$ DoFs for transmitting parity symbols.

In the FF-TDMA mode, the polarization-adjusted vector (PAV) in its regular form (or regular PAV) is expressed as $\mu_{\rm reg}^{\rm td} = (\mu_{1}, \mu_{2})$, as defined in \cite{FFMA2}. The polarization-adjusted scaling (PAS) is defined as $\mu_{\rm pas} = \frac{\mu_{1}}{\mu_{2}}$, as detailed in \cite{FFMA}. 
For the FF-CCMA mode, the regular form of the PAV is expressed as
$\mu_{\rm reg}^{\rm cc} = (\mu_{1}, \mu_{2}, \mu_{c})$ \cite{FFMA2}. 

In the FF-TDMA mode or the FF-CCMA without channel coding mode, we examine the output element ${\bf c}_j \in {\mathbb P}^{1 \times m}$, where ${\bf c}_j = ({\bf d}_j, {\bf 0}, {\bf c}_{j,\rm red})$. The transmit vector consists of three distinct sub-sequences: ${\bf d}_j \in {\mathbb P}^{1 \times K}$, ${\bf 0} \in {\mathbb P}^{1 \times (m-K-Q)}$, and ${\bf c}_{j,\rm red} \in {\mathbb P}^{1 \times Q}$. It is clear that the effect of the zero vector ${\bf 0}$ can be disregarded. Therefore, the power allocation scheme focuses on redistributing the remaining power $(m-K-Q) \cdot P_{avg}$ to optimize the performance metric. Additionally, since the information part ${\bf d}_j$ and the parity part ${\bf c}_{j,\rm red}$ typically operate at different rates, the transmit element ${\bf c}_j$ can be regarded as a \textit{multirate sequence}. 
It is important to note that the codeword ${\bf c}_j$ is the sum of $K$ codewords (or elements) from ${\mathcal C}_{mc}$, which can also be interpreted as an FFSP block. Furthermore, the FFSP block $w$ for a $J$-user system is the sum of the codewords ${\bf c}_1, {\bf c}_2, \dots, {\bf c}_J$, i.e., ${w} = \sum_{j=1}^{J} {\bf c}_j$, which is also a codeword of ${\mathcal C}_{mc}$. Specifically, when $J = 1$, we have $w = {\bf c}_j$.

Similarly, in the FF-CCMA with channel coding mode, the codeword ${\bf v}_j \in {\mathbb P}^{1 \times N}$, where ${\bf v}_j = ({\bf d}_j, {\bf 0}, {\bf c}_{j,\rm red}, {\bf 0}, {\bf v}_{j,\rm red})$. The transmit vector now consists of five sub-sequences: ${\bf d}_j \in {\mathbb P}^{1 \times K}$, ${\bf 0} \in {\mathbb P}^{1 \times (m-K-Q)}$, ${\bf c}_{j,\rm red} \in {\mathbb P}^{1 \times Q}$, ${\bf 0} \in {\mathbb P}^{1 \times (K_{gc}-m)}$, and ${\bf v}_{j,\rm red} \in {\mathbb P}^{1 \times R}$. Hence, the total length of the zero vector ${\bf 0}$ is equal to $K_{gc} - K - Q$, which implies that the remaining power to be allocated is $(K_{gc}-K-Q) \cdot P_{avg}$. Clearly, the transmit codeword ${\bf v}_j$ also represents a multirate sequence.

\vspace{-0.15in}
\subsection{Multirate Sequence}

For a \( p \)-ary source transmission system, assume that both the multiuser codes \( {\mathcal C}_{mc} \) (or UD-EA codes) and channel codes are constructed over \( \text{GF}(p) \). Since a multiuser code can be derived from a channel code, we can use the multiuser code as an example to analyze the loading factor and coding rate.

\vspace{-0.05in}
\begin{definition} (\textbf{Loading Factor})
  The loading factor of a multiuser code \( {\mathcal C}_{mc} \) is defined as the ratio of the number of occupied degrees of freedom (DoFs), denoted by \( M \), to the total number of available DoFs, denoted by \( m \), in the extension axis (also referred to as the E-axis). This is expressed as:
  \begin{equation}
    \eta_{mc} = \frac{M}{m},
  \end{equation}
  where \( {\mathcal C}_{mc} \) is constructed based on an \( (m, M) \) linear block code over \( \text{GF}(p) \).
\end{definition}

\begin{definition} (\textbf{Coding Rate})
  For a multiuser code \( {\mathcal C}_{mc} \) constructed over \( \text{GF}(p) \), if it occupies \( m \) degrees of freedom (DoFs) in the E-axis to transmit \( p^M \) codewords, the coding rate of \( {\mathcal C}_{mc} \) is defined as:
  \[
  R_{q, mc} = \frac{\log_2 (p^M)}{m} = \frac{M}{m} \cdot \log_2{p}
  = \eta_{mc} \cdot \log_2{p},
  \]
  where \( \eta_{mc} \) is the loading factor of the multiuser code \( {\mathcal C}_{mc} \).
\end{definition}

For an \( (N, K_{gc}) \) channel code \( {\mathcal C}_{gc} \) constructed over GF($p$), the codeword length \( N \) and information length \( K_{gc} \) correspond to the E-axis of a finite field. Therefore, the loading factor is given by
\(
\eta_{gc} = \frac{K_{gc}}{N}.
\)
Similarly, the coding rate of the channel code \( {\mathcal C}_{gc} \) is 
\(
R_{q, gc} = \eta_{gc} \cdot \log_2{p}.
\)
If the channel code \( {\mathcal C}_{gc} \) is constructed over \( \text{GF}(2) \), which corresponds to a binary channel code (i.e., when \( p = 2 \)), the coding rate equals the loading factor, i.e., $R_{q, gc} = \eta_{gc}$.

Based on these definitions, we proceed to investigate the rate characteristics of the sequence \( {\bf c}_j \) in the FF-TDMA mode as an example. This analysis can be easily extended to the sequence \( {\bf v}_j \) in the FF-CCMA mode. Therefore, the sequence \( {\bf v}_j \) will not be repeated anymore.

In the FF-TDMA mode, the transmit vector ${\bf c}_j$ consists of the information vector ${\bf d}_j \in {\mathbb P}^{1 \times K}$ and the parity vector ${\bf c}_{j,\rm red} \in {\mathbb P}^{1 \times Q}$. 
The coding rate of the information vector ${\bf d}_{j} \in {\mathbb P}^{1 \times K}$ of the $j$-th user, defined by $R_{j, 1}$, is equal to 
  \begin{equation} \label{e.R_j1}
    R_{j, 1} = \frac{\log_2 (p^K)}{K} = \log_2{p},
  \end{equation}
where $\eta_{j,1} = 1$ is the loading factor of the information vector.

The parity vector \( {\bf c}_{j, \rm red} \) is determined by both the information vector \( {\bf d}_j \) and the generator matrix \( {\bf G}_{mc} \) of \( {\mathcal C}_{gc} \), i.e., \( {\bf c}_{j, \rm red} = {\bf d}_j \cdot {\bf F}_{\rm red} \), where \( {\bf F}_{\rm red} \) is a \( K \times Q \) linearly independent submatrix of the generator matrix \( {\bf G}_{mc} \), as defined in \cite{FFMA2}. The coding rate of the parity vector \( {\bf c}_{j, \rm red} \) can be summarized in the following Lemma \ref{lemma.SU_parity}.

\begin{lemma} \label{lemma.SU_parity}
  Suppose \( {\bf d}_j \) is a \( 1 \times K \) input vector over GF(\( p \)), and \( {\bf F}_{\rm red} \) is a \( K \times Q \) linearly independent matrix. The output vector \( {\bf c}_{j, \rm red} \) is given by \( {\bf c}_{j, \rm red} = {\bf d}_j \cdot {\bf F}_{\rm red} \). Thus, the coding rate of the output vector \( {\bf c}_{j, \rm red} \), which is a \( 1 \times Q \) vector over GF(\( p \)), is given by
  \begin{equation} \label{e.R_j2}
    R_{j, 2} = \frac{\log_2 |{\bf c}_{j, \rm red}|}{Q} 
    = \frac{\min\{K, Q\}}{Q} \cdot (\log_2 p) 
    = \eta_{j,2} \cdot \log_2 p 
  \end{equation}
  where \( |{\bf c}_{j, \rm red}| \) denotes the number of codewords of \( {\bf c}_{j, \rm red} \), and \( \eta_{j,2} = \frac{\min\{K, Q\}}{Q} \) is the loading factor of \( {\bf c}_{j, \rm red} \). Considering the relationship between \( K \) and \( Q \), there are two cases:
  \begin{enumerate}
    \item If \( K < Q \), the coding rate is \( R_{j, 2} = \frac{K}{Q} \cdot \log_2 p \).
    \item If \( K \ge Q \), the coding rate is \( R_{j, 2} = \log_2 p \).
  \end{enumerate}
\end{lemma}

\begin{proof}
  Since \( {\bf c}_{j, \rm red} = {\bf d}_j \cdot {\bf F}_{\rm red} \), the number of codewords of \( {\bf c}_{j, \rm red} \) is determined by the rank of the \( K \times Q \) matrix \( {\bf F}_{\rm red} \). Considering that each digit is from the finite field GF(\( p \)), the number of codewords of \( {\bf c}_{j, \rm red} \) is given by
  \[
    |{\bf c}_{j, \rm red}| = p^{{\rm rank}({\bf F}_{\rm red})} = p^{\min\{K, Q\}}.
  \]
  Substituting this into the definition of coding rate, we obtain Eq. (\ref{e.R_j2}).
\end{proof}

Similarly to the analysis of the coding rate for ${\bf c}_{j, \rm red}$, we can deduce the coding rate for the parity vector ${\bf v}_{j, \rm red}$ of the channel code ${\mathcal C}_{gc}$. The parity vector ${\bf v}_{j,\rm red}$ is solely determined by the input information vector ${\bf d}_j$, and ${\bf v}_{j,\rm red}$ is a \( 1 \times R \) vector over GF(\( p \)), where \( R = N - K_{gc} \). Therefore, the coding rate of the output vector \( {\bf v}_{j, \rm red} \) is given by
\begin{equation} \label{e.R_jc}
  R_{j, c} 
  = \frac{\min\{K, R\}}{R} \cdot (\log_2 p) 
  = \eta_{j,c} \cdot \log_2 p, 
\end{equation}
where \( \eta_{j,c} = \frac{\min\{K, R\}}{R} \) is the loading factor of \( {\bf v}_{j, \rm red} \).

Based on Lemma \ref{lemma.SU_parity}, when \( K < Q \), the sequence \( {\bf c}_j = ({\bf d}_j, {\bf 0}, {\bf c}_{j,\rm red}) \) can be viewed as a concatenation of two sub-sequences with different rates. Specifically, \( {\bf d}_j \) is a sub-sequence with rate \( \log_2 p \), and \( {\bf c}_{j,\rm red} \) is a sub-sequence with rate \( \frac{K}{Q} \cdot \log_2 p \). Therefore, we refer to this type of sequence \( {\bf c}_j \) as a \textit{multirate sequence}.
On the other hand, when \( K \ge Q \), both the sub-sequence \( {\bf d}_j \) and the sub-sequence \( {\bf c}_{j,\rm red} \) have the same rate of \( \log_2 p \). In this case, the sequence \( {\bf c}_j \) is a \textit{single-rate sequence}.
We can deduce the aforementioned result for the case of \( K = Q \) in the following Corollary \ref{cor.SU_parity}.

\begin{corollary} \label{cor.SU_parity}
  Let \( {\bf c}_j = ({\bf d}_j, {\bf c}_{j,\rm red}) \) represent a codeword of a channel code \( {\mathcal C}_{mc} \), where \( {\bf d}_j \) is the information vector and \( {\bf c}_{j,\rm red} \) is the parity vector. Assume that \( {\mathcal C}_{mc} \) is a binary code constructed over GF(2). Let the length of the information vector \( {\bf d}_j \) be \( K \), and the length of the parity vector \( {\bf c}_{j,\rm red} \) be \( Q \). If \( K \ge Q \), the codeword \( {\bf c}_j \) forms a concatenated sequence with a single rate. In this case, the coding rate of the channel code \( {\mathcal C}_{mc} \) is \( R_{mc} = \frac{K}{K+Q} \ge 0.5 \). Therefore, for a binary linear block channel code with a rate greater than or equal to $0.5$, the encoded codeword is a single-rate sequence. 
\end{corollary}

This single-rate feature explains why it is easier to design a binary LDPC code with a rate greater than or equal to \( 0.5 \), as most designs assume the codeword has a single rate by default. However, this is not the case for low-rate codes. Most current LDPC designs, particularly for low-rate codes, do not account for the multi-rate nature of the code. As a result, designing low-rate LDPC codes is more challenging.






\section{Channel Capacity of PA-FFMA Systems}

It is important to note that \textit{channel capacity} and \textit{error performance} are two fundamental performance metrics in a communication system. While the \textit{maximum channel capacity} and \textit{minimum error probability} may occasionally align, they do not always coincide. Therefore, when designing a system, these metrics should be treated as distinct entities and considered independently. 

In this section, we examine the channel capacity of the proposed PA-FFMA system along with its corresponding power allocation scheme. In the following section, we will investigate the error performance metric and the associated power allocation strategy.
We begin by presenting the channel capacity of a single-user (SU) PA-FFMA system over an AWGNC, which is equivalent to a \textit{PA-LDPC code} transmitted through an AWGNC. To illustrate, consider the $j$-th user as an example. Following this, we examine the channel capacity of multiuser (MU) PA-FFMA transmission over a GMAC and introduce the corresponding power allocation scheme.

\vspace{-0.1in}
\subsection{Channel Capacity of a Single-user FF-TDMA System}

First, we present the channel capacity of a single-user FF-TDMA system, which is equivalent to a FF-CCMA system without channel coding. Therefore, the degrees of freedom of the FF-TDMA system are equal to \( m \). In other words, the FFMA system passes through \( m \) independent AWGNCs. By analyzing the channel capacity, we can also derive the optimal power allocation for the system. The channel capacity and the optimal power allocation scheme are summarized in the following theorem.

\begin{theorem} \label{theorem.SU.TDMA}
  (\textbf{Channel Capacity of a Single-User FF-TDMA System}) 
  Consider a system with \( m \) independent AWGNCs, which results in a DoF of \( m \). Suppose one user occupies \( K \) DoFs for transmitting information symbols and \( Q \) DoFs for transmitting parity symbols. The total transmit power is set to \( m P_{\text{avg}} \), where \( P_{\text{avg}} \) denotes the average power of each symbol. The optimal channel capacity of this system is given by:
  \begin{equation} \label{e.C_su_tdma}
    C_{\rm SU}^{\rm td} = (K + Q) \cdot C\left( \frac{m}{K + Q} \gamma_a \right)
    = \frac{K + Q}{2} \log_2\left(1 + \frac{m}{K + Q} \gamma_a\right),
  \end{equation}
  where $C(x) = \frac{1}{2} \log_2 \left( 1 + x \right)$, \( \gamma_a = \frac{P_{\text{avg}}}{\sigma^2} \) represents the average signal-to-noise power ratio (SNR), and \( \sigma^2 \) is the variance of the AWGNC.
  The subscript ``SU'' stands for ``single-user'' and the superscript ``td'' means ``FF-TDMA mode''.
  Consequently, the corresponding polarization-adjusted factors are:
  \[
  \mu_{1} = \mu_{2} = \frac{m}{K + Q}.
  \]
  This result implies that an equal power allocation (EPA) between the information vector and the parity vector achieves the maximum channel capacity. At this point, the polarization-adjusted scaling \( \mu_{\text{pas}} \) equals 1, i.e., \( \mu_{\rm pas} = \frac{\mu_{1}}{\mu_{2}} = 1 \).
\end{theorem}

\begin{proof}
For a single-user FF-TDMA system, we assume that the multiuser code is based on an $(m, M)$ linear block code, denoted as ${\mathcal C}_{mc}$. Each codeword $w_j \in {\mathbb P}^{1 \times m}$ from ${\mathcal C}_{mc}$ is mapped to a modulated signal vector, denoted as ${\bf x}_j = (x_{j,0}, x_{j,1}, \dots, x_{j,i}, \dots, x_{j,m-1}) \in {\mathbb C}^{1 \times m}$, where each element $x_{j,i}$ belongs to a set $\mathcal X$, i.e., $x_{j,i} \in {\mathcal X}$. Let the received signal vector be denoted as ${\bf y} = (y_{j,0}, y_{j,1} \dots, y_{j,i}, \dots, y_{j,m-1}) \in {\mathbb C}^{1 \times m}$, where each element $y_{j,i}$ belongs to a set $\mathcal Y$, i.e., $y_{j,i} \in {\mathcal Y}$. Suppose the detected codeword $\widehat{w}$ is estimated as $\widehat{w} = g({\bf y})$. Thus, the process can be described as $w_j \to {\bf x}_j \to {\bf y}_j \to \widehat{w}_j$.

According to the data processing inequality, we have the following relationship:
\begin{equation}
  {\rm I}({w}; {\widehat{w}}) \le {\rm I}({\mathcal X}^m; {\mathcal Y}^m) 
  = {\rm H}({\mathcal Y}^m) - \sum_{i=0}^{m-1}{\rm H}({\mathcal Y}_i|{\mathcal X}_i) \le \sum_{i=0}^{m-1} {\rm I}({\mathcal X}; {\mathcal Y}),
\end{equation}
which will be discussed in further detail in the following section. 
Since the codeword $w_j$ is a multirate sequence, consisting of both the information vector and the parity vector, we calculate their respective channel capacities. 
The channel capacity of the information vector is given by
\begin{equation} \label{e.C_j_1}
  C_{j, 1} = {\rm I}(x_{j,i}; y_{j,i}) = C(\mu_{1} \gamma_a)
  = \frac{1}{2} \log_2 \left( 1 + \mu_{1} \gamma_a \right) ,
\end{equation}
where $\mu_{1} \gamma_a$ represents the average SNR per information symbol. 

Similarly, the channel capacity of the parity vector is given by
\begin{equation} \label{e.C_j_2}
  C_{j, 2} = {\rm I}(x_{j,i}; y_{j,i}) = C(\mu_{2} \gamma_a)
  = \frac{1}{2} \log_2 \left( 1 + \mu_{2} \gamma_a \right),
\end{equation}
where $\mu_{2} \gamma_a$ represents the average SNR per parity symbol of the multiuser code ${\mathcal C}_{mc}$.

Hence, the channel capacity of the proposed FF-TDMA (or PA-LDPC) system over an AWGNC is given by:
  \begin{equation} \label{e.su_tdma_capacity}
      C_{\rm SU}^{\rm td} 
      = \frac{1}{2} \sum_{i=0}^{m-1} \log_2 (1 + \gamma_{i})
      = \frac{K}{2} \log_2 (1 + \mu_{1} \gamma_{a})
      + \frac{Q}{2} \log_2 (1 + \mu_{2} \gamma_{a}).
  \end{equation}
Therefore, our goal is to maximize the channel capacity and find the corresponding polarization-adjusted factors. The optimal function and corresponding power constraint are given as:
\begin{equation} \label{e.SU_tdma_c_obj}
  \begin{array}{ll}
    (\mu_{1}, \mu_{2}) &= \arg\max C_{\rm SU}^{\rm td} \\ 
   \text{s.t.} \quad C1:  & K \cdot \mu_{1} + Q \cdot \mu_{2} = m,\\
  \end{array}
\end{equation}
where the condition \( C1 \) ensures that the total power is constant at \( m P_{\text{avg}} \).

To solve the optimization problem, we introduce a Lagrange multiplier to the objective function (\ref{e.SU_tdma_c_obj}) and construct the corresponding Lagrangian function. The optimal solution for the parameters \( \mu_1 \) and \( \mu_2 \) are derived as follows:
\[
\mu_{1} = \mu_{2} = \frac{m}{K + Q},
\]
which is then substituted into (\ref{e.su_tdma_capacity}). 
Thus, we obtain the maximum channel capacity as shown in (\ref{e.C_su_tdma}).
\end{proof}

According to Theorem \ref{theorem.SU.TDMA}, it can be observed that the channel capacity of the $j$-th user is equivalent to that of a $(K+Q, K)$ code with a loading factor of $\eta = \frac{K}{K+Q}$. The total power, $m P_{avg}$, is equally distributed among the $K + Q$ symbols, such that each symbol receives $\frac{m}{K+Q} P_{avg}$. In this scenario, both the information and parity symbols share the same power. However, this result holds true only when the codeword $w_j$ represents a single-rate sequence.


\vspace{-0.1in}
\subsection{Channel Capacity of a Single-user FF-CCMA System}
For the single-user FF-CCMA with a channel coding system, the DoFs are determined by the codeword length of the channel code ${\mathcal C}_{gc}$. Since the codeword length is $N$, this indicates that the FF-CCMA system passes through $N$ independent AWGNCs. The following lemma summarizes the channel capacity and optimal power allocation scheme of the single-user PA-FF-CCMA.

\begin{lemma}
  (\textbf{Channel Capacity of a Single-User FF-CCMA System}) 
  Consider a system with $N$ independent AWGNCs. Suppose one user occupies $K$ DoFs for transmitting information symbols with symbol power $\mu_1 P_{\text{avg}}$, $Q$ DoFs for transmitting parity symbols of the multiuser code ${\mathcal C}_{mc}$ with symbol power $\mu_2 P_{\text{avg}}$, and $R$ DoFs for transmitting additional parity symbols of the channel code ${\mathcal C}_{gc}$ with symbol power $\mu_c P_{\text{avg}}$. The total transmit power is set to $N P_{\text{avg}}$. The optimal channel capacity of this system is given by:
  \begin{equation} \label{e.C_su_ccma}
    C_{\rm SU}^{\rm cc} = (K + Q + R)\cdot C\left(\frac{N}{K+Q+R} \gamma_a \right)
      = \frac{K + Q + R}{2} \log_2\left(1 + \frac{N}{K + Q + R} \gamma_a\right),
  \end{equation}
  where $\gamma_a $ represents the average SNR, and the superscript ``cc'' stands for ``FF-CCMA mode''.
  Consequently, the corresponding polarization-adjusted factors are:
  \[
  \mu_{1} = \mu_{2} = \mu_{c} = \frac{N}{K + Q + R}.
  \]
  This result implies that equal power allocation (EPA) between the information vector and the parity vectors achieves the maximum channel capacity.
\end{lemma}

\begin{proof}

  In the PA-FF-CCMA system, the multiuser code ${\mathcal C}_{mc}$ is concatenated with the channel code ${\mathcal C}_{gc}$. Thus, there are three distinct components: the information part ${\bf d}_j$ with length \(K\), the parity part ${\bf c}_{j, \rm red}$ of the multiuser code ${\mathcal C}_{mc}$ with length \(Q\), and the parity part ${\bf v}_{j, \rm red}$ of the multiuser code ${\mathcal C}_{gc}$ with length \(R\).
  The channel capacity of the information part is given by equation (\ref{e.C_j_1}), and the channel capacity of the parity part determined by the multiuser code \( {\mathcal C}_{mc} \) is given by equation (\ref{e.C_j_2}). The channel capacity of the check part determined by the channel code \( {\mathcal C}_{gc} \) is:
  \begin{equation}  \label{e.C_j_3}
  C_{j, c} = {\rm I}({\mathcal X}; {\mathcal Y}) = C(\mu_{c} \gamma_a)
           = \frac{1}{2} \log_2 \left( 1 + \mu_{c} \gamma_a \right),
  \end{equation}
  where \( \mu_{c} \gamma_a \) represents the average SNR of each parity symbol of the channel code \( {\mathcal C}_{gc} \).
  
  Thus, the total channel capacity of the proposed FF-CCMA system over AWGNC is:
  \begin{equation}
  C_{\rm SU}^{\rm cc} = K \cdot C(\mu_{1} \gamma_{a}) + Q \cdot C(\mu_{2} \gamma_{a}) + R \cdot C(\mu_{c} \gamma_{a}).
  \label{e.C_total}
  \end{equation}
  Our goal is to maximize the channel capacity \( C_{\rm SU}^{\rm cc} \) and determine the corresponding polarization adjustment factors. The optimization problem and the associated power constraint are:
  \begin{equation}
    \begin{array}{rl}
  (\mu_{1}, \mu_{2}, \mu_{c}) &= \arg\max C_{\rm SU}^{\rm cc} \\ 
  \text{s.t.} \quad C2:  & K \cdot \mu_{1} + Q \cdot \mu_{ 2} + R \cdot \mu_{c} = N,
   \end{array}
  \end{equation}
  where the condition \( C2 \) ensures that the total power is constant at \( N P_{\text{avg}} \).

  To maximize \( C_{\rm SU}^{\rm cc} \) subject to the constraint, we directly apply the result from Theorem \ref{theorem.SU.TDMA}, given by:
  \begin{equation}
    \begin{array}{ll}
  C_{\rm SU}^{\rm cc} 
        &\overset{(a)}{\le} (K + Q) \cdot C\left( \frac{m}{K + Q} \gamma_a \right)
        + R \cdot C(\mu_{c} \gamma_{a}), \\
        &\overset{(b)}{\le} (K + Q + R) \cdot C\left( \frac{N}{K + Q + R} \gamma_a \right),\\
    \end{array}
  \end{equation}
where step $(a)$ is derived using Theorem \ref{theorem.SU.TDMA} under the conditions \( \mu_1 = \mu_2 = \frac{m}{K+Q} \) and \( N = m + R \cdot \mu_{c} \). Similarly, we obtain step $(b)$, which holds under the condition \( \mu_c = \mu_1 = \mu_2 = \frac{N}{K+Q+R} \).
\end{proof}

\subsection{Channel Capacity of a Multiuser FF-TDMA System}

We now examine the multiuser FF-TDMA system, where each user transmits with the same average power per symbol and utilizes an identical power allocation vector.
The following theorem presents the channel capacity and the optimal power allocation scheme for this multiuser FF-TDMA system.

\begin{theorem} \label{theorem.TDMA_MU}
  (\textbf{Channel Capacity of a Multiuser FF-TDMA System}) 
  Consider a multiuser transmission system with $J_{mc}$ users, which utilizes a total of \( m \) independent DoFs. Each user occupies \( K \) orthogonal DoFs for transmitting information symbols, so the $J$ users collectively occupy $J_{mc} \cdot K$ DoFs. Additionally, the $J_{mc}$ users share \( Q \) DoFs for transmitting their respective parity symbols, meaning the parity symbols are superposed and received at the receiver.
  The total transmit power for each user is set to \( m P_{\text{avg}} \). The optimal channel capacity for this system is given by:
  \begin{equation} \label{e.mu_maximum_tdma}
    C_{\rm MU}^{\rm td} = \frac{K \cdot J_{mc} 
    + Q}{2} \log_2 \left( 1 + \frac{J_{mc} \cdot m \cdot r_a}{K \cdot J_{mc} + Q} \right)
  \end{equation}
  where \( \gamma_a  \) represents the average SNR, and the subscript ``MU'' stands for ``multiuser''.
  Consequently, the corresponding polarization-adjusted factors are:
  \[
    \mu_{1} = J_{mc} \cdot \mu_{2} =\frac{J_{mc} \cdot m}{K \cdot J_{mc} + Q}, \quad 
    \mu_{2} = \frac{m}{K \cdot J_{mc} + Q}.
  \]
  Hence, the polarization-adjusted scaling factor \( \mu_{\text{pas}} \) is equal to the number of users \( J_{mc} \), i.e., 
  \(
  \mu_{\rm pas} = \frac{\mu_{1}}{\mu_{2}} = J_{mc}.
  \)
  We refer to the power allocation scheme as the multiuser equal power allocation (MU-EPA).
\end{theorem}

\begin{proof}

In the FF-TDMA system, the information parts of different users are kept orthogonal, meaning that each user is assigned distinct DoFs and maintains the power $\mu_1 P_{avg}$. Consequently, the capacity of the information vector remains equal to $C_{j,1}$, as given by Eq. (\ref{e.C_j_1}).

For the parity vector in a multiuser setup, as shown in \cite{FFMA, FFMA2}, there is no multiuser interference (MUI). This is because the receiver directly converts the CFSP block (or superposition) into a FFSP block using the transformation function ${\rm F}_{\rm C2F}$, without needing to separate users in the complex field. According to multiuser information theory \cite{Thomas}, the capacity for the parity vector is given by
\begin{equation} \label{e.C_sum_2}
  C_{\text{sum}, 2} 
  = C(J_{mc} \cdot \mu_{2} \gamma_a)
  = \frac{1}{2} \log_2 \left( 1 + J_{mc} \cdot \mu_{2} \gamma_a \right).
\end{equation}
Thus, the channel capacity for the $J$-user system is given by
\begin{equation} \label{e.C_MU_tdma}
  \begin{aligned}
    C_{\rm MU}^{\rm td} 
    &= KJ_{mc} \cdot C(\mu_1 \gamma_a) + Q \cdot C(J_{mc} \cdot \mu_{2} \gamma_a)\\
    &=  \frac{K J_{mc}}{2} \log_2 \left( 1 + \mu_{1} \gamma_a \right) 
      + \frac{Q}{2} \log_2 \left( 1 + J_{mc} \cdot \mu_{2} \gamma_a \right).
  \end{aligned}
\end{equation}
Our objective is to maximize the channel capacity and identify the corresponding polarization-adjusted factors. The optimal function and the associated power constraint are given as:
\begin{equation} \label{e.MU_c_obj}
  \begin{array}{ll}
    (\mu_{1}, \mu_{2}) &= \arg\max C_{\rm MU}^{\rm td} \\ 
   \text{s.t.} \quad C1:  & K \cdot \mu_{1} + Q \cdot \mu_{2} = m.\\
  \end{array}
\end{equation}





To solve the optimization problem, we employ the Lagrangian function and determine the optimal values for the parameters \( \mu_1 \) and \( \mu_2 \) as follows:
\[
  \mu_1 = \frac{J_{\text{mc}} \cdot m}{K \cdot J_{\text{mc}} + Q}, \quad 
  \mu_2 = \frac{m}{K \cdot J_{\text{mc}} + Q}.
\]
These values are then substituted into (\ref{e.C_MU_tdma}), yielding the maximum channel capacity as shown in (\ref{e.mu_maximum_tdma}).
\end{proof}

When the maximum number of served users, \( J_{mc} \), is known, an $(m, M)$ multiuser code ${\mathcal C}_{mc}$ can be designed with the conditions \( M = J_{mc} \cdot K \) and \( m = M + Q \). In this case, it follows that \( \mu_{1} = J_{mc} \) and \( \mu_{2} = 1 \), which corresponds to the \textit{Maximum Information Power (MIP)} method as presented in \cite{FFMA2}. In fact, when the transmitter does not know the total number of users, the MIP method is more suitable for the general case. We can summarize this result as follows:

\begin{corollary}
  For an $(m, M)$ multiuser code ${\mathcal C}_{mc}$ with the conditions \( M = J_{mc} \cdot K \) and \( m = M + Q \), the optimal power allocation scheme to achieve maximum channel capacity of a $J_{mc}$-user PA-FF-TDMA system is the \textit{Maximum Information Power (MIP)} method.
\end{corollary}

\subsection{Channel Capacity of a Multiuser FF-CCMA System}

This subsection examines the channel capacity and the optimal power allocation scheme for the multiuser FF-CCMA system. The framework of the FF-CCMA system is detailed in \cite{FFMA2}, where the information section consists of $T$ data blocks, with each data block serving $J_{mc}$ users. Therefore, the total number of users served is $J = J_{mc} \cdot T$.

\begin{lemma}
  (\textbf{Channel Capacity of a Multiuser FF-CCMA System}) 
  Consider a multiuser transmission system with $J$ users, utilizing a total of \( N \) independent DoFs. The $J$ users are divided into $T$ groups, with each group consisting of $J_{mc}$ users, such that the relationship $J = T \cdot J_{mc}$ holds. Each user occupies \( K \) orthogonal DoFs for transmitting their information symbols, so collectively the $J$ users occupy \( J \cdot K \) DoFs. Additionally, each group of $J_{mc}$ users shares \( Q \) DoFs for transmitting their respective parity symbols. Furthermore, all $J$ users share another \( R \) DoFs for transmitting their own respective parity symbols.
  The total transmit power for each user is set to \( N P_{\text{avg}} \). The optimal channel capacity for this system is given by:
  \begin{equation} \label{e.mu_maximum_ccma}
    C_{\rm MU}^{\rm cc} = \frac{\left( K \cdot J + Q \cdot T + R \right)}{2} \log_2\left( 1 + \frac{N \cdot J \cdot \gamma_a}{K \cdot J + Q \cdot T + R} \right),
  \end{equation}
  where \( \gamma_a \) represents the average SNR.
  Consequently, the corresponding polarization-adjusted factors are:
  \begin{equation} 
    \begin{array}{ll}
      \mu_{c} = \frac{N}{K \cdot J + Q \cdot T + R}, \quad
      \mu_{1} = J \cdot \mu_{c} = J \cdot \frac{N}{K \cdot J + Q \cdot T + R}, \quad
      \mu_{2} = T \cdot \mu_{c} = T \cdot \frac{N}{K \cdot J + Q \cdot T + R}.
    \end{array}
  \end{equation}
  Therefore, the polarization-adjusted scaling factor \( \mu_{\text{pas}} \) is defined as the ratio of \( \mu_{1} \) to \( \mu_{2} \), i.e.,
  \[
  \mu_{\rm pas} = \frac{\mu_{1}}{\mu_{2}} = J_{mc},
  \]
  which implies that the PAS of the multiuser code is equivalent to that of the FF-TDMA mode.
\end{lemma}

\begin{proof}

For the PA-FF-CCMA system, both the parity vector of the multiuser code \( {\mathcal C}_{mc} \) and the parity vector of the channel code \( {\mathcal C}_{gc} \) consist of superimposed signals from multiple users. Specifically, the parity vector of the multiuser code \( {\mathcal C}_{mc} \) is the sum of the signals from \( J_{mc} \) users, while the parity vector of the channel code \( {\mathcal C}_{gc} \) is the sum of signals from \( J \) users. As shown in Eq. (\ref{e.C_sum_2}), the capacity of the parity vector of the multiuser code \( {\mathcal C}_{mc} \) is $C(J_{mc} \cdot \mu_{2} \gamma_a)$, and the capacity of the parity vector of the channel code \( {\mathcal C}_{gc} \) is:
  \begin{equation} \label{e.C_sum_c_ccma}
      C_{\text{sum}, c} = C(J \cdot \mu_{c} \gamma_a)
      = \frac{1}{2} \log_2 \left( 1 + J \cdot \mu_{c} \gamma_a \right).
  \end{equation}

Hence, the total channel capacity for the \(J\)-user system is given by:
  \begin{equation} \label{e.C_MU_cc}
    \begin{aligned}
      C_{\rm MU}^{\rm cc}  
      &= KJ \cdot C(\mu_1 \gamma_a) + QT \cdot C(J_{mc} \cdot \mu_{2} \gamma_a)
         + R \cdot C(J \cdot \mu_{c} \gamma_a)\\
      &= \frac{KJ}{2} \log_2(1 + \mu_{1} r_a) 
       + \frac{QT}{2} \log_2(1 + J_{mc} \cdot \mu_{2} r_a) 
       + \frac{R}{2}  \log_2(1 +      J \cdot \mu_{c} r_a)
    \end{aligned}
  \end{equation}
  
Our objective is to maximize the channel capacity \(C_{\rm MU}^{\rm cc}\) and identify the corresponding polarization-adjusted factors. The optimal function and the associated power constraint are given as:
  \begin{equation} \label{e.MU_ccma_obj}
    \begin{array}{rl}
      (\mu_{1}, \mu_{2}, \mu_{c}) &= \arg\max C_{\rm MU}^{\rm cc} \\ 
     \text{s.t.} \quad C2:  & K \cdot \mu_{1} + Q \cdot \mu_{2} 
                            + R \cdot \mu_{c} = N, \\
                       C3:  & J_{mc} \cdot T = J.
    \end{array}
  \end{equation}
To maximize the channel capacity \( C_{\text{MU}}^{\text{cc}} \) under the given constraints, we construct the Lagrange function. Through this approach, we derive the optimal parameter \( \mu_{c} \) as
\(
  \mu_{c} = \frac{N}{K \cdot J + Q \cdot T + R},
\)
where \( N \), \( K \), \( J \), \( Q \), \( T \), and \( R \) are system-defined constants. By substituting the optimal values of \( \mu_{1} \), \( \mu_{2} \), and \( \mu_{c} \) into (\ref{e.C_MU_cc}), we obtain the maximum channel capacity, as expressed in (\ref{e.mu_maximum_ccma}).
%
\end{proof}

Suppose the maximum number of served users is \( J \), where \( J \) is the product of two integers, i.e., \( J = J_{mc} \cdot T \). Then, we can design an $(m, M)$ multiuser code ${\mathcal C}_{mc}$ with the conditions \( M = J_{mc} \cdot K \) and \( m = M + Q \), along with an $(N, K_{gc})$ channel code ${\mathcal C}_{gc}$ where \( K_{gc} = J \cdot K + Q \cdot T \) and \( N = K_{gc} + R \). In this case, it follows that \( \mu_{1} = J \), \( \mu_{2} = T \), and \( \mu_{c} = 1 \), which corresponds to the \textit{Maximum Block Information Power (MBIP)} method as presented in \cite{FFMA2}. We can summarize this result as follows:

\begin{corollary}
For an $(m, M)$ multiuser code ${\mathcal C}_{mc}$ with the conditions \( M = J_{mc} \cdot K \) and \( m = M + Q \), and an $(N, K_{gc})$ channel code ${\mathcal C}_{gc}$ with the conditions \( K_{gc} = J \cdot K + Q \cdot T \) and \( N = K_{gc} + R \), the optimal power allocation scheme to achieve the maximum channel capacity of a PA-FF-CCMA system is the \textit{Maximum Block Information Power (MBIP)} method.
\end{corollary}
Above, we have introduced the channel capacity of FF-TDMA and FF-CDMA. In fact, in addition to the FF-TDMA and FF-CCMA modes, there are other modes in FFMA systems. Due to page limitations, these other modes will not be analyzed further. It is important to note that among all these modes, FF-TDMA is the fundamental mode of an FFMA system. Therefore, we use the FF-TDMA mode as an example to analyze system performance. For simplicity, we refer to FF-TDMA as FFMA.

\section{Finite-Blocklength Analysis of PA-FFMA}

In this section, we present a finite-blocklength (FBL) performance analysis of our proposed FFMA system, evaluating its behavior under both point-to-point (P2P) Gaussian channels and GMAC scenarios. 

\subsection{Point-to-Point Gaussian Channel Analysis}
The FBL analysis for P2P Gaussian channels, originally established by Polyanskiy et al. \cite{FBL_2,FBL_3}, establishes two fundamental limitations: first, the finite blocklength induces a non-negligible capacity reduction relative to the Shannon capacity; second, achieving equivalent error performance requires substantially higher $E_b/N_0$ compared to asymptotic blocklength regimes.
We extend this foundational framework to our single-user FFMA system, incorporating the average power constraint across codewords. These considerations yield the following corollary:
\begin{corollary}
  \label{corollary_SU_FFMA_FBL}
  Consider a PA-FFMA system with $m$ degrees of freedom (DoFs), where $K$ DoFs carry information symbols and $Q$ DoFs carry parity symbols. Under a total transmit power constraint of $m P_{\text{avg}}$ with per-symbol average power $P_{\text{avg}}$, let $\mu_1 P_{\text{avg}}$ and $\mu_2 P_{\text{avg}}$ denote the average power allocated to information and parity symbols, respectively, where $\mu_1$ and $\mu_2$ represent polarization adjusted factors. For any error probability $\varepsilon \in (0,1)$, the maximal achievable rate $R_q$ of an $(m, K, Q, \varepsilon, P_{\text{avg}})$ PA-FFMA system satisfies:

  \vspace{-0.15in}
  \begin{small}
  \begin{equation} 
  \label{e.FBL_SU_FFMA}
  \begin{aligned}
  R_q &\leq \frac{C_{\mathrm{SU}}^{\mathrm{td}}}{m} - \sqrt{ \frac{V(P_{\text{avg}})}{m} } Q^{-1}(\varepsilon) + \frac{\log m}{2m} + \mathcal{O}\left(\frac{1}{m}\right) \\
  &= \frac{K C(\mu_1 \gamma_a) + Q C(\mu_2 \gamma_a)}{m} - \sqrt{ \frac{V(P_{\text{avg}})}{m} } Q^{-1}(\varepsilon) 
   + \frac{\log m}{2m} + \mathcal{O}\left(\frac{1}{m}\right),
  \end{aligned}
  \end{equation}
  \end{small}
  where $C_{\mathrm{SU}}^{\mathrm{td}}$ is deduced in Theorem \ref{theorem.SU.TDMA}, and the channel dispersion is
  \(
  V(P_{\text{avg}}) = \frac{P_{\text{avg}}(P_{\text{avg}}+2)}{2(1+P_{avg})^2}.
  \)
  \end{corollary}

In Corollary~\ref{corollary_SU_FFMA_FBL}, we utilize the derived capacity expression $C_{\mathrm{SU}}^{\mathrm{td}}$ in place of the original channel capacity from \cite{FBL_3,FBL_MU}, while maintaining all other terms unchanged. This substitution is justified since the single-user PA-FFMA system operates under identical power constraints as the P2P Gaussian scenario.

\subsection{Gaussian Mutiple-access Channel}
In \cite{FBL_MU}, the authors extended the FBL analysis to multiuser GMACs. Compared to the P2P Gaussian channel scenario, the multiuser FBL analysis exhibits two fundamental differences: first, the single-user channel capacity is replaced by multiuser sum capacity; second, the capacity reduction stems from both finite blocklength effects and multiuser interference.
Consequently, maintaining the same error probability requires higher $E_b/N_0$ in multiuser scenarios compared to P2P cases. Building upon these results, we extend the FBL analysis to our proposed multiuser PA-FFMA system, leading to the following corollary:

\begin{corollary} \label{corollary.MU_FBL}  
Consider a $J$-user transmission system operating over $m$ independent DoFs, where each user is allocated $K$ orthogonal DoFs for information symbols, yielding a total information-carrying dimension of $JK$ DoFs. The system employs $Q = m - JK$ shared DoFs for parity symbols across all users. For any target error probability $\varepsilon \in (0,1)$, the maximal achievable sum rate $R_q$ of the $(m,J,K,Q,\varepsilon,P_{\text{avg}})$ PA-FFMA system is bounded by:

\begin{small}
  \begin{equation} \label{eq:sum_rate_bound}
  \begin{aligned}
  J R_q &\leq \frac{C_{\mathrm{MU}}^{\mathrm{td}}}{m} - \sqrt{ \frac{V(\mu_1 P_{\text{avg}})}{JK} + \frac{V_{\mathrm{cr}}(J, \mu_2 P_{\text{avg}})}{Q} } Q^{-1}(\varepsilon) 
  \frac{\log m}{2m} + \mathcal{O}\left(\frac{1}{m}\right) \\
  &= \frac{JK C(\mu_1 \gamma_a) + Q C(J\mu_2 \gamma_a)}{m} 
  - \sqrt{ \frac{V(\mu_1 P_{\text{avg}})}{JK} + \frac{V_{\mathrm{cr}}(J, \mu_2 P_{\text{avg}})}{Q} } Q^{-1}(\varepsilon) 
  \frac{\log m}{2m} + \mathcal{O}\left(\frac{1}{m}\right),
  \end{aligned}
  \end{equation}
\end{small}
  where $C_{\mathrm{MU}}^{\mathrm{td}}$ represents the sum capacity derived in Theorem~\ref{theorem.TDMA_MU}, and the cross-dispersion term characterizing multiuser interference is:
  \begin{equation}
  V_{\mathrm{cr}}(J, \mu_2 P_{\text{avg}}) \triangleq \frac{J(J-1)(\mu_2 P_{\text{avg}})^2}{2(1+J\mu_2 P_{\text{avg}})^2}.
  \end{equation}
\end{corollary}

In the proposed PA-FFMA system, the information symbols maintain orthogonality among the $J$ users, making the information section equivalent to the P2P case analyzed in Corollary~\ref{corollary_SU_FFMA_FBL}. In contrast, the parity symbols form superimposed signals, exhibiting characteristics of the multiuser finite-blocklength scenario.

Corollary~\ref{corollary.MU_FBL} utilizes the derived sum capacity expression $C_{\mathrm{MU}}^{\mathrm{td}}$ for $J$ users, replacing the original sum capacity from \cite{FBL_MU}. Furthermore, we modify both the channel dispersion and cross-dispersion terms as follows:
\begin{itemize}
    \item The channel dispersion term $\frac{V(\mu_1 P_{\text{avg}})}{JK}$ captures the effects of the orthogonal information section, where $V(\cdot)$ represents the channel dispersion function for the P2P case;
    
    \item The cross-dispersion term $\frac{V_{\mathrm{cr}}(J, \mu_2 P_{\text{avg}})}{Q}$ characterizes the interference in the superimposed parity section, with $V_{\mathrm{cr}}(\cdot)$ denoting the cross-dispersion function for multiuser scenarios.
\end{itemize}
This analytical framework remains valid as the multiuser PA-FFMA system preserves identical power constraints to the GMAC scenario. The \textit{orthogonal-superimposed hybrid structure} provides enhanced spectral efficiency through orthogonal transmission of information symbols to avoid multiuser interference while employing superimposed parity symbols to improve resource utilization.

\begin{figure}[t] 
  \centering
  \includegraphics[width=0.46\textwidth]{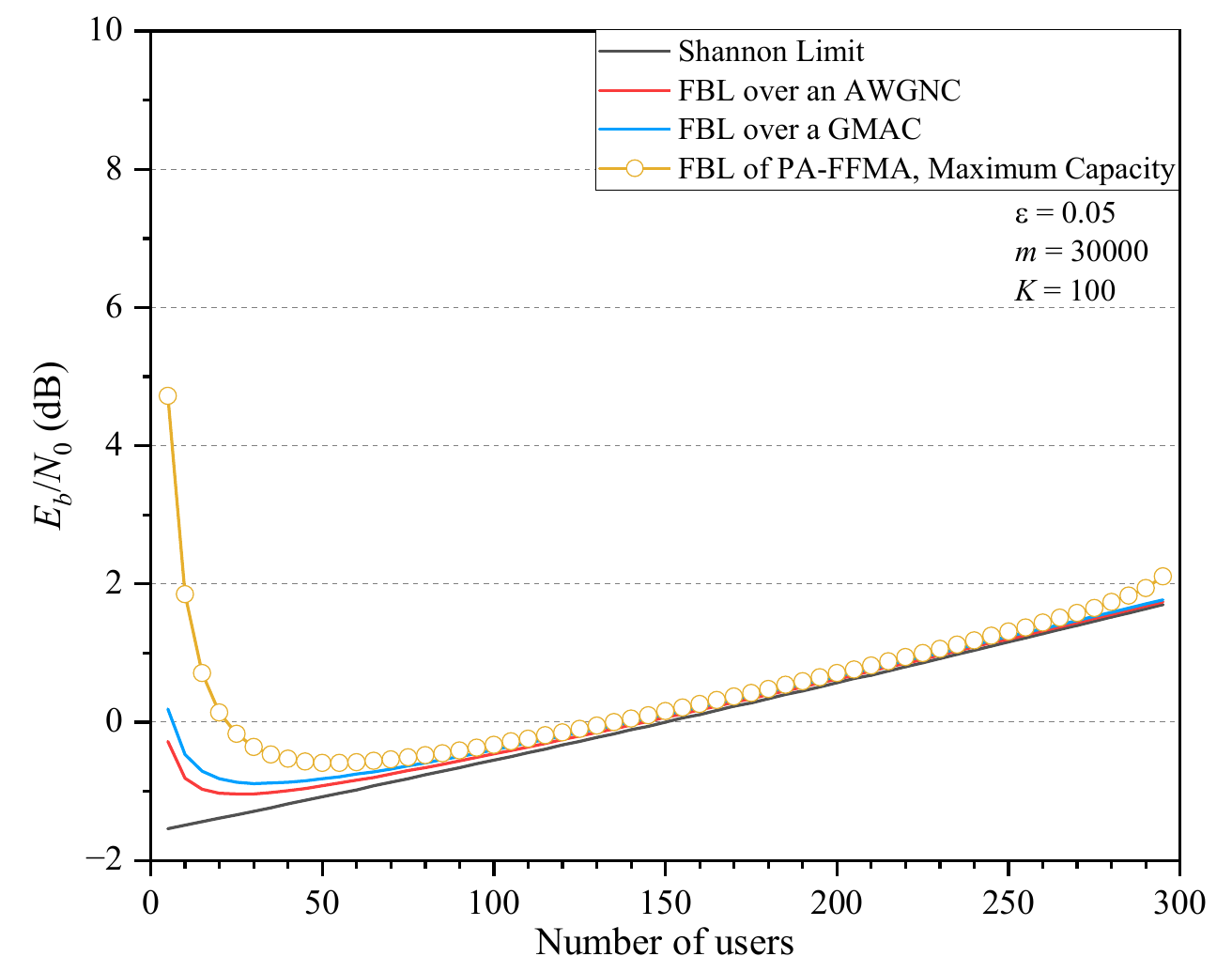}
  \caption{Minimum required $E_b/N_0$ (dB) for different communication scenarios, including: (1) theoretical Shannon limits, (2) FBL P2P AWGNC, (3) FBL GMAC, and (4) multiuser PA-FFMA.}
  \label{Fig_Graph9}
  \vspace{-0.2in}
\end{figure}

Figure~\ref{Fig_Graph9} presents a comprehensive comparison of FBL effects across four fundamental communication scenarios: the theoretical Shannon capacity limits, FBL performance for P2P AWGNCs, FBL characteristics for multiuser GMACs, and our proposed multiuser FFMA systems. The simulation configuration employs fixed parameters of $m = 30, 000$ DoFs, $K = 100$ information DoFs per user, and a error probability of $\varepsilon = 0.05$.
Through extensive Monte Carlo simulations, we evaluate the PA-FFMA system under the MU-EPA scheme, which is determined by the maximum channel capacity, where \( \mu_1 = J \mu_2 \) and \( \mu_2 = \frac{m}{KJ + Q} \), achieving the sum capacity \( C_{\mathrm{MU}}^{\mathrm{td}} = (KJ + Q) \cdot C\left(\frac{J m r_a}{KJ + Q}\right) \).

From Fig. \ref{Fig_Graph9}, it reveals that the PA-FFMA systems require moderately higher $E_b/N_0$ than both P2P and multiuser GMAC scenarios. This performance difference stems from the fundamental trade-off between multiplexing and diversity in our system design. 
The multiplexing gain scales linearly with the available degrees of freedom, specifically as $\mathcal{O}(JK)$, where $JK$ represents the total orthogonal dimensions allocated for information transmission in the system. This linear scaling arises from the parallel transmission capability enabled by the orthogonal information channels. On the other hand, the diversity gain exhibits logarithmic scaling $\mathcal{O}(\log(P))$ with respect to power allocation $P$, as it fundamentally depends on the power distribution strategy for the information symbols. This characteristic logarithmic behavior stems from the diminishing returns of power allocation in improving detection reliability.



\vspace{-0.1in}
\section{Rate-Driven Capacity Alignment Theorem}

In this section, we examine the error performance metric and the corresponding power allocation strategy based on the \textit{capacity-to-rate ratio (CRR)}. First, we present the maximum CRR criterion for a single-rate sequence scenario, followed by the derivation of the max-min CRR criterion for a multi-rate sequence scenario. Using the max-min CRR criterion, we derive the \textit{rate-driven capacity alignment theorem}, which leads to the \textit{capacity alignment (CA) power allocation} method. Finally, we discuss the conditions under which both performance metrics, i.e., maximizing channel capacity and minimizing error probability, can be achieved simultaneously.

\vspace{-0.1in}
\subsection{Maximum Capacity-to-Rate Ratio Criterion for a Single-rate Sequence}

First, we define the \textit{capacity-to-rate ratio (CRR)}, a metric for evaluating error performance in communication systems.

\begin{definition} (\textbf{Capacity-to-Rate Ratio (CRR)})
  For a single-rate sequence, the Capacity-to-Rate Ratio is defined as the ratio of the channel capacity to the coding rate, i.e.,
    \begin{equation}
      \lambda = \frac{C}{R_q},
    \end{equation}
  where \( C \) denotes the channel capacity of the single-rate sequence, and \( R_q \) represents its coding rate.
\end{definition}

The CRR \(\lambda\) can be categorized into three distinct scenarios:
\begin{enumerate}
  \item \(\lambda > 1\): This scenario indicates that the channel capacity surpasses the actual information rate, suggesting underutilization of channel resources. Here, the channel possesses more capacity than required, leading to inefficient resource usage.
  
  \item \(\lambda = 1\): This scenario signifies that the channel capacity precisely matches the actual information rate, implying optimal utilization of channel resources. In this case, the channel capacity is perfectly aligned with the information rate, achieving full efficiency.
  
  \item \(\lambda < 1\): This scenario indicates that the channel capacity is inadequate to support the actual information rate, potentially resulting in transmission failure. Here, the channel fails to meet the required information rate, which may lead to performance deterioration or complete transmission breakdown.
\end{enumerate}

Next, using the defined CRR, we demonstrate the relationship between the error probability \( P_{e,w} \) and the CRR, as illustrated by the following maximum CRR criterion.

\begin{theorem} \label{lemma.MaxCRR_SingleRate} 
(\textbf{Maximum Capacity-to-Rate Ratio Criterion}) 
Consider an $(m, M)$ channel code ${\mathcal C}_{mc}$ over GF($p$), with a loading factor \( \eta = \frac{M}{m} \geq 0.5 \), indicating that the codeword in ${\mathcal C}_{mc}$ forms a single-rate sequence. The coding rate is \( R_q = \frac{M}{m} \log_2 (p) \), and the codeword length, or the DoFs, is given by \( m \). Each codeword \( w \in {\mathbb P}^{1 \times m} \) of ${\mathcal C}_{mc}$ is mapped to a signal vector, denoted as ${\bf x} \in {\mathcal X}^m$, which is transmitted through an AWGNC. The received signal vector is denoted as ${\bf y} \in {\mathcal Y}^m$, and the detected codeword \( \widehat{w} \) is estimated as \( \widehat{w} = g({\bf y}) \). The process can be described as \( w \to {\bf x} \to {\bf y} \to \widehat{w} \). Let the channel capacity be \( C \), where \( C = {\rm I}({\mathcal X}^m; {\mathcal Y}^m) \). Suppose the error probability of the codeword is \( {P}_{e, w} = {\rm Pr}(\widehat{w} \neq w) \). The error probability \( P_{e,w} \) is inversely proportional to the capacity-to-rate ratio (CRR) \( \lambda \). Hence, to minimize the error probability, it is proportional to maximizing the CRR, i.e.,
    \begin{equation}
      \min {P}_{e,w} \propto 
      \max \left\{ \lambda + \frac{1}{M \log_2 p} \right\}  
      \overset{(a)}{=} \max \left\{ \lambda \right\},
    \end{equation}
where \( \lambda = \frac{C}{R_q} \) represents the CRR, and step $(a)$ follows when \( \frac{1}{M \log_2 p} \) is constant.
\end{theorem}

\begin{proof}
We are unable to provide a rigorous proof of the relationship between the error probability \( P_{e,w} \) and the CRR \( \lambda \). However, we can derive an approximate relationship based on Fano's inequality and the channel coding theorem. The steps are as follows:
    \begin{equation}
      \begin{array}{ll}
        \log_2 p^M  \overset{(a)}{=} m \cdot R_q 
                   & \overset{(b)}{=} \text{H}(w) 
                    \overset{(c)}{=} \text{H}(w|\widehat{w}) + \text{I}(w;\widehat{w}), \\
                   & \overset{(d)}{\le} 1 + P_{e,w} \cdot (m R_q) + \text{I}({\mathcal X}^m; {\mathcal Y}^m)
                   \overset{(e)}{\le} 1 + P_{e,w} \cdot (m R_q) + m \cdot C, \\
      \end{array}
    \end{equation}
where the steps are explained as follows. 
Steps (a) and (b) are directly based on the definition of the coding rate. In Step (c), we express the relationship between entropy and mutual information. Step (d) is derived from Fano's inequality, and Step (e) assumes that the channels are independent, with each channel having the same capacity.
Thus, from the above, we can deduce that
    \begin{equation}
      P_{e,w} \ge 1 - \left( \lambda + \frac{1}{M \log_2 p} \right),
    \end{equation}
which indicates that the error probability \( P_{e,w} \) is inversely proportional to the CRR \( \lambda \).
\end{proof}
  
From Theorem \ref{lemma.MaxCRR_SingleRate}, we observe that for a single-rate sequence, the error probability $P_{e,w}$ is inversely proportional to the CRR \( \lambda \). Therefore, we can optimize the system design by maximizing the CRR to achieve improved error performance.

\vspace{-0.1in}
\subsection{Max-Min Capacity-to-Rate Ratio Criterion for a Multirate Sequence}

Now, we examine the relationship between error performance and the CRR in a multirate sequence scenario. A multirate sequence is characterized by having multiple CRRs, with each subsequence of the sequence associated with its own CRR. In general, the error performance is predominantly determined by the worst-case scenario, which typically corresponds to the minimum CRR of the sequence. 
Therefore, to achieve optimal error performance, it is essential to maximize the minimum CRR, which corresponds to a max-min optimization problem. This criterion is formally stated in the following theorem.

\begin{theorem} \label{lemma.MaxMin_CRR}
(\textbf{Max-Min Capacity-to-Rate Ratio Criterion}) 
Let ${\bf s}$ represent a multirate sequence of length \( m \), which corresponds to \( m \) DoFs along the E-axis. The sequence ${\bf s}$ is composed of \( L \) sub-sequences, denoted as ${\bf s} = ({\bf s}_1, {\bf s}_2, \dots, {\bf s}_l, \dots, {\bf s}_L)$, where the \( l \)-th sub-sequence, ${\bf s}_l$, has length \( m_l \) (or equivalently, \( m_l \) DoFs). Thus, the total length of the sequence is \( m = \sum_{l=1}^{L} m_l \). 
Let the rates of the sub-sequences ${\bf s}_1, {\bf s}_2, \dots, {\bf s}_L$ be \( R_1, R_2, \dots, R_L \), respectively. When the multirate sequence ${\bf s}$ is transmitted over an AWGNC, the channel capacities of the subsequences ${\bf s}_1, {\bf s}_2, \dots, {\bf s}_L$ are given by \( C_1, C_2, \dots, C_L \), respectively. Let \( \lambda_l = \frac{C_l}{R_l} \) represent the capacity-to-rate ratio (CRR) for the \( l \)-th sub-sequence. Then, the CRR set for the multirate sequence ${\bf s}$ is defined as \( \Lambda = \{\lambda_1, \lambda_2, \dots, \lambda_L\} \).
In order to minimize the error probability \( P_{e,w} \), it is required to maximize the minimum value of the CRR set \( \Lambda \), that is,
  \begin{equation}
    \min P_{e,w} \propto {\rm{Maximize}} \quad \min \{\lambda_1, \lambda_2, \dots, \lambda_L\},
  \end{equation} 
which is a max-min optimization problem. The optimal solution occurs when the CRR values of all sub-sequences are equal, i.e., 
  \(
    \lambda_1 = \lambda_2 = \dots = \lambda_L,
  \)
which implies a phenomenon of rate-driven capacity alignment in the multirate sequence. Therefore, the minimum error probability is achieved when the CRRs of all sub-sequences are maximized in a balanced manner.
\end{theorem}

\begin{proof}
Let \(\zeta = \min\{\lambda_1, \lambda_2, \dots, \lambda_L\}\). The goal is to maximize \(\zeta\). By the definition of the minimum, we have the following inequality for each \(\lambda_l\):  
    \(
    \lambda_l \geq \zeta \quad \forall l \in \{1, 2, \dots, L\}.
    \)
To maximize \(\zeta\), we seek the largest possible value of \(\zeta\) such that all \(\lambda_l\)'s satisfy \(\lambda_l \geq \zeta\).
    
First, consider the case where all \(\lambda_l\)'s are equal, i.e., \(\lambda_1 = \lambda_2 = \dots = \lambda_L = \zeta\). In this case, \(\zeta\) is trivially the maximum possible minimum value, as all \(\lambda_l\)'s are equal to \(\zeta\).
    
Next, assume that there exists at least one \(\lambda_k\) such that \(\lambda_k > \zeta\). Without loss of generality, assume \(\lambda_1 = \zeta\) and \(\lambda_2 > \zeta\). In this situation, we can increase \(\zeta\) by setting \(\lambda_1 = \zeta + \epsilon\) and adjusting \(\lambda_2\) to \(\lambda_2 - \epsilon\), where \(\epsilon > 0\) is chosen such that \(\lambda_2 - \epsilon \geq \zeta + \epsilon\). This adjustment increases the value of \(\zeta\) to \(\zeta + \epsilon\), which contradicts the assumption that \(\zeta\) was the maximum.
    
Therefore, the only configuration that maximizes \(\zeta\) is when all \(\lambda_i\)'s are equal. Thus, the maximum value of \(\min\{\lambda_1, \lambda_2, \dots, \lambda_L\}\) is attained when \(\lambda_1 = \lambda_2 = \dots = \lambda_L\).
\end{proof}

In general, the channel capacity of a sub-sequence over an AWGNC is determined by the transmit power allocated to it. Thus, we can adjust the channel capacity by assigning different powers to the subsequence. 
According to Theorem \ref{lemma.MaxMin_CRR}, the power allocation strategy under the max-min CRR criterion is to assign more power to sub-sequences with higher rates. In other words, sub-sequences with higher rates are allocated greater channel capacity (or more power), while those with lower rates receive less channel capacity (or power).

In fact, the proposed max-min CRR criterion closely resembles the well-known \textit{water-filling rule}. The key distinction lies in the fact that the ``water level'' is determined by the rate, while the ``water'' represents the \textit{channel capacity} rather than direct power. Therefore, we refer to this phenomenon in a multirate sequence as \textit{rate-driven capacity alignment}, and the corresponding power allocation scheme as \textit{capacity alignment (CA) power allocation}.

\begin{theorem} \label{theorem.CA}
  (\textbf{Rate-Driven Capacity Alignment Theorem})  
  To minimize the error probability, the channel capacity should be allocated according to the distribution of rates. Sequences with higher rates are allocated more capacity, while those with lower rates receive less. This allocation guarantees that the capacity-to-rate ratio (CRR) remains constant across all sequences. This principle is referred to as the rate-driven capacity alignment theorem.
\end{theorem}

Based on the rate-driven capacity alignment theorem, channel capacity can be allocated to various scenarios. Below, we outline three such scenarios:
\begin{itemize}
  \item {Single Multirate Sequence}: In this scenario, a single multirate sequence consists of multiple sub-sequences, each with different rates. However, the component digits of all sub-sequences are drawn from the same finite field.
  \item {Multiple Single-Rate Sequences}: Here, multiple single-rate sequences are considered, where the component digits of each sequence are drawn from different finite fields.
  \item {Multiple Multirate Sequences}: This scenario involves multiple multirate sequences, where the component digits of each sub-sequence are drawn from different finite fields.
\end{itemize}
In this paper, we focus on the first and second scenarios. The first scenario is investigated in this section, while the second scenario will be discussed in the following section.

\textbf{Example 9:}
We now examine the CA power allocation scheme for the single-user FF-TDMA systems, where the received FFSP blocks can be interpreted as a single multirate sequence from the same finite field. The codeword of the \(j\)-th user is treated as a multirate sequence, denoted as \({\bf c}_j \in {\mathbb P}^{1 \times m}\), where \({\bf c}_j = ({\bf d}_j, {\bf 0}, {\bf c}_{j,\rm red})\). Consequently, the average rate of the codeword is given by \(R_j = \frac{K}{m}\). The rates of \({\bf d}_j\) and \({\bf c}_{j,\rm red}\) are denoted as \(R_{j,1}\) and \(R_{j,2}\), respectively, with their specific values provided by equations (\ref{e.R_j1}) and (\ref{e.R_j2}). The channel capacities of \({\bf d}_j\) and \({\bf c}_{j,\rm red}\) are represented as \(C_{j,1}\) and \(C_{j,2}\), respectively, with their values given by equations (\ref{e.C_j_1}) and (\ref{e.C_j_2}). 
The optimization equations and constraints for the CA power allocation are as follows:
  \begin{equation}
    \begin{array}{ll}
      (\mu_1, \mu_2) = 
      \underset{C1}{\rm{argmax}} 
      \left\{ \min \left\{\lambda_{j,1}, \lambda_{j,2} \right\} \right\} = 
      \underset{C1}{\rm{argmax}}
      \left\{ \min \left\{\frac{C_{j,1}}{R_{j,1}}, \frac{C_{j,2}}{R_{j,2}} \right\} \right\}, \\
    \text{s.t.} \quad C1:  K \cdot \mu_1 + Q \cdot \mu_2 = m.
    \end{array}
  \end{equation}
  By solving this max-min optimization problem, the optimal polarization adjusted factors, $\mu_1$ and $\mu_2$, are determined.

It is noted that for the single-user FF-TDMA system, when \( K \geq Q \), the system operates as a single-rate sequence. In this case, the CA power allocation coincides with the equal power allocation (EPA) derived from the maximum channel capacity criterion. $\blacktriangle \blacktriangle$

\vspace{-0.1in}
\subsection{Comparison of Maximum Channel Capacity and Minimum Error Probability}

From the aforementioned discussion, we conclude that the criterion for minimizing error probability is equivalent to the maximum CRR criterion for a single-rate sequence, and the max-min CRR criterion for a multi-rate sequence. Therefore, the comparison between maximizing channel capacity and minimizing error probability is essentially the same as comparing maximizing channel capacity with the maximum CRR criterion or the max-min CRR criterion.

Next, we re-examine the channel capacity and explore the relationship between the maximum channel capacity and the maximum CRR criterion. We use the single-user PA-FF-TDMA system as an example, noting that the analysis for other systems follows a similar approach and will not be repeated.
The channel capacity \( C_{\rm SU}^{\rm td} \) of the single-user PA-FF-TDMA system is given by Eq. (\ref{e.su_tdma_capacity}). Consequently, the maximum channel capacity can be derived through the following steps:
\begin{equation} \label{e.capacity_CRR_su_TDMA}
  \begin{array}{ll} 
    \underset{\mu_{1}, \mu_{2}}{\max} {{C_{\rm SU}^{\rm td}}}
    &\overset{(a)}{=} \underset{\mu_{1}, \mu_{2}}{\max}  
    \left( \frac{ \frac{K}{2} \log_2 (1 + \mu_{1} \cdot \gamma_{a})
    + \frac{Q}{2} \log_2 (1 + \mu_{2} \cdot \gamma_{a})}{K\log_2 p} \right), \\
    &\overset{(b)}{=} \underset{\mu_{1}, \mu_{2}}{\max}  \left( 
    \frac{ \frac{1}{2} \log_2 (1 + \mu_{1} \cdot \gamma_{a})}{K/K} 
    + \frac{ \frac{1}{2} \log_2 (1 + \mu_{2} \cdot \gamma_{a})}{K/Q}  
    \right), \\
    &\overset{(c)}{=} \underset{\mu_{1}, \mu_{2}}{\max}  
    \left( \frac{ C_{j,1} }{{R}_{j,1}} 
    + \frac{ C_{j,2} }{\breve{R}_{j,2}}\right)
    \overset{(d)}{=} \underset{\mu_{1}, \mu_{2}}{\max}  
    \left( {\lambda}_{j, 1} + \breve{\lambda}_{j, 2} \right). \\
  \end{array} 
\end{equation}
In step (a) of Eq. (\ref{e.capacity_CRR_su_TDMA}), the constant \( K \), representing the number of information DoFs, is factored out. For a binary transmission system, we define \( {R}_{j,1} = 1 \) and \( \breve{R}_{j,2} = \frac{K}{Q} \), where \( C_{j,1} = \frac{1}{2} \log_2 (1 + \mu_{1} \cdot \gamma_{a}) \) and \( C_{j,2} = \frac{1}{2} \log_2 (1 + \mu_{2} \cdot \gamma_{a}) \). These definitions lead to step (c) of Eq. (\ref{e.capacity_CRR_su_TDMA}). Step (d) is then derived by introducing \( {\lambda}_{j,1} = \frac{ C_{j,1} }{R_{j,1}} \) and \( \breve{\lambda}_{j,2} = \frac{ C_{j,2} }{\breve{R}_{j,2}} \).

From this, maximizing the channel capacity can be expressed as the following relationship:
\begin{equation}
  \underset{\mu_{1}, \mu_{2}}{\max}  
  \left( {\lambda}_{j, 1} + \breve{\lambda}_{j, 2} \right)
  \stackrel{?}{=}
  \underset{\mu_{1}, \mu_{2}}{\max} \quad \min
  \left\{ {\lambda}_{j, 1}, \breve{\lambda}_{j, 2} \right\}
  \stackrel{?}{=}
  \underset{\mu_{1}, \mu_{2}}{\max} \quad \min
  \left\{{\lambda}_{j, 1}, {\lambda}_{j, 2} \right\},
\end{equation}
where the first \(\stackrel{?}{=}\) holds true only if \( R_{j,1} = \breve{R}_{j,2} \), and the second \(\stackrel{?}{=}\) holds true only if \( K \le Q \). The first condition indicates that the multirate sequence reduces to a single-rate sequence, while the latter condition implies \( \breve{\lambda}_{j,2} = {\lambda}_{j,2} \). This is because \( R_{j, 2} = \frac{\min\{K, Q\} }{Q} \), whose numerator is \( \min\{K, Q\} \), whereas the numerator of \( \breve{R}_{j,2} \) is \( K \).

\begin{figure}[t]
  \centering
  \subfigure[PAS of single-user FF-TDMA.]{\includegraphics[width=0.43\textwidth]{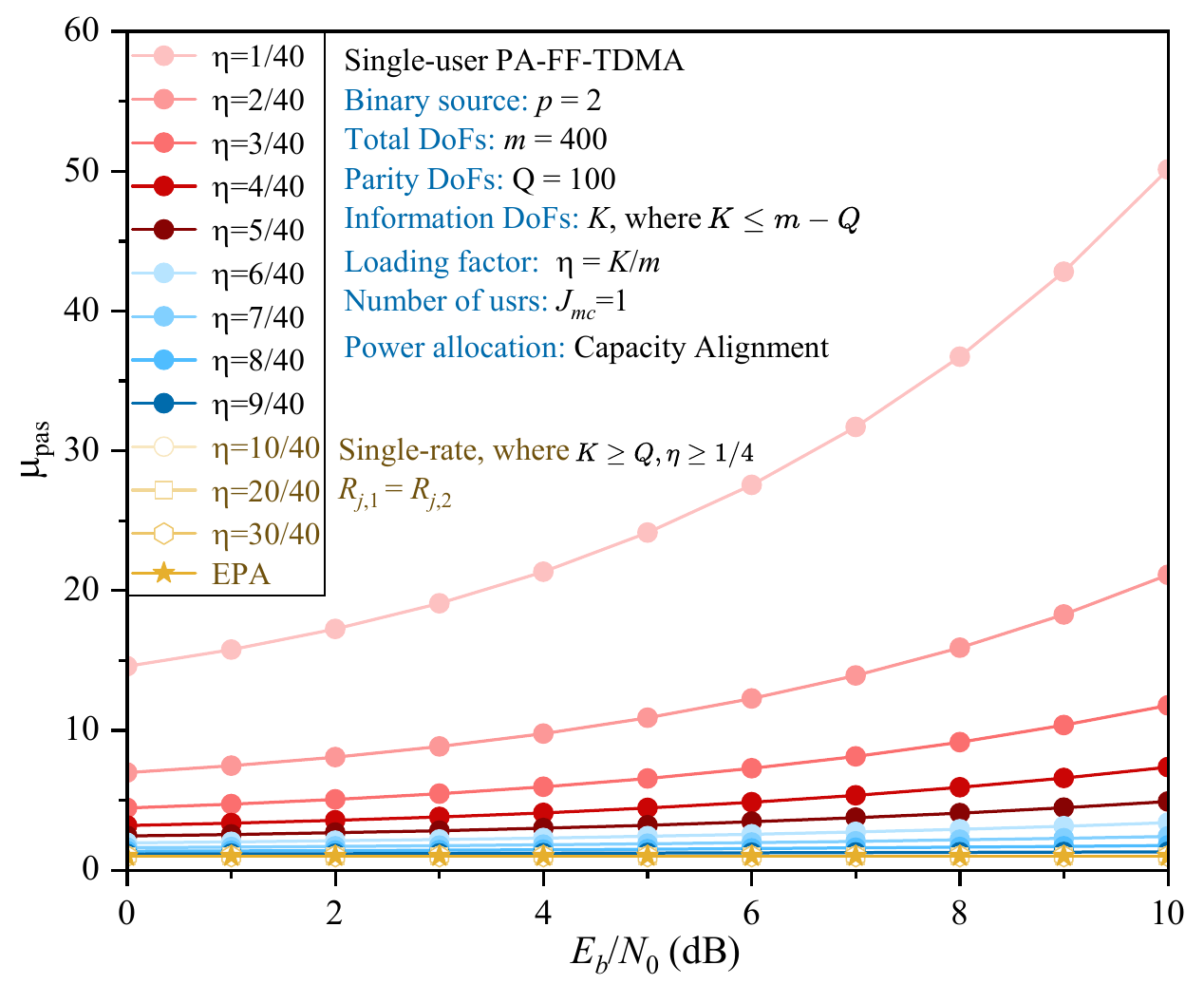}}
  \label{TDMA_mu_sub1}
  \subfigure[PAS of multiuser FF-TDMA.]{\includegraphics[width=0.43\textwidth]{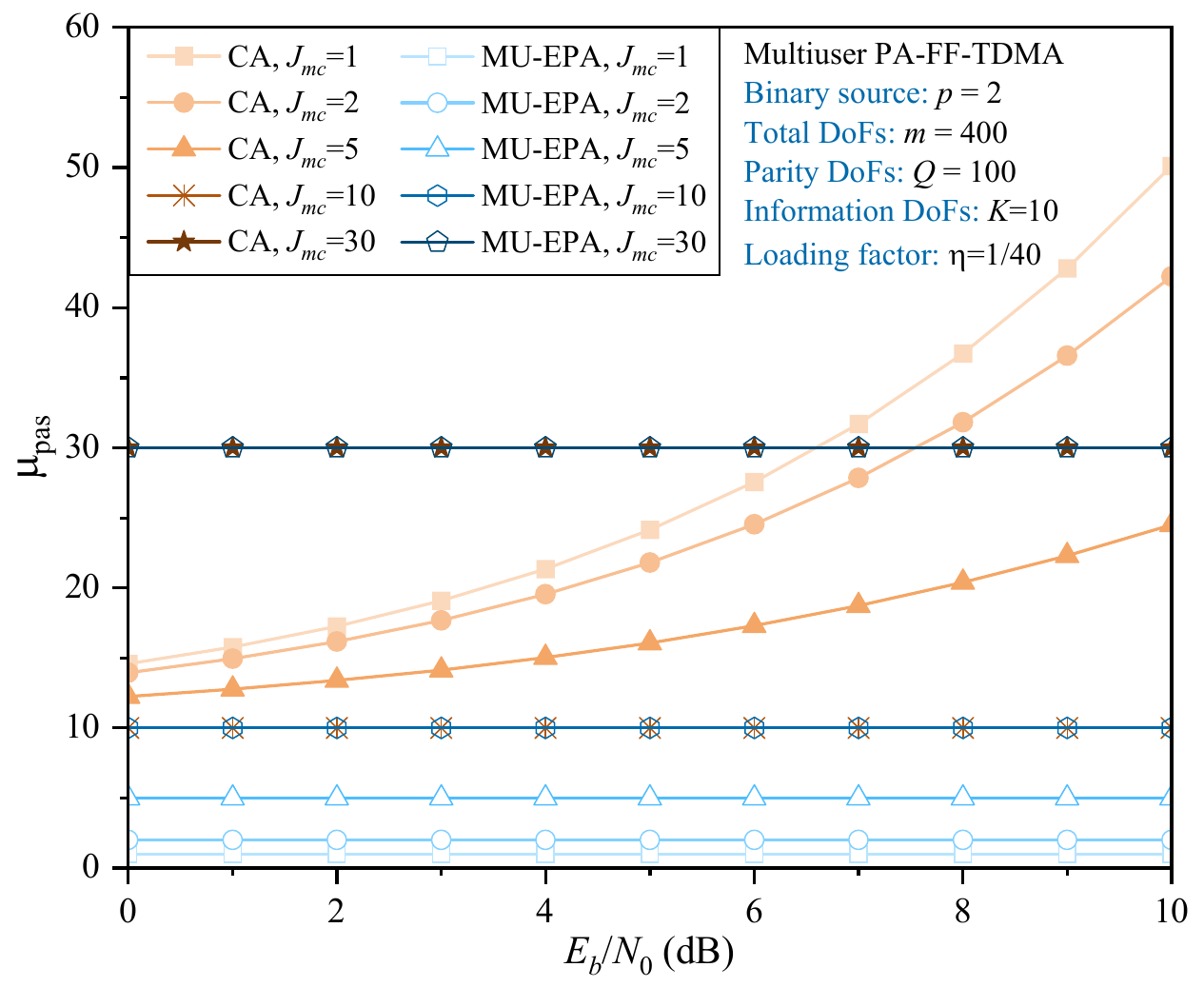}}
  \label{TDMA_mu_sub2}
    \caption{The PAS $\mu_{\rm pas}$ of FF-TDMA systems, where a $(400, 300)$ LDPC code is used for error control. 
    (a) PAS of single-user FF-TDMA. 
    (b) PAS of multiuser FF-TDMA, where $\eta = \frac{1}{40}$, and $J_{mc} = 1, 2, 5, 10, 30$.}
  \label{f.TDMA_mu}
  \vspace{-0.2in}
\end{figure}



\subsection{Monte Carlo Simulations}
We employ Monte Carlo simulations to analyze the polarization-adjusted scaling factor \(\mu_{\rm pas}\) of the FF-TDMA systems. The polarization-adjusted scaling factor \(\mu_{\rm pas}\) is chosen as the performance metric since it directly reflects the power allocation ratio between the information section and the parity section. The simulated PAS of the FF-TDMA systems is shown in Fig.~\ref{f.TDMA_mu}. In the simulation, we assume a total of \(m = 400\) DoFs, with \(Q = 100\) DoFs allocated to parity. The loading factors are used as the key comparison parameters.

The single-user FF-TDMA system is depicted in Fig.~\ref{f.TDMA_mu} (a). When \(K < Q\), the loading factors are \(\eta = \frac{1}{40}, \frac{2}{40}, \frac{3}{40}, \frac{4}{40}, \frac{5}{40}, \frac{6}{40}, \frac{7}{40}, \frac{8}{40}, \frac{9}{40}\), as shown in Fig.~\ref{f.TDMA_mu} (a). In this case, for a given loading factor \(\eta\), the PAS \(\mu_{\rm pas}\) increases exponentially with \(E_b/N_0\). This is because the channel capacity grows logarithmically with \(E_b/N_0\). However, the channel alignment criterion requires the channel capacity to be balanced linearly. As a result, to maintain a constant ``water level" (i.e., channel capacity), the allocated power must increase exponentially. 
Moreover, for a given \(E_b/N_0\) and \(K < Q\), sequences with smaller loading factors exhibit larger PAS \(\mu_{\rm pas}\), indicating that more power is allocated to the information section. This occurs because, when the loading factor is small, the rate of the parity section is \(K/Q\), which is much smaller than the information section rate (equal to \(1\)). Therefore, more power is allocated to the information section to align the channel capacity.
When \(K \geq Q\), the loading factors are \(\eta = \frac{1}{4}, \frac{2}{4}, \frac{3}{4}\), as shown in Fig.~\ref{f.TDMA_mu} (a), and the system behaves as a single-rate sequence. In this scenario, the PAS \(\mu_{\rm pas}\) of the CA power allocation remains constant and equals \(1\), which coincides with the equal power allocation (EPA).

The multiuser FF-TDMA system is depicted in Fig.~\ref{f.TDMA_mu} (b), where the loading factor is \(\eta = 1/40\), and the number of users is \(J_{mc} = 1, 2, 5, 10, 30\). When \(K \cdot J_{mc} < Q\), the number of users is \(J_{mc} = 1, 2, 5\), as shown in Fig.~\ref{f.TDMA_mu} (b). In this case, for a given \(E_b/N_0\), systems with fewer users exhibit a larger PAS \(\mu_{\rm pas}\). This is because the rate of the parity section, given by \(K \cdot J_{mc} / Q\), increases as the number of users grows. As a result, more power is allocated to the information section to maintain channel capacity alignment when there are fewer users.
When \(K \cdot J_{mc} \geq Q\), the number of users is \(J_{mc} = 10, 30\) in Fig.~\ref{f.TDMA_mu} (b), and the superimposed parity section matches the rate of the information section. In this scenario, the PAS \(\mu_{\rm pas}\) of the CA power allocation aligns with the multiuser equal power allocation (MU-EPA), given by \(\mu_{\rm pas} = J_{mc}\). This indicates that \(\mu_{\rm pas}\) increases linearly with the number of users.

\begin{figure}[t] 
  \centering
  \includegraphics[width=0.45\textwidth]{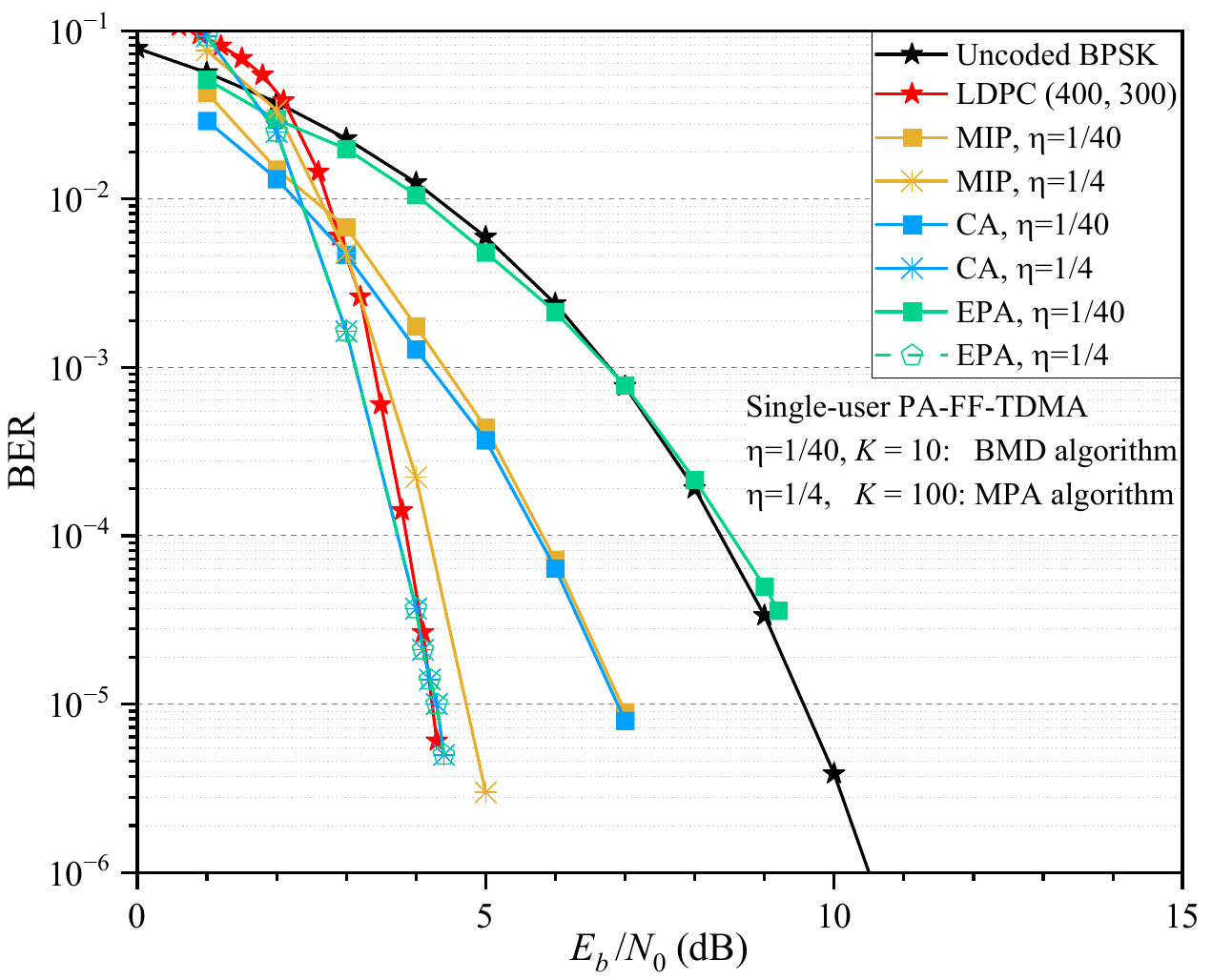}
  \caption{The BER performance of the PA-FF-TDMA system with various power allocation schemes, where the loading factors are \(\eta = 1/40\) and \(\eta = 1/4\).}
  \label{Fig_SU_TDMA_BER}
  \vspace{-0.3in}
\end{figure}



The BER performance of the single-user FF-TDMA system with various power allocation schemes is shown in Fig.~\ref{Fig_SU_TDMA_BER}, where the channel code is a binary \((400, 300)\) LDPC code. Note that, since we only consider a single-user scenario, the FF-TDMA system can also be viewed as a PA-LDPC code at this stage.
We compare three power allocation schemes: maximum information power allocation (MIP) \cite{FFMA2}, CA power allocation, and EPA.

When \(\eta = 1/40\) and \(K = 10\), the system operates as a multirate sequence, as previously mentioned, and the BMD algorithm is used for decoding. In this case, the BER performance of the system with the CA power allocation scheme outperforms that of the MIP and EPA schemes. Compared to the MIP scheme, the proposed CA scheme exhibits slightly better BER performance. However, the system with the EPA scheme demonstrates the worst error performance. The reason is that the BMD decoding algorithm relies on the information section as prior information. When the information section has the same power as the parity section, it cannot provide a more reliable detection range, resulting in poor decoding performance.

When \(\eta = 1/4\) and \(K = 100\), we have \(K = Q\), which enables the application of the message passing algorithm (MPA) as introduced in \cite{FFMA2}. Furthermore, when \(K = 100\), the system operates as a single-rate sequence. In this case, the CA power allocation scheme becomes equivalent to the EPA scheme. An interesting phenomenon is observed: the BER performance of the system using the CA and/or EPA schemes not only matches that of the original LDPC code but also provides slightly improved error performance. This suggests that, under the same Tanner graph, the proposed FF-TDMA (or PA-LDPC code) may offer even better error performance due to the introduced power allocation. It would be intriguing to further explore the impact of the polarization-adjusted feature, or power allocation, on classical LDPC codes. Power allocation could introduce an additional dimension to LDPC code design, potentially leading to novel and innovative design methodologies.

In this paper, we focus solely on the BER performance for the single-user scenario with the CA power allocation scheme, as the multiuser scenario with this scheme does not perform as well. The BER performance depends not only on power allocation but also on various decoding algorithms. Additionally, the CA power allocation is based on the channel capacity expression. However, since channel capacity is an FBL issue, this paper only provides a preliminary analysis. Therefore, the CA power allocation scheme does not achieve the desired error performance in the multiuser case. In fact, due to the limitations of the BMD decoding algorithm, the MIP scheme has proven to be more effective in the multiuser scenario.

\section{From Binary to $p$-ary Source Transmission}

In this section, we conduct a systematic investigation of $p$-ary source transmission and its performance characteristics. Building upon the rate-driven CA theorem established in Section~X, we analyze the spectral efficiency and power allocation requirements specific to $p$-ary transmission. Then, we present Monte Carlo simulation results that validate our theoretical predictions. 

\subsection{Capacity Alignment Theorem for Evaluating $p$-ary Transmission}
We use the binary system as a reference to evaluate the performance of a $p$-ary transmission system. Generally, the performance metrics of a communication system include error performance, transmit power, and spectral efficiency. As discussed earlier, the CRR can be used to measure the error performance of a system. Additionally, the CRR also reflects the effects of transmit power and spectral efficiency. Based on this, we present the following lemma.

\begin{lemma} \label{lemam.CA_pary}
  \textbf{(Capacity Alignment Theorem for Evaluating $p$-ary Transmission)}
  For a given loading factor \(\eta\), the coding rates of the $p$-ary and binary transmission systems are given by \(R_{p} = \eta \cdot \log_2 p\) and \(R_{b} = \eta\), respectively. Consequently, the spectral efficiency gain of the $p$-ary system over the binary system is \(\log_2 p\). For a given bit-energy-to-noise power spectral density ratio \(E_b/N_0\), the transmit signal-to-noise ratios (SNRs) of the $p$-ary and binary transmission systems are defined as
  \[
  \gamma_p = 2\mu_p \cdot R_p \frac{E_b}{N_0} \quad \text{and} \quad \gamma_b = 2\mu_b \cdot R_b \frac{E_b}{N_0},
  \]
  where \(\mu_p\) and \(\mu_b\) are the transmit power scaling factors of the $p$-ary and binary transmission systems, respectively. According to the capacity alignment theorem, we set the capacity-to-rate ratios (CRRs) of the $p$-ary and binary systems to be equal, i.e., \(\lambda_p = \lambda_b\), which yields
  \begin{equation} \label{e.CA_pary}
    \frac{0.5 \cdot \log_2 (1 + \gamma_p)}{R_{p}} = 
    \frac{0.5 \cdot \log_2 (1 + \gamma_b)}{R_{b}}.
  \end{equation}
  By solving \(\lambda_p = \lambda_b\), we obtain the polarization-adjusted scaling (PAS) factor \(\mu_{\rm pas} = \frac{\mu_p}{\mu_b}\). This factor serves as a performance metric for the $p$-ary system, as it quantifies the relative transmit power efficiency of the $p$-ary system compared to the binary system.
\end{lemma}

\noindent
Based on the relationship between the PAS \(\mu_{\rm pas}\) and the spectral efficiency gain \(\log_2 p\), we draw the following conclusions:
\begin{enumerate}
  \item If \(\mu_{\rm pas} < \log_2 p\), the \(p\)-ary system requires less transmit power while achieving a higher spectral efficiency. In this case, the \(p\)-ary system outperforms the binary system in terms of both power efficiency and spectral efficiency.
  
  \item If \(\mu_{\rm pas} = \log_2 p\), the \(p\)-ary system requires proportionally more transmit power to achieve the same spectral efficiency gain. Here, the \(p\)-ary system performs similarly to the binary system in terms of error performance.
  
  \item If \(\mu_{\rm pas} > \log_2 p\), the \(p\)-ary system requires significantly more transmit power to achieve the same spectral efficiency gain. In this case, the \(p\)-ary system underperforms compared to the binary system.
\end{enumerate}

Building upon Lemma~\ref{lemam.CA_pary} which establishes the capacity alignment for single-user $p$-ary transmission systems, we now generalize these results to $p$-ary multiuser communications. The $p$-ary architecture exhibits particularly favorable scaling properties in multiuser environments due to its inherent finite-field structure. By extending the single-user framework through multiuser rate adaptation, we obtain the generalized capacity alignment condition:
\begin{equation} \label{e.CA_pary_MU}
  \frac{0.5 \cdot \log_2 (1 + J \cdot \gamma_p)}{J \cdot R_{p}} = 
  \frac{0.5 \cdot \log_2 (1 + J \cdot \gamma_b)}{J \cdot R_{b}},
\end{equation}
where $J$ represents the number of active users in the system. Solving \eqref{e.CA_pary_MU} yields the polarization-adjusted scaling factor $\mu_{\mathrm{pas}} = \mu_p/\mu_b$, which serves as a metric for comparing the power efficiency between $p$-ary and binary multiuser systems. This factor quantitatively characterizes the relative transmit power advantage of the $p$-ary implementation while maintaining equivalent rate performance across all users.



\begin{figure}[t]
  \centering
  \subfigure[The PAS $\mu_{\rm pas}$ of a $3$-ary system.]{\includegraphics[width=0.43\textwidth]{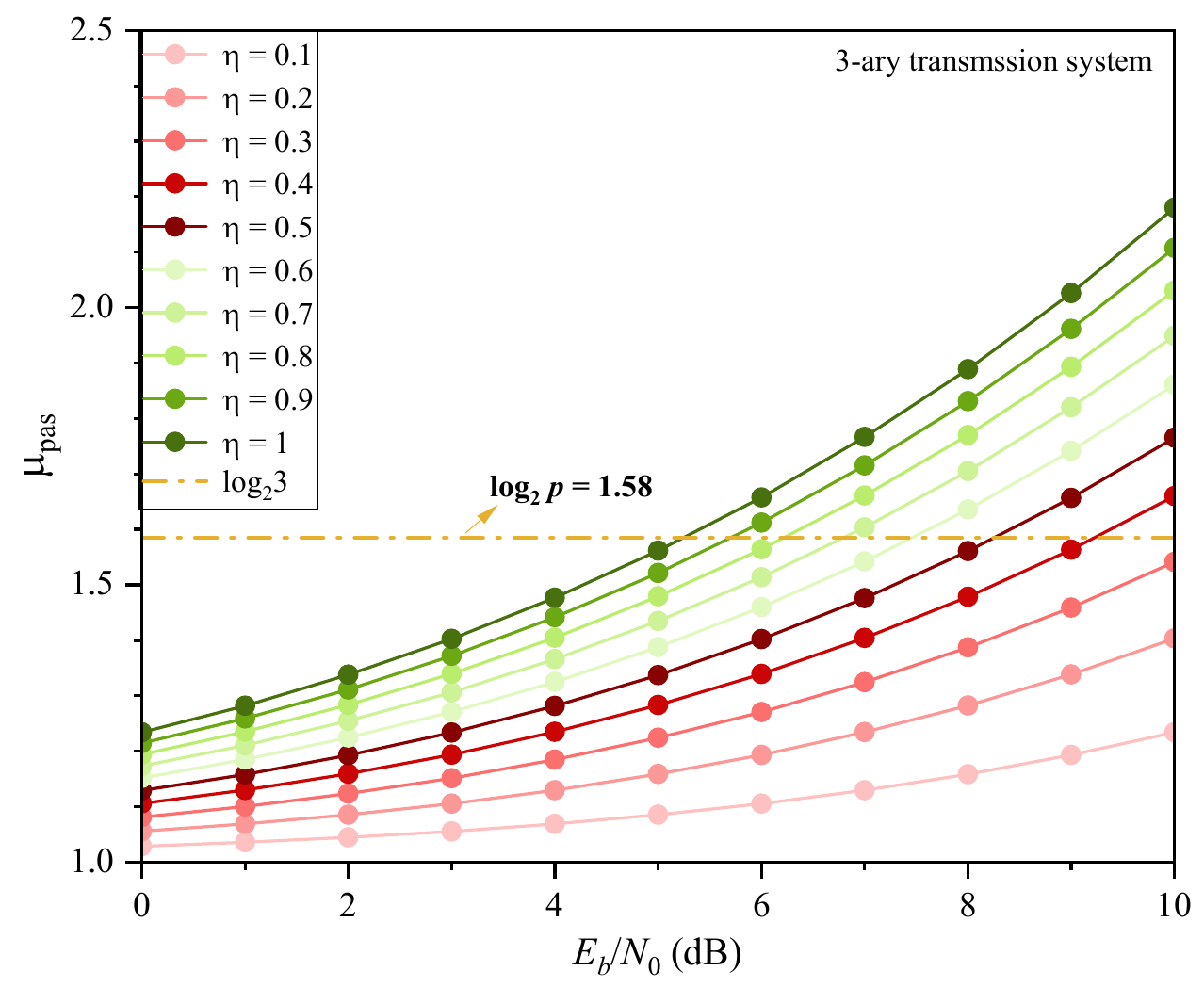}}
  \label{pary_mu_sub1}
  \subfigure[The PAS $\mu_{\rm pas}$ of a $p$-ary system.]{\includegraphics[width=0.43\textwidth]{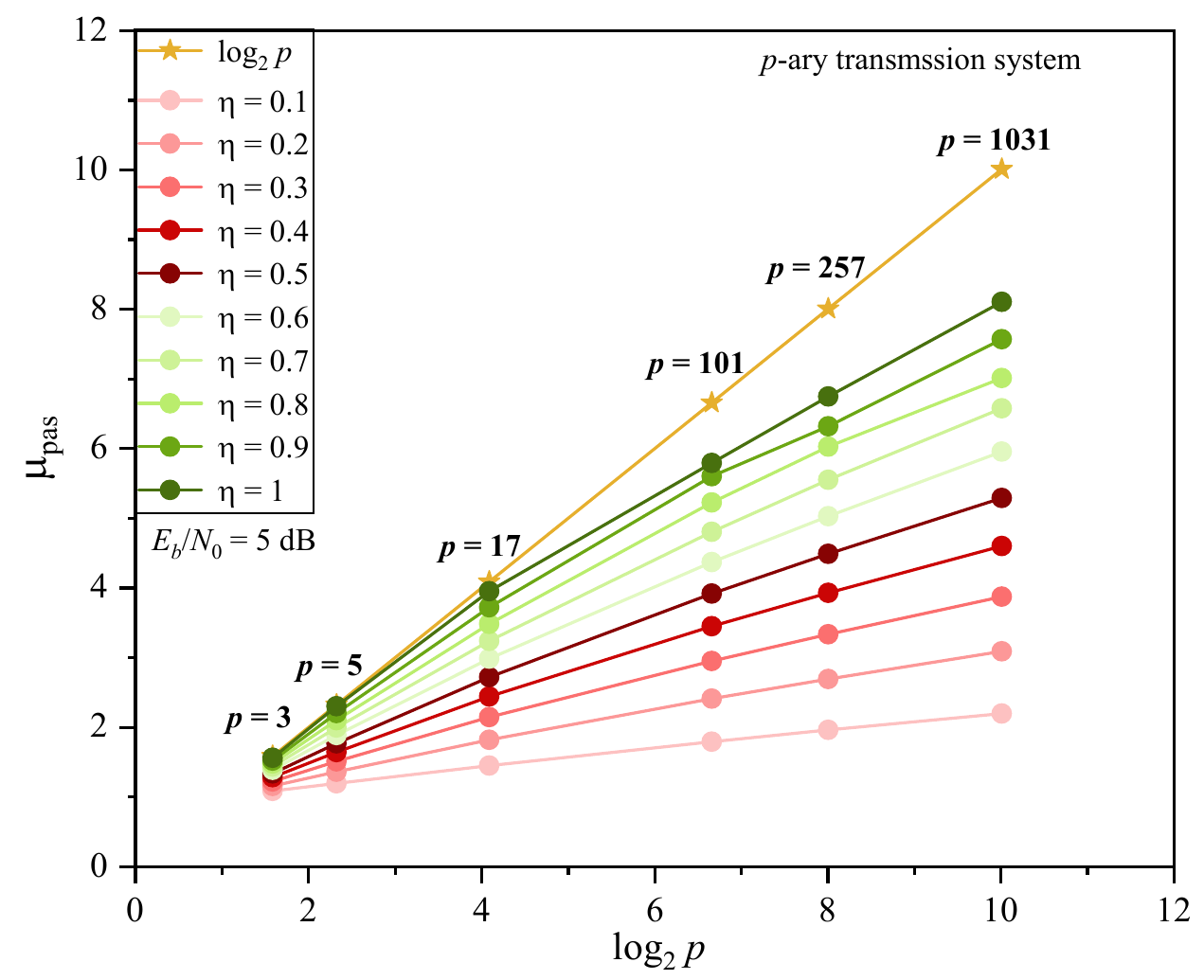}}
  \label{pary_mu_sub2}
  \caption{Polarization adjusted scaling factor $\mu_{\mathrm{pas}}$ for various $p$-ary systems under different loading factors.}
  \label{f.pary_mu}
  \vspace{-0.3in}
\end{figure}

\subsection{Monte Carlo Simulations}
We employ Monte Carlo simulations to analyze the polarization-adjusted scaling factor $\mu_{\mathrm{pas}}$ for $p$-ary transmission systems, with numerical results presented in Figs.~\ref{f.pary_mu} and \ref{f.MU_pary}. 

In Fig.~\ref{f.pary_mu} (a), the PAS of a \(3\)-ary system is shown. It can be observed that, for a given loading factor \(\eta\), the PAS \(\mu_{\rm pas}\) increases with increasing \(E_b/N_0\). Additionally, for a fixed \(E_b/N_0\), the PAS \(\mu_{\rm pas}\) increases as \(\eta\) increases. This implies that, when the loading factor is small, the \(3\)-ary system requires less power than the binary system to achieve the same error performance. 
When \(E_b/N_0\) is less than \(5\) dB, we have \(\mu_{\rm pas} < \log_2 3\) for all loading factors in the \(3\)-ary system. This indicates that the \(3\)-ary system outperforms the binary system in this regime. As \(E_b/N_0\) increases, the \(3\)-ary systems with smaller loading factors (\(\eta < 0.3\)) continue to exhibit better performance than the binary system. However, for larger loading factors, e.g., \(\eta > 0.3\), the \(3\)-ary system performs worse than the binary system.
The PAS of different \(p\)-ary systems is shown in Fig.~\ref{f.pary_mu} (b), where the X-axis represents the spectral efficiency gain, defined as \(\log_2 p\), for \(p = 3, 5, 17, 101, 257, 1031\) at \(E_b/N_0 = 5\) dB. From this figure, it can be observed that, for a given loading factor \(\eta\), the \(p\)-ary system with a larger \(p\) value significantly outperforms the binary system, even for large values of \(\eta\). Furthermore, as \(p\) increases, the required PAS \(\mu_{\rm pas}\) grows logarithmically. This suggests that higher-order \(p\)-ary systems can provide even better performance compared to the binary system, which serves as the reference. 

These results demonstrate that, for small loading factors, a low-order \( p \)-ary system (e.g., \( p = 3 \)) can outperform the binary system. For larger loading factors, higher-order \( p \)-ary systems (e.g., \( p = 257 \) or \( p = 1031 \)) may still achieve superior performance compared to the binary system. Therefore, it is interesting to further investigate and develop \( p \)-ary systems. However, considering the complexity involved in designing a \( p \)-ary system, we can decompose a \( p \)-ary system into a ternary system for practical applications, as discussed in Sect. VI.



\begin{figure}[t]
  \centering
  \subfigure[$\mu_{\mathrm{pas}}$ for $p$-ary systems.]{\includegraphics[width=0.44\textwidth]{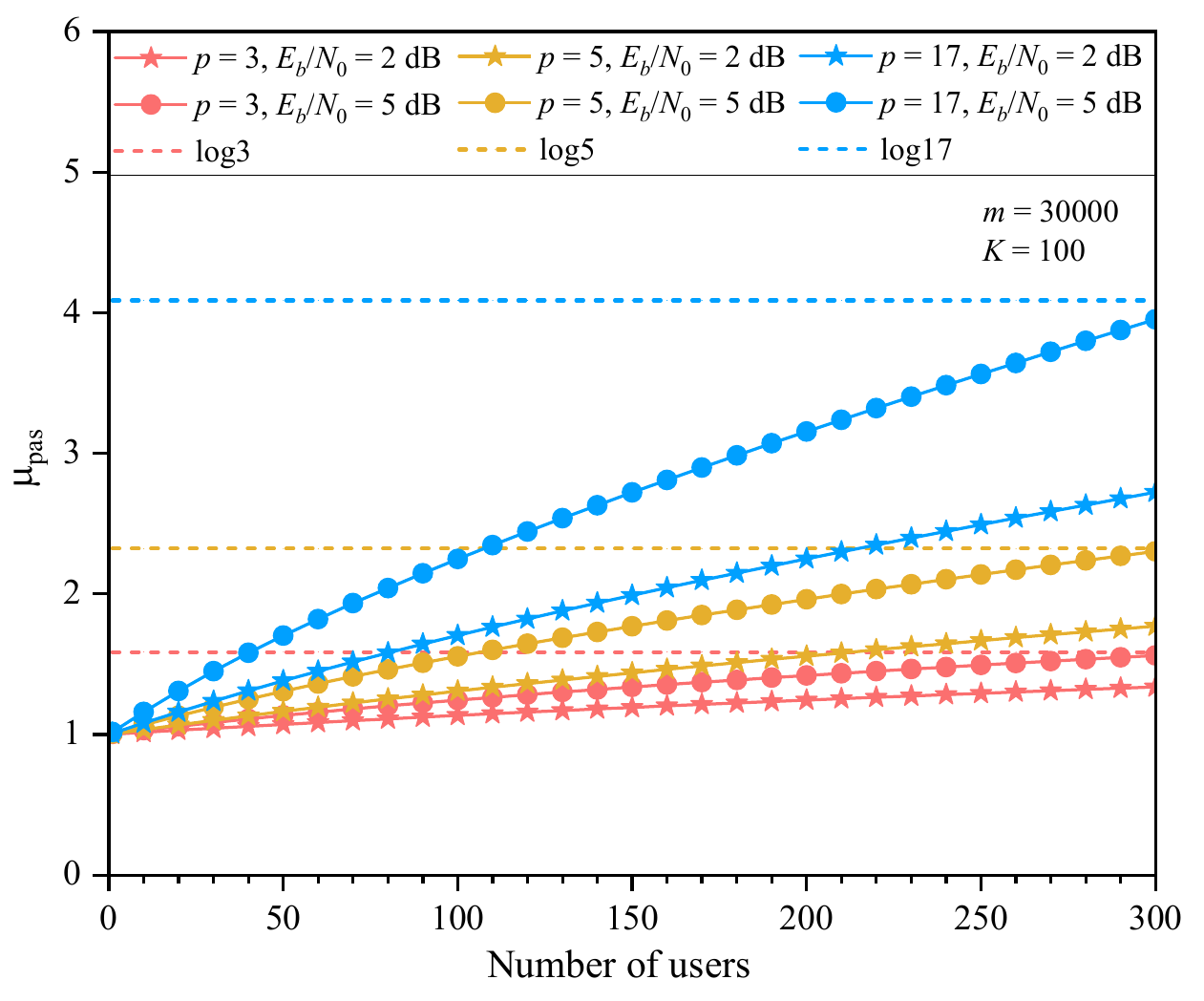}}
  \label{Fig_Graph11}
  \subfigure[BER between binary and ternary systems.]{\includegraphics[width=0.445\textwidth]{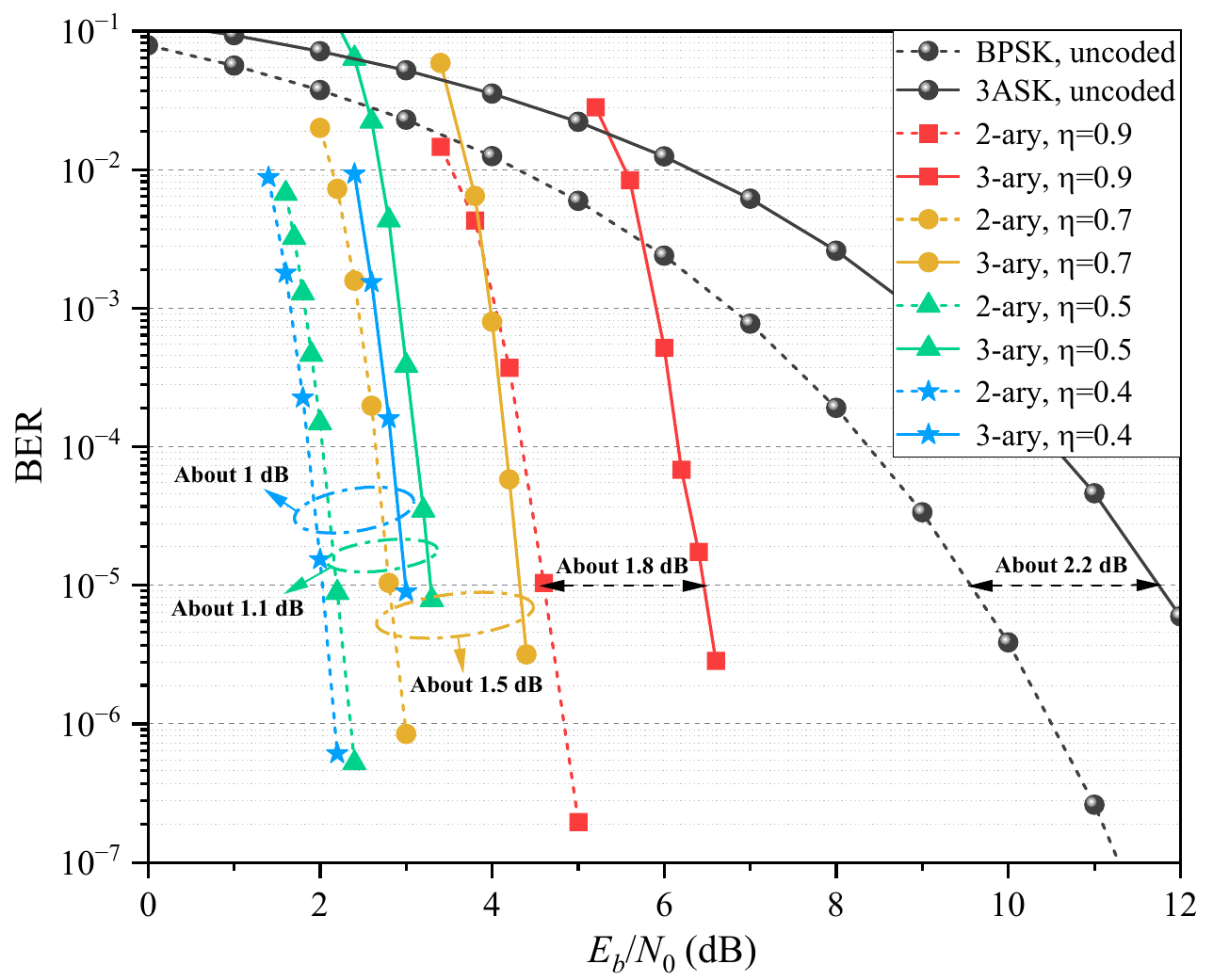}}
  \label{Fig_Graph10}
  \caption{Performance analyse between $p$-ary and binary system.}
  \label{f.MU_pary}
  \vspace{-0.3in}
\end{figure}

We now analyze the PAS factor $\mu_{\mathrm{pas}}$ for $p$-ary multiuser systems, investigating its behavior under different number of users. The system parameters are configured with $p = 3, 5, 17$, fixed information DoFs  $K = 100$, total DoFs $m = 30000$, and $E_b/N_0 = 2, 5$ dB, as shown in Fig.~\ref{f.MU_pary} (a).

The $p$-ary multiuser system exhibits two distinctive characteristics. First, for any given $E_b/N_0$ and prime $p$, $\mu_{\mathrm{pas}}$ demonstrates logarithmic growth with increasing user count $J$. For example, in the ternary ($p=3$) case at $E_b/N_0 = 5$ dB, $\mu_{\mathrm{pas}}$ evolves from $1.003$ ($J=1$) to $1.135$ ($J=50$) and finally to $1.561$ ($J=300$), confirming that the $p$-ary architecture maintains power efficiency even with substantial user loads. 
Second, the $p$-ary advantage becomes particularly pronounced for systems with both higher order $p$ and moderate user counts. Comparative analysis reveals that when $J=10$ at $E_b/N_0 = 5$ dB, the ternary ($p=3$) system achieves a power gain of $\log_2(3) - 1.0299 \approx 0.585$, while the higher-order system ($p=17$) delivers a more substantial gain of $\log_2(17) - 1.1596 \approx 2.887$. This scaling behavior highlights a key advantage of $p$-ary multiuser systems: they provide flexible performance trade-offs between user capacity and power efficiency through appropriate selection of the prime parameter $p$.

These results collectively demonstrate that $p$-ary systems successfully extend their single-user advantages to multiuser scenarios, offering system designers a powerful framework for balancing spectral efficiency against power requirements in next-generation communication systems.

To validate the theoretical analysis of the $3$-ary system, we perform numerical simulations using four LDPC codes: a $(2000, 1800)$ LDPC code with loading factor $\eta = 0.9$, a $(2000, 1400)$ LDPC code with $\eta = 0.7$, a $(2000, 1000)$ LDPC code with $\eta = 0.5$ for binary system and a $(2000, 1001)$ LDPC code with $\eta \approx 0.5$ for ternary system, and a $(2000, 800)$ LDPC code with $\eta = 0.4$. The BER performance is evaluated for both binary (BPSK modulation) and ternary (3ASK modulation) transmission schemes, including uncoded cases as baseline references.

Fig.~\ref{f.MU_pary} (b) presents a detailed comparison of the error performance between binary and ternary transmission systems. At a BER of \(10^{-5}\), the ternary system requires approximately $1.0$ dB (or \(1.2589 \times\)) higher \(E_b/N_0\) than the binary system at \(\eta = 0.4\). This difference increases to $1.1$ dB (\(1.2882 \times\)) at \(\eta = 0.5\), $1.5$ dB (\(1.4125 \times\)) at \(\eta = 0.7\), and $1.8$ dB (\(1.5136 \times\)) at \(\eta = 0.9\). The progressive increase in the required \(E_b/N_0\) with higher loading factors demonstrates the energy efficiency trade-off inherent in ternary modulation. For the uncoded case (\(\eta = 1\)), the difference becomes more pronounced, reaching $2.2$ dB (\(1.6596 \times\)).

The coded results show that all power ratios (ranging from \(1.2589\) to \(1.5136\)) fall below the theoretical spectral efficiency gain of \(\log_2 3 \approx 1.585\), confirming the superiority of the ternary system. Notably, the performance advantage diminishes as the loading factor increases, with the uncoded case exceeding the spectral efficiency limit. This behavior aligns perfectly with our theoretical predictions, demonstrating that ternary systems offer greater benefits for lower loading factors, where coding gain plays a more significant role.

\vspace{-0.1in}
\section{Conclusions}
In this paper, we introduce EA codes over \( \text{GF}(p^m) \) to support \( p \)-ary source transmission. A general structural constraint, the USPM property, is proposed for designing UD-EA codes. Two specific types of codes are constructed: AI-D-CWEA and BD-D-CWEA codes. We demonstrate that linear block codes can be employed to construct UD-CWEA codes. Furthermore, we extend EP-coding to EA-coding, with a focus on NO-CWEA codes and their USPM constraint in the complex field. We then proceed to construct \( p \)-ary CWEA codes using a basis decomposition method.

In addition to developing EA coding for FFMA systems, we perform a comprehensive performance analysis, addressing both channel capacity and error performance. First, we prove that EPA is optimal for achieving maximum channel capacity and analyze its FBL characteristics. We then introduce CRR as a metric for evaluating error performance, leading to the derivation of the rate-driven CA theorem. Using the CA theorem, we propose a CA power allocation scheme and conduct a systematic comparison between \( p \)-ary and conventional binary transmission systems.

In conclusion, drawing inspiration from the ancient Chinese philosophy in the famous book \textit{Daodejing}: ``The Dao springs from the One; from the One unfolds the Two; from the Two emerges the Three; and from the Three, the myriad creatures of the world are born.'' This work follows a similar progression, extending binary systems to \( p \)-ary systems and unlocking new possibilities for future communication technologies.

\vfill

\vspace{-0.1in}
\begin{thebibliography}{99}

\bibitem{6G}
C. -X. Wang et al., ``On the Road to 6G: Visions, Requirements, Key Technologies, and Testbeds,'' \textit{IEEE Communications Surveys \& Tutorials}, vol. 25, no. 2, pp. 905-974, 2023.


\bibitem{Polar_1}
J. Dai, K. Niu, et al., ``Polar-Coded Non-Orthogonal Multiple Access,'' \textit{IEEE Trans. Signal Processing}, vol. 66, no. 5, 2018, pp. 1374-1389. 

\bibitem{Polar_2}
E. Abbe and E. Telatar, ``Polar Codes for The $m$-user Multiple Access Channel,'' \textit{IEEE Trans.  Information Theory}, vol. 58, no. 8, 2012. 


\bibitem{Polar_3}
M. Zheng, Y. Wu, et al., ``Polar Coding and Sparse Spreading for Massive Unsourced Random Access,'' \textit{IEEE Vehicular Technology Conference, VTC 2020}, Nov. 2020, pp. 8-13.

\bibitem{Polar_4}
A. K. Pradhan, V. K. Amalladinne, K. R. Narayanan and J. -F. Chamberland, ``Polar Coding and Random Spreading for Unsourced Multiple Access," 2020 IEEE International Conference on Communications (ICC), Dublin, Ireland, 2020, pp. 1-6.


\bibitem{IDMA_1} 
P. Li, L. Liu, K. Wu and W. K. Leung, ``Interleave division multiple-access,'' \textit{IEEE Transactions on Wireless Communications}, vol. 5, no. 4, pp. 938-947, April 2006.

\bibitem{IDMA_2}
A. K. Pradhan, V. K. Amalladinne et al., ``Sparse IDMA: A Joint
Graph-Based Coding Scheme for Unsourced Random Access,'' \textit{IEEE
Trans. Commun.}, vol. 70, no. 11, pp. 7124-7133, 2022.


\bibitem{FBL_1}
V. Y. F. Tan and M. Tomamichel, ``The third-order term in the normal approximation for the AWGN channel,'' \textit{IEEE Trans. Inf. Theory}, vol. 61, no. 5, pp. 2430-2438, May 2015.

\bibitem{FBL_2}
Y. Polyanskiy, ``A Perspective on Massive Random-Access,'' \textit{in Proc.IEEE International Symposium on Information Theory (ISIT)}, pp. 2523-2527, 2017.

\bibitem{FBL_3}
Y. Polyanskiy, H. V. Poor and S. Verdú, ``Channel coding rate in the finite blocklength regime,'' \textit{IEEE Trans. Inf. Theory}, vol. 56, no. 5, pp. 2307-2359, May 2010.

\bibitem{FBL_MU}
R. C. Yavas, V. Kostina and M. Effros, ``Gaussian Multiple and Random Access Channels: Finite-Blocklength Analysis,'' \textit{IEEE Trans. Inf. Theory}, vol. 67, no. 11, pp. 6983-7009, Nov. 2021.


\bibitem{FFMA}
Qi-yue Yu, Jiang-xuan Li, and Shu Lin, ``Finite Field Multiple Access,'' https://arxiv.org/abs/2303.14086.


\bibitem{FFMA_ITW}
Qi-yue Yu, Shi-wen Lin, and Shu Lin, ``Finite Field Multiple Access for Sourced Massive Random Access with Finite Blocklength,'' 2024 IEEE Information Theory Workshop (ITW), Shenzhen, China, 2024, pp. 741-746. 


\bibitem{FFMA2}
Qi-yue Yu, Shi-wen Lin, and Ting-wei Yang "Finite Field Multiple Access II: from Symbol-wise to Codeword-wise," https://arxiv.org/abs/2503.09991.




\bibitem{History_Binary_1990}
W. Aspray ``John von Neumann and the Origins of Modern Computing (History of Computing),'' Boston, Cambridge: MIT, 1990.

\bibitem{Ternary_Computing_1990}
J. Connely, ``Ternary Computing Testbed 3-Trit Computer Architecture,'' California Polytechnic State University of San Luis Obispo, August 29th, 2008.

\bibitem{Ternary_Computing}
D. Roy, and Jr. Merril ``Ternary Logic in Digital Computers,'' Proceedings of the SHARE design automation project (DAC '65), ACM New York, NY, USA, pp. 6.1-6.17.

\bibitem{TernaryCode_2017}
S. Ferdowsi, S. Voloshynovskiy, D. Kostadinov and T. Holotyak, ``Sparse ternary codes for similarity search have higher coding gain than dense binary codes,'' \textit{in Proc. IEEE International Symposium on Information Theory (ISIT)}, Aachen, Germany, 2017, pp. 2653-2657.



\bibitem{Ternary_CC_2020}
D. Efanov, ``Classification of Errors in Ternary Code Vectors from the Standpoint of Their Use in the Synthesis of Self-Checking Digital Systems,'' 2020 IEEE East-West Design \& Test Symposium (EWDTS), Varna, Bulgaria, 2020, pp. 1-7.


\bibitem{Ternary_cyclic_2013}
Ding, Y. Gao and Z. Zhou, ``Five Families of Three-Weight Ternary Cyclic Codes and Their Duals,''\textit{IEEE Trans. Inf. Theory}, vol. 59, no. 12, pp. 7940-7946, Dec. 2013.

\bibitem{TPSK_2016}
M. Abdelaziz and T. A. Gulliver, ``Ternary Convolutional Codes for Ternary Phase Shift Keying,'' \textit{IEEE Communications Letters}, vol. 20, no. 9, pp. 1709-1712, Sept. 2016.

\bibitem{TPSK_2022}
Hama and H. Ochiai,``Binary-Input Ternary-Output Turbo Codes for Ternary PSK Transmission,'' \textit{IEEE Communications Letters}, vol. 26, no. 9, pp. 1974-1978, Sept. 2022.


\bibitem{Ternary_Hamm_2013}
J. -J. Wang, H. -D. Chen, T. -Y. Yang, H. Chen and C. -Y. Lin, ``Data Hiding Technique by Ternary Hamming Codes,'' 2013 IEEE 37th Annual Computer Software and Applications Conference, Kyoto, Japan, 2013, pp. 163-164.


\bibitem{Ternary_ParityCode_2019}
D. V. Efanov, ``Ternary Parity Codes: Features,'' 2019 IEEE East-West Design \& Test Symposium (EWDTS), Batumi, Georgia, 2019, pp. 1-5.


\bibitem{Ternary_Sum_Code_2020}
D. Efanov, "Ternary Sum Codes," 2020 IEEE East-West Design \& Test Symposium (EWDTS), Varna, Bulgaria, 2020, pp. 1-8.



\bibitem{QAM9_2017}
Kuznetsov, V. Batura, A. Solodkov and A. Malyshev, ``Using QAM-9 and ternary noise-immune codes to approach the Shannon bound,'' 2017 IEEE Conference of Russian Young Researchers in Electrical and Electronic Engineering (EIConRus), St. Petersburg and Moscow, Russia, 2017, pp. 169-172.



\bibitem{Ternary_MPA_2022}
M. Zhu, M. Jiang and C. Zhao, ``Ternary Message Passing Decoding of RS-SPC Product Codes,'' \textit{in Proc. IEEE International Symposium on Information Theory (ISIT)}, Espoo, Finland, 2022, pp. 2916-2921.




\bibitem{UD_CDMA1_2012}
O. Mashayekhi and F. Marvasti, ``Uniquely Decodable Codes with Fast Decoder for Overloaded Synchronous CDMA Systems," \emph{IEEE Trans. Commun.}, vol. 60, no. 11, pp. 3145-3149, Nov. 2012.


\bibitem{UD_CDMA2_2012}
M. Kulhandjian and D. A. Pados, ``Uniquely decodable code-division via augmented Sylvester-Hadamard matrices," \emph{in Proc. IEEE Wireless Communications and Networking Conference (WCNC)}, 2012, pp. 359-363.


\bibitem{UD_CDMA3_2014}
A. Singh and P. Singh, ``Uniquely decodable codes for overloaded synchronous CDMA with two sets of orthogonal signatures," \emph{2014 Annual IEEE India Conference (INDICON)}, 2014, pp. 1-4.


\bibitem{UD_CDMA4_2016}
M. Li and Q. Liu, ``Fast Code Design for Overloaded Code-Division Multiplexing Systems," \emph{IEEE Trans. Vehicular Technology}, vol. 65, no. 1, pp. 447-452, Jan. 2016.


\bibitem{UD_CDMA5_2018}
M. Kulhandjian, C. D'Amours and H. Kulhandjian, ``Uniquely Decodable Ternary Codes for Synchronous CDMA Systems," \emph{in Proc. IEEE 29th Annual International Symposium on Personal, Indoor and Mobile Radio Communications (PIMRC)}, 2018, pp. 1-6.


\bibitem{UD_CDMA6_2019}
M. Kulhandjian, H. Kulhandjian, C. D'Amours, H. Yanikomeroglu and G. Khachatrian, ``Fast Decoder for Overloaded Uniquely Decodable Synchronous Optical CDMA," \emph{in Proc. IEEE Wireless Communications and Networking Conference (WCNC)}, 2019, pp. 1-7.





\bibitem{ZDing2017_survey}
Z. Ding, X. Lei, G. K. Karagiannidis, R. Schober, J. Yuan, and V. K. Bhargava, ``A Survey on Non-Orthogonal Multiple Access for 5G Networks: Research Challenges and Future Trends,'' \textit{IEEE Journal on Selected Areas in Communications}, vol. 35, no. 10, pp. 2181 - 2195, July 2017.


\bibitem{YChen_2018}
Y. Chen et al., ``Toward the Standardization of Non-Orthogonal Multiple Access for Next Generation Wireless Networks,'' \textit{IEEE Communications Magazine}, vol. 56, no. 3, pp. 19-27, March 2018.


\bibitem{UMA_2022}
Y. Li et al., ``Unsourced multiple access for 6G massive machine type communications,'' \textit{China Communications}, vol. 19, no. 3, pp. 70-87, March 2022.







































































\bibitem{Yu_UDAS}
Q. Yu, and K. Song, ``Uniquely Decodable Multi-Amplitude Sequence for Grant-Free Multiple-Access Adder Channels,'' \textit{IEEE Trans. Wireless communications}, 2023, Early Access.



\bibitem{Shu2009}
William E. Ryan and Shu Lin, Channel Codes classical and Modern, Cambridge University Press, 2009.

\bibitem{John2009}
John G. Proakis, Digital Communications, Fifth Edition, Beijing, Publishing House of Electronics Industry, 2009.

\bibitem{LinBook3}
J. Li, S. Lin, K. Abdel-Ghaffar, W. E. Ryan, and D. J. Costello, ``LDPC Code Designs, Constructions, and Unification," Cambridge University Press, 2017.


\bibitem{Thomas}
Thomas M. Cover, and Joy A. Thomas, ``Elements of Information Theory,'' Tsinghua University Press, 2010.



\end{thebibliography}
\end{document}